\documentclass[runningheads,a4paper]{llncs}

\usepackage{a4wide}

\usepackage[latin1]{inputenc}
\usepackage{graphicx}
\usepackage{amsfonts}
\usepackage{amssymb}
\usepackage{amsmath}
\usepackage{framed}
\usepackage{latexsym}
\usepackage{breakcites}
\usepackage{color}
\usepackage{url}
\usepackage{tikz}
\usepackage{algorithm}
\usepackage[noend]{algpseudocode}
\usepackage{enumerate}
\usepackage{stmaryrd}

\usepackage{thm-restate}

\newtheorem{fact}{Fact}

\makeatletter
\newcounter{restate}
\newenvironment{restate}[2]{
	\setcounter{restate}{\value{#1}}
	\setcounter{#1}{\@ifundefined{r@#2}{0}{\ref{#2}}}
	\addtocounter{#1}{-1}
	\def\thname{#1}
	\begin{\thname}
		\setcounter{#1}{\value{restate}}
	}{
	\end{\thname}
}
\makeatother

\newcommand{\cA}{\mathcal{A}}
\newcommand{\cC}{\mathcal{C}}
\newcommand{\cD}{\mathcal{D}}
\newcommand{\cE}{\mathcal{E}}

\newcommand{\cS}{\mathcal{S}}
\newcommand{\cU}{\mathcal{U}}

\newcommand{\cv}{{\mathbf{c}}}
\newcommand{\dv}{{\mathbf{d}}}
\newcommand{\ev}{{\mathbf{e}}}
\newcommand{\hv}{{\mathbf{h}}}
\newcommand{\mv}{{\mathbf{m}}}
\newcommand{\rv}{{\mathbf{r}}}
\newcommand{\sv}{{\mathbf{s}}}
\newcommand{\uv}{{\mathbf{u}}}
\newcommand{\vv}{{\mathbf{v}}}
\newcommand{\xv}{{\mathbf{x}}}
\newcommand{\yv}{{\mathbf{y}}}

\newcommand{\Am}{{\mathbf{A}}}
\newcommand{\Bm}{{\mathbf{B}}}
\newcommand{\Gm}{{\mathbf{G}}}
\newcommand{\Hm}{{\mathbf{H}}}
\newcommand{\Rm}{{\mathbf{R}}}
\newcommand{\Imat}{{\mathbf{I}}}

\newcommand{\Pm}{{\mathbf{P}}}
\newcommand{\Sm}{{\mathbf{S}}}
\newcommand{\Um}{{\mathbf{U}}}

\newcommand{\mat}[1]{\ensuremath{\boldsymbol{#1}}}

\newcommand{\zero}{{\mat{0}}}

\newcommand{\hash}{h}
\newcommand{\Cc}{{\mathcal C}}
\newcommand{\Dc}{{\mathcal D}}
\newcommand{\Ec}{{\mathcal E}}
\newcommand{\Fc}{{\mathcal F}}
\newcommand{\Hc}{{\mathcal H}}

\newcommand{\Lc}{{\mathcal L}}
\newcommand{\Pc}{{\mathcal P}}
\newcommand{\Qc}{{\mathcal Q}}

\newcommand{\Uc}{{\mathcal U}}

\newcommand{\Cpub}{\Cc_{\text{pub}}}
\newcommand{\Cr}{\Cc_{\text{rand}}}
\newcommand{\Hpub}{{\Hm_{\textup{pub}}}}
\newcommand{\Hsec}{{\Hm_{\textup{sec}}}}
\newcommand{\Csec}{\Cc_{\text{sec}}}

\newcommand{\Psucc}{{P_{\text{succ}}}}

\newcommand{\Lcts}{S}
\newcommand{\Lcth}{H_\textup{true}}
\newcommand{\Lcf}{H_\textup{false}}

\newcommand{\F}{\mathbb{F}}

\newcommand{\R}{\mathbb{R}}
\newcommand{\Unif}{\hookleftarrow}

\newcommand{\prob}{\mathbb{P}}
\newcommand{\esp}{\mathbb{E}}
\newcommand{\eqdef}{\mathop{=}\limits^{\triangle}}

\newcommand{\Psd}{D^{\textnormal{Prange}}}

\DeclareMathOperator*{\punc}{Punc}
\DeclareMathOperator*{\Sp}{Supp}
\DeclareMathOperator*{\supp}{Supp}

\DeclareMathOperator*{\Ver}{\text{\tt{Vrfy}}}
\DeclareMathOperator*{\Sgn}{\text{\tt{Sgn}}}
\DeclareMathOperator*{\Gen}{\text{\tt{Gen}}}
\newcommand{\DOOM}{\ensuremath{\mathrm{DOOM}}}
\newcommand{\Sgnsk}{\ensuremath{\mathtt{Sgn}^{\mathrm{sk}}}}
\newcommand{\Vrfypk}{\ensuremath{\mathtt{Vrfy}^{\mathrm{pk}}}}

\newcommand{\Mrs}{M_{\text{rs}}}

\newcommand{\wt}[1]{|#1|}
\newcommand{\listM}{L_\mv}
\newcommand{\qhash}{q_{\textup{hash}}}
\newcommand{\qsig}{q_{\textup{sign}}}

\newcommand{\Dpub}{\Dc_{\textup{pub}}}
\newcommand{\Drand}{\Dc_{\textup{rand}}}

\newcommand{\Dpubw}{\Dc^{\textup{pub}}_w}
\newcommand{\Dpubwq}{\Dc^{\textup{pub}}_{w,q}}
\newcommand{\Dsw}[1]{\Dc_w^{#1}}
\newcommand{\UV}{(U,U+V)}

\begin{document}
\title{The problem with the SURF scheme
\thanks{This work was supported in part by the Commission of the European Communities through the Horizon 2020 program under project number 645622 PQCRYPTO.}}
\author{Thomas Debris-Alazard\inst{1,2}   \and Nicolas Sendrier\inst{2} \and Jean-Pierre Tillich \inst{2}
 }
\institute{Sorbonne Universit\'{e}s, UPMC Univ Paris 06 \and Inria, Paris\\
\email{\{thomas.debris,nicolas.sendrier,jean-pierre.tillich\}@inria.fr}}

\maketitle
\begin{abstract}
There is a serious problem with one of the assumptions made in the security proof of the SURF scheme. This problem turns out to be easy in the
regime of parameters needed for the SURF scheme to work.
 We give afterwards the old version of the paper for the
reader's convenience. 
\end{abstract}


\section{The problem}

	We give here a polynomial-time algorithm which distinguishes a permuted $(U,U+V)$-code from a random linear code of the same length and dimension. This invalidates one of the security assumptions of the 
security proof of SURF. Our algorithms is based on the fact that for the SURF  parameters the hull of a permuted $(U,U+V)$-code is typically much bigger than that is expected for random linear codes of  the same length and  dimension. Let us start by some definitions and notation. 
	\newline

	{\bf Notation.} Vectors will be written with  bold letters (such as $\xv$). 
        The $i$-th component of $\xv$ is denoted by $x_i$.
	When $\xv$ and $\yv$ are two vectors, $(\xv,\yv)$ denotes their concatenation.
        A binary linear code of length $n$ and dimension $k$ is referred to as an $[n,k]$-code.
	\newline

	{\bf Definitions.} Let $\cC$ be an $[n,k]$-code.
 We now define its hull as:
	$$
	\mathrm{hull}(\cC) \eqdef \cC \cap \cC^{\bot}
	$$
	where $\cC^{\bot}$ denotes its dual which is defined as:
	$$
	\cC^{\bot} = \{ \hv \in \F_{2}^{n} : \forall \cv \in \cC, \mbox{ } \langle \cv,\hv \rangle = 0 \}.
	$$
The scalar product $\langle \cv,\hv \rangle$ is performed here over $\F_2$.
	Moreover, if $\pi$ is a permutation of length $n$ we define:
	$$
	\pi(\cC) \eqdef \{ \pi(\cv) \mbox{ } : \mbox{ } \cv \in \cC \} \mbox{ where } \pi(\cv) = (c_{\pi(i)})_{1 \leq i \leq n}.
	$$
	Let $U,V$ be binary linear codes of length $n/2$ and respective dimension $k_{U},k_{V}$. We define the subset of $\F_{2}^{n}$:
	$$
	(U,U+V) \eqdef \{ (\uv,\uv+\vv) \mbox{ }:\mbox{ } \uv \in U \mbox{ and } \vv \in V \}
	$$
	which is a linear code of length $n$ and dimension $k_{U} + k_{V}$. 
	Public keys of SURF are $\pi(U,U+V)$ codes where both $U$ and $V$ have been chosen uniformly at random among the $\lbrack n/2,k_{U}\rbrack$,$\lbrack n/2,k_{V}\rbrack$-codes and $\pi$ is a random permutation of length $n$. Let us denote by $\Dpub$ this distribution and by $\Drand$ the distribution of $\lbrack n,k_{U}+k_{V}\rbrack$-codes chosen uniformly at random.

	\noindent The security reduction of SURF relies on the difficulty to distinguish $\Dpub$ and $\Drand$. The algorithm we are going to give performs for the SURF parameters in polynomial-time this task with a success probability $1 - \varepsilon$ where $\varepsilon$ is a negligible function in $n$.
	\newline

		It consists in computing the hull of the permuted $(U,U+V)$-code and decide that the code belongs to $\Dpub$
if and only if the dimension is $k_U - k_V$ and to $\Drand$ otherwise. It is readily seen that the hull can be computed in polynomial time.	
	
The reason explaining the correction of this algorithm is given by the two following propositions.
\begin{proposition}[\cite{S97a}] The expected dimension of the hull of a random linear 
code is $O(1)$. It is smaller than $t$ with probability $\geq 1 - O(2^{-t})$.

\end{proposition}

	\begin{proposition} Assume that $k_U \geq k_V$. 
If $\cC$ is picked according to the distribution $\Dpub$ we have with  probability $1 - O(2^{k_V-k_U})$ 
		$$
		\dim\left(\mathrm{hull}(\cC)\right) = k_{U} - k_{V}.
		$$		
	\end{proposition} 

The correction of our attack easily follows from these two propositions.

	\begin{proof}[Sketch of the proof of Proposition 1] Without loss of generality we are going to show that with probability  $1 - O(2^{k_V-k_U})$ we have $\dim(\mathrm{hull}(U,U+V)) = k_{U} - k_{V}$ as the hull is invariant by permuting 
the code positions. It is readily seen that 
$$
(U,U+V)^\bot = (V^{\bot}+U^{\bot},V^{\bot}).
$$
This implies that
		$$
		\mathrm{hull}((U,U+V)) = (U,U+V) \cap (U^{\bot}+V^{\bot},V^{\bot}).
		$$
	Therefore for any vector $(\uv,\uv+\vv) \in \mathrm{hull}((U,U+V))$ where $\uv \in U$ and $\vv \in V$ there exists $\vv \in V^{\bot}$ and $\uv^{\bot} \in U^{\bot}$ such that	
	$$
	\begin{array}{ll}
	 & 
	 \left\lbrace
	 	\begin{array}{l}
			 \uv = \vv^{\bot} + \uv^{\bot} \\
			 \uv + \vv = \vv^{\bot}
	 	\end{array}
	 \right. \\[0.5cm]
	 
	 \iff & 
	 \left\lbrace
	 	\begin{array}{l}
	 		\vv = \uv^{\bot} \\
	 		\uv + \vv = \vv^{\bot}
	 	\end{array}
	 \right. \\[0.5cm]
	\end{array}
	$$
	In this way, the vector $\vv$ lives in $V \cap U^{\bot}$. But we remark that 
	$$
	\dim(V) + \dim(U^{\bot}) = k_{V} + n/2 - k_{U} = n/2 + k_{V} - k_{U} < n/2.
	$$
This can be used to prove that with probability $1-O(2^{k_V-k_U})$ 
we have $V \cap U^{\bot} = \{0\}$ and $\vv = \uv^{\bot} = \mathbf{0}$. It follows that vectors of $\mathrm{hull}(U,U+V)$ are with high probability of the form $(\xv,\xv)$ where $\xv \in U \cap V^{\bot}$. Once again we remark that:
	$$
	\dim(U) + \dim(V^{\bot}) = k_{U} + n/2 - k_{V} = n/2 + k_{U} - k_{V} > n/2 .
	$$
	This can be used to prove that with probability $1-O(2^{k_V-k_U})$ we have
$$
\dim(U \cap V^\bot) = k_U-k_V.
$$
This discussion implies that with probability $1-O(2^{k_V-k_U})$ we have
	$$
	\dim (\mathrm{hull}((U,U+V))) = \dim(U \cap V^{\bot}) = k_{U} - k_{V}. 
	$$
	This gives with a high probability a hull of dimension  $k_{U} - k_{V}$.  This concludes the proof. 
	\end{proof} 

\section*{Discussion}

\par{\bf Bootstrapping from here.} 
This algorithm does not only invalidate our security proof, it can also be used to mount an attack. It starts
by noticing that vectors in the hull are of the form $\pi(\uv,\uv)$ where $\uv \in U \cap V^{\bot}$.
This yields partial information on $\pi$ which leads to a feasible attack on the parameters proposed for the SURF scheme.

\medskip
\noindent
\paragraph{\bf On the condition $k_U \geq k_V$.} What makes the attack feasible is the fact that $k_U \geq k_V$. 
This condition is necessary for our signature scheme to work (this is essentially a consequence of Proposition 3 in the old paper). 
It is precisely this condition that enables to produce for any possible syndrome an error of 
a sufficiently low weight that an attacker who does not know the $(U,U+V)$ structure has a hard time to produce.
The distinguisher does not work anymore in the regime where $k_U < k_V$. However in this case our signature scheme does not work anymore.

\medskip
\paragraph{\bf On the NP-completeness of the distinguishing problem.} Interestingly enough, the NP-completeness of Problem 5
(namely Problem P2': weak $(U,U+V)$ distinguishing) works for parameters that really correspond to the case
$k_U < k_V$. The same holds for another proof we have on the NP-completeness of the $(U,U+V)$ distinguishing problem itself, namely 
		\begin{problem}\textit{(}$(U,U+V)$-distinguishing\textit{)}~\\
		\label{prob:UVDist}
		\begin{tabular}{ll}
			Instance: & A binary linear code $\cC$ and an integer $k_{U}$, \\
			Question: & Is there a permutation $\pi$ of length $n$ such that $\pi(\cC)$ is a $(U,U+V)$-code \\ 
			& where $\dim(U) = k_{U}$ and $|\Sp(V)| = n/2$? 
		\end{tabular}
	\end{problem}
The NP-completeness of this problem is given in the old version in the appendix in Subsection \ref{subsec:proofUVDistbis}. The reduction to three dimensional matching is also 
in the regime where $k_U < k_V$.

\newpage
\begin{center}
\Huge{\bf The old version of the paper starts here.}
\end{center}
\newpage
\title{SURF: A new code-based signature scheme
\thanks{This work was supported in part by the Commission of the European Communities through the Horizon 2020 program under project number 645622 PQCRYPTO.}}
\author{Thomas Debris-Alazard\inst{1,2}   \and Nicolas Sendrier\inst{2} \and Jean-Pierre Tillich \inst{2}
 }
 
\institute{Sorbonne Universit\'{e}s, UPMC Univ Paris 06 \and Inria, Paris\\
\email{\{thomas.debris,nicolas.sendrier,jean-pierre.tillich\}@inria.fr}}

\maketitle

\begin{abstract} We present here a new code-based digital signature scheme. This scheme uses $\UV$ codes where both $U$ and $V$ are random. We show that the distribution of signatures is uniform by suitable rejection sampling. This is one of the key ingredients for our proof  that the scheme achieves {\em existential unforgeability under adaptive chosen message attacks} (EUF-CMA) in the random oracle model (ROM) under two assumptions from coding theory, both NP-complete and strongly related to the hardness of decoding in a random linear code. Another crucial ingredient is the proof that the syndromes produced by $\UV$ codes are statistically indistinguishable from random syndromes. Note that these two key properties are also required for applying a recent and generic proof for code-based signature schemes in the QROM \cite{CD17}. As noticed there, this allows to instantiate the code family which is needed and yields a security proof of our scheme in the QROM. Our scheme also enjoys an efficient signature generation and verification. For a (classical) security of 128 bits, the signature size is less than one kilobyte. Contrarily to a current trend in code-based or lattice cryptography which reduces key sizes by using structured codes or lattices based on rings, we avoid this here and still get reasonable public key sizes (less than 2 megabytes for the aforementioned security level). Our key sizes  compare favorably with TESLA-2, which is an (unstructured) lattice-based signature scheme that has also a security reduction in the QROM. This gives the first
practical signature scheme based on binary codes which comes with a security proof and which scales well with the security parameter: for a security level of $2^\lambda$, the signature size is of order $O(\lambda)$, public key size is of size $O(\lambda^2)$, signature generation cost is of order $O(\lambda^3)$, and signature verification cost is of order $O(\lambda^2)$.
\end{abstract}

\textbf{Keywords:} code-based cryptography, digital signature scheme,
decoding algorithm, security proof.

\section{Introduction}\label{sec:introduction}

\subsubsection*{Code-based signature schemes.} It is a long standing open problem to build an efficient and secure signature scheme based on the
hardness of decoding a linear code which could compete in all respects
with DSA or RSA. Such schemes could indeed give a quantum resistant 
signature for replacing in
practice the aforementioned signature schemes  that are well known to be broken by  
quantum computers. A first partial answer to this question was given in
\cite{CFS01}. It consisted in adapting the Niederreiter scheme \cite{N86} for
this purpose. This requires a linear code for which there exists an efficient  decoding algorithm for 
a non-negligible set of inputs.
This means that if $\Hm$ is an
$r \times n$ parity-check matrix of the code,
there exists for a non-negligible set of elements $\sv$ in  $\mathbb{F}_{2}^r$ an efficient way to
find a word $\ev$ in $\mathbb{F}_{2}^n$ of smallest Hamming weight such that $\Hm
\ev^T=\sv^T$. In such a case, we say that $\sv$, which is generally called a syndrome in the literature,  can be decoded. 
To sign a message $\mv$, a hash function $h$ is used to produce a
sequence $\sv_0, \dots, \sv_\ell$ of elements of $\mathbb{F}_{2}^r$.
For instance $\sv_0 = h(\mv)$ and $\sv_i = h(\sv_0,i)$ for $i>0$.
The first $\sv_i$ that can be decoded defines the signature of $\mv$
as the word $\ev$ of smallest Hamming weight such that
\begin{displaymath}
  \Hm \ev^T = \sv_i^T.
\end{displaymath}

\subsubsection*{The CFS signature scheme.} The authors of \cite{CFS01} noticed that very high rate Goppa codes are
able to fulfill this task, and their scheme can indeed be
considered
as the first step towards a solution of the aforementioned problem.
Moreover they gave a security proof of their
scheme relying only on the assumption that two problems were hard, namely (i) decoding a generic
linear code and (ii) distinguishing a Goppa code from a random linear
code with the same parameters.
However, afterwards it was realized that the parameters proposed in \cite{CFS01} can be attacked
by an unpublished attack of Bleichenbacher.
The significant  increase of parameters needed to thwart the Bleichenbacher
attack was fixed by a slight variation \cite{F10}.  However, this modified scheme
is not able to fix two other worrying drawbacks of the CFS scheme, namely
\begin{itemize}
\item[(i)] a lack of security proof in light of the distinguisher of high rate Goppa codes found in \cite{FGOPT11} 
  (see also \cite{FGOPT13} for more details) which shows that the hypotheses used in \cite{CFS01} 
  to give a security proof of the signature scheme were not met,
\item[(ii)]   poor scaling of the parameters
  when security has to be increased. 
It can be readily seen that the complexity $S$ of the best known attack scales only polynomially as $S \approx K^{t/2}$ where
$K$ is the key size in bits and $t$ is some parameter that has to be kept very small (say smaller than 12 in practice), since  the number 
of syndromes $\sv_i$ that have to be computed before finding one that can be decoded is roughly $t!$.
\end{itemize}

\subsubsection*{Other code-based signature schemes.} Other signature schemes based on codes were also given  in the literature
such as for instance the KKS scheme \cite{KKS97,KKS05} or its variants
\cite{BMS11,GS12}.
But they can be considered at best to be one-time signature schemes in
the light of the attack given in \cite{COV07} and great care has to be
taken to choose the parameters of these schemes as shown by \cite{OT11}
which broke all the parameters proposed in \cite{KKS97,KKS05,BMS11}.

There has been some revival of 
the CFS strategy \cite{CFS01},
by choosing other code families. The new code families that were used are 
LDGM codes in \cite{BBCRS13}, i.e. codes with a Low Density
Generator Matrix, or (essentially) convolutional codes
\cite{GSJB14}. There are still some doubts
that there is a way to choose the parameters of the scheme \cite{GSJB14} in order to 
avoid the attack \cite{LT13} on the McEliece cryptosystem based on convolutional codes 
\cite{LJ12}
and the LDGM scheme was broken in \cite{PT16}. 

A last possibility is to use the Fiat-Shamir heuristic to turn a zero-knowledge authentication scheme into a signature scheme.
When based on the Stern authentication scheme \cite{S93} this gives a code-based signature scheme. However this approach leads to really large signature sizes (of the order of hundreds of thousands of bits). This represents a complete picture 
of code-based signature schemes based on the Hamming metric. There has been 
some recent progress in this area for another metric, namely the rank metric \cite{GRSZ14} with the 
RankSign scheme. This scheme enjoys remarkably small key sizes, 
it is of order  tens of thousands bits for 128 bits of security. It also comes with a partial security proof showing that signatures do not leak information when the number of available signatures is smaller than some bound
depending on the code alphabet, but ensuring this condition 
represents a rather strong constraint on the parameters of  RankSign. Moreover there is no overall reduction of the security  to well identified problems in (rank metric) coding theory. 
Irrespective of the merits of this signature scheme,
it is certainly desirable to also have  a  signature scheme  for the Hamming metric due to the general faith in the hardness of decoding in it.

\subsubsection*{Moving from error-correcting codes to lossy source codes.} 
It can be argued that the main problem with the CFS approach is to find a family of linear codes that are at the same time (i) indistinguishable from a random code 
and (ii) that have
a non-negligible fraction of syndromes that can be decoded. There are not so many codes for which (ii) can be achieved and this is probably too much to ask for.
However if we relax a little bit what we ask for the code, namely just a code such that the equation (in $\ev$) 
\begin{equation}
  \label{eq:fundamental_equation}
  \Hm \ev^T = \sv^T
\end{equation}
admits for most of the $\sv$'s 
a solution $\ev$ of small enough weight, then there are many more codes that are able to fulfill this task. This kind of codes are not used in  error-correction but can be found in lossy source coding or source-distortion theory where the problem is to find codes with an associated decoding algorithm which can approximate {\em any} word of the ambient space by a close enough codeword. In the case of linear codes, this means a code and a associated decoding algorithm that can find for any syndrome $\sv$ a vector $\ev$ of small enough weight satisfying \eqref{eq:fundamental_equation} where $\Hm$ is a parity-check matrix of the code. 

Solving \eqref{eq:fundamental_equation} is the basic problem upon which all code-based cryptography relies. 
This problem has been studied for a long time and despite many efforts on this issue \cite{P62,S88,D91,B97b,MMT11,BJMM12,MO15,DT17} the best algorithms for solving this problem \cite{BJMM12,MO15} are exponential in the weight $w$ of $\ev$ as long as
$w = (1-\epsilon)r/2$ for any $\epsilon>0$. 
Furthermore when $w$ is sublinear in $n$, the exponent of the best known algorithms has not changed \cite{CS16} since the Prange algorithm \cite{P62} dating back to the early sixties. Moreover,  it seems very difficult to lower this 
exponent by a multiplicative factor smaller than $\frac{1}{2}$ in the quantum 
computation model as illustrated by 
\cite{B10,KT17}.

\subsubsection*{Our contribution: a new signature scheme based on $(U,U+V)$ codes.} 
Convolutional codes, LDGM and polar codes  come with a decoding algorithm which is polynomial for weights below $r/2$. 
They could theoretically be used in this context.
However in the light of the key attacks \cite{LJ12,PT16, BCDOT16} performed on related schemes, it seems  very difficult to propose parameters which avoid those attacks.
We are instead introducing a new class of codes in this context namely $(U,U+V)$ codes.
A $(U,U+V)$ code is just a way of building a code of length $n$ when we have two codes $U$ and $V$ of length $n/2$. It consists in 
$$
(U,U+V) \eqdef \{(\uv,\uv+\vv): \uv \in U, \vv \in V\}.
$$
Generalized $(U,U+V)$ codes have already been proposed in the cryptographic context for building a McEliece encryption scheme \cite{MT16}.
However, there it was suggested to take $U$ and $V$ to be codes that have an efficient decoding algorithm (this is mandatory in
the encryption context). In the signature context, when we just need to find a small enough solution of \eqref{eq:fundamental_equation} this is not needed.
In our case, we can afford to choose {\em random} codes for $U$ and $V$. It turns out that if we choose $U$ and $V$ random with the right choice
of the dimension of $U$ and $V$, then a suitable use of the Prange algorithm on the code $U$ and the code $V$ provides an advantage in this setting.
It allows to solve \eqref{eq:fundamental_equation} for weights $w$ that are significantly below $r/2$, that is in the range of weights 
for which there are an exponential number of solutions but the best decoding algorithms are still exponential. 

Moreover, by tweaking a little bit the output of the Prange algorithm in our case and performing an appropriate rejection sampling, it turns out that the 
signatures are indistinguishable from a random word of weight $w$. 
Furthermore we also show that syndromes $\Hm \ev^T$ associated to this kind of codes are statistically indistinguishable from random syndromes when errors $\ev$ are drawn uniformly at random among the words of weight $w$.
These are the two key properties that allow to give a tight security proof of our signature scheme which relies only on two problems:
\begin{itemize}
\item[P1:] Solving the decoding problem \eqref{eq:fundamental_equation} when $w$ is sufficiently below $r/2$
\item[P2:] Deciding whether a linear code is permuted $\UV$ code or not.
\end{itemize}
Interestingly enough some recent work \cite{CD17} has shown that these two properties (namely statistical indistinguishability of the signatures and
the syndromes associated to the code family chosen in the scheme) are also enough to obtain a tight security proof in the quantum random oracle model (QROM) for generic code-based signatures under the assumption that Problem P1 stays hard against a quantum computer and that the code family used is computationally indistinguishable from generic linear codes. In other words, as noticed in \cite{CD17}, this can be used to give a tight security proof 
of our $\UV$ codes in the QROM.
 
Problem P1 is the problem upon which all code-based cryptography
relies. 
Here we are in a case where there are multiple solutions of 
\eqref{eq:fundamental_equation} and the adversary may produce any number of instances of
\eqref{eq:fundamental_equation} with the same matrix $\Hm$ and various
syndromes $\sv$ and is interested in solving only one of them. This relates to the, so called, Decoding One Out of
Many (DOOM) problem. This problem was first considered in \cite{JJ02}.
It was shown there how to modify slightly the known algorithms for decoding
a linear code in order to solve this modified problem. This modification was 
later analyzed in \cite{S11}. The parameters of the known algorithms
for solving \eqref{eq:fundamental_equation} can be easily adapted to this
scenario where we have to decode 
simultaneously
multiple instances which all have multiple solutions.

Problem P2 might seem at first sight to be an ad-hoc problem. However as we are going to show, it is an NP-complete problem (see Theorem \ref{theo:UVNP} in Subsection \ref{ss:NPcomplete}). 
Problem P1 is known to be NP-complete and therefore we have a signature scheme whose security relies entirely on NP-complete problems. This is the 
first time that a code-based signature scheme is proposed with such features. Interestingly enough, even weak versions of this problem are NP-complete.
For instance, even in the case when the permutation is restricted to leave globally stable the right and left part, detecting whether the resulting code is a permuted $\UV$-code is already an NP-complete problem (see Problem \ref{prob:P1'} and 
Theorem \ref{th:NPcomplete} in Subsection \ref{ss:NPcomplete}).
Furthermore, we are really in a situation where the resulting permuted $(U,U+V)$ code is actually very close to a random code. The only different behavior
that can be found seems to be in the weight distribution for small weights. In this case, the permuted $(U,U+V)$ code has some codewords of 
a weight slightly smaller than the minimum distance of a
random code of the same length and dimension. It is very tempting to conjecture that the best algorithms for solving 
Problem P2 come from detecting such codewords. This approach can be 
easily thwarted  by choosing the parameters of the scheme in such a way that the best algorithms for solving this task are of prohibitive complexity. Notice that the best algorithms that we have for detecting such codewords are in essence precisely the generic algorithms for solving Problem P1. In some sense, it seems that we might rely on the very same problem, even if our proof technique does not show this.

All in all, this gives the first practical signature scheme based on binary codes which comes with a security proof and which scales well with the parameters: it can be shown that if one wants a security level of $2^\lambda$, then signature size is of order $O(\lambda)$,  public key size is of order $O(\lambda^2)$, signature generation  is of order $O(\lambda^3)$, whereas signature verification is of order $O(\lambda^2)$. It should be noted that contrarily to the current thread of research in code-based or lattice-based cryptography which consists in relying on structured codes or lattices based on ring structures in order to decrease the key-sizes we did not follow this approach here. This allows for instance to rely on the NP-complete problem P1 which is generally believed to be hard on average rather that on decoding in quasi-cyclic codes for instance whose status is still unclear with a constant number of circulant blocks. Despite the fact that we did not use the standard approach for reducing the key sizes relying on quasi-cyclic codes for instance, we obtain 
acceptable key sizes (less than 2 megabytes for 128 bits of security) which compare very favorably to unstructured lattice-based signature schemes
such as TESLA-2 for instance \cite{ABBDEGKP17}. This is due in part to the tightness of our security reduction.

We would like to conclude this introduction by pointing out the simplicity of the parameter selection of our scheme (see end of \S\ref{sec:statDist}). The parameters for a given length $n$ are
chosen as
$$
w = \lfloor n \frac{3-\sqrt{5}}{4} \rfloor \quad ; \quad k_{V} = w \quad ; \quad k_{U} = n/2 - k_{V}
$$
where $w$ is the signature weight and $k_{U}$ (\textit{resp.} $k_{V}$) is the dimension of the code $U$ (\textit{resp.} $V$).  The length $n$ is then chosen in order to thwart the attacks on problems P1 and P2 (actually it is Problem P1 that will govern the length selection).

\subsubsection*{Organization of the paper.}
The paper is organized as follows, we present our scheme in
\S\ref{sec:UVsgnScheme}, in \S\ref{sec:securityProof} we prove it is secure under 
\textit{existential unforgeability under an adaptive chosen message
  attack} (EUF-CMA) in the ROM, in relation with this proof we respectively examine in
\S\ref{sec:statDist}, \S\ref{sec:attsourcedistortion}, and
\S\ref{sec:keyAttack}, how to produce uniformly distributed
signatures as well as the best message and key attacks. Finally we
give some set of parameters on par with the security reduction and
with the current state-of-the-art for decoding techniques.

\medskip

 \section{Notation} 

We provide here some notation that will be used throughout the paper. 
\medskip

\noindent
{\bf General notation.} 
The notation $x \eqdef y$ means that $x$ is defined to be equal to  $y$. We denote by $\mathbb{F}_{2}$
the finite field with $2$  elements and by
$S_{w}$ the subset of $\F_2^n$ of words of weight $w$. 
\medskip

\noindent
{\bf Vector notation.} 
Vectors will be written with  bold letters (such as $\ev$) and  uppercase bold letters are used to denote matrices (such as $\Hm$). Vectors are in row notation.
Let $\xv$ and $\yv$ be two vectors, we will write $(\xv,\yv)$ to denote their concatenation.
For a vector $\xv=(x_i)_{1 \leq i \leq n}$ and a permutation $\pi$ of length $n$ we denote by 
$\pi(\xv)$ the vector $(x_{\pi(i)})_{1 \leq i \leq n}$.
We also denote for a subset $I$ of positions of the vector $\xv=(x_i)_{1 \leq i \leq n}$ by $\xv_I$ the vector whose components are those of $\xv$ which are indexed by $I$, i.e. 
$$
\xv_I = (x_i)_{i \in I}.
$$ 
We define the support of $\xv$ as
$$
\Sp(\xv) \eqdef \{ i \in \{1,\cdots,n \} \mbox{ such that } x_{i} \neq 0 \}
$$
The Hamming weight of $\xv$ is denoted by 	
$|\xv|$.
By some abuse of notation, we will use the same notation 
to denote the size of a finite set: $|S|$ stands for the size of the finite set $S$.  
It will be clear from the context whether $|\xv|$ means the Hamming weight or the size of a finite set. Note that 
$$
|\xv| = |\Sp(\xv)|.
$$
\medskip

\noindent
{\bf Probabilistic notation.} Let $S$ be a finite set, then $x \Unif S$ means 
that $x$ is assigned to be a random element chosen uniformly at random in $S$. For a distribution $\Dc$ we write $\xi \sim \Dc$ to indicate that the random variable $\xi$ is chosen according to $\Dc$. 
The uniform distribution on a certain discrete set is denoted by $\Uc$. The set will be specified in the text. We denote the uniform distribution on $S_{w}$ by $\mathcal{U}_{w}$. When we have probability distributions $\Dc_1$, $\Dc_2$, \dots, $\Dc_n$ over discrete sets 
$\Ec_1$, $\Ec_2$, \dots, $\Ec_n$, we denote by $\Dc_1 \otimes \Dc_2 \otimes \cdots \otimes \Dc_n$ the product probability distribution, i.e
$\Dc_1 \otimes \cdots \otimes \Dc_n(x_1,\dots,x_n) \eqdef \Dc_1(x_1) \dots \Dc_n(x_n)$ for 
$(x_1,\dots,x_n) \in \Ec_1 \times \cdots \times \Ec_n$. The $n$-th power product of a distribution $\Dc$ is denoted by $\Dc^{\otimes n}$, i.e.
$\Dc^{\otimes n} \eqdef \underbrace{\Dc \otimes \cdots \otimes \Dc}_{n \;\text{times}}$.

Sometimes when we wish to emphasize on which probability space the probabilities or the expectations are taken, we denote by a subscript the random variable specifying the associated probability space over which the probabilities or expectations are taken.
For instance the probability $\prob_X(\Ec)$ of the event $\Ec$ is taken over $\Omega$ the probability space over which the random variable
$X$ is defined, \textit{i.e.} if $X$ is for instance a real random variable, $X$ is a function from a probability space $\Omega$ to $\R$, and the aforementioned probability is taken according to
the probability chosen for $\Omega$.
\medskip

\noindent
{\bf Coding theory.}
A binary linear code $\mathcal{C}$ of length $n$ and dimension $k$ is a subspace of $\mathbb{F}_{2}^{n}$ of dimension $k$ and is usually defined by a parity-check matrix $\Hm$ of size $r \times n$ as
$$
\Cc = \left\{ \xv \in \mathbb{F}_{2}^{n}: \Hm \xv^T=\mathbf{0}\right\}.
$$ 	
When $\Hm$ is of full rank (which is usually the case) we have $r = n-k$.
The rate of this code (that we denote by $R$)  is defined as
$R \eqdef \frac{k}{n}$. In this case we say that $\cC$ is a $\lbrack n,k \rbrack$-code.
 \section{The $(U,U+V)$-signature Scheme}
\label{sec:UVsgnScheme} 

\subsection{The general scheme $\mathcal{S}_{\textup{code}}$}
\label{sec:UVsgnScheme1}

Our scheme can be viewed as a probabilistic version of the full domain hash (FDH) signature scheme as defined in \cite{BR96} which is similar to the probabilistic 
signature scheme introduced in \cite{C02} except that we replace RSA with a trapdoor function based upon the hardness of Problem P1.
Let $\Cc$ be a binary linear code of length $n$ defined by a parity-check matrix $\Hm$. The one way function $f_{\Hm,w}$ we  consider is given by 
\begin{displaymath}
  \begin{array}{lccc}
    f_{\Hm,w} :  &S_{w} & \longrightarrow & \mathbb{F}_{2}^{n-k}\\
           &\ev & \longmapsto & \ev\Hm^{T}
  \end{array}
\end{displaymath}
Inverting this function on an input $\sv$ amounts to solve Problem P1.
We are ready now to give the general scheme we consider. 
We assume that we have a family of codes which is defined by a set $\Fc$ of parity-check matrices of size $(n-k)\times n$ such that for all $\Hsec \in \Fc$ we have an algorithm $D_{\Hsec,w}$ which on input $\sv$ computes $\ev \in f_{\Hsec,w}^{-1}(\sv)$. Then we pick uniformly at random $\Hsec \in \Fc$, an $n \times n$ permutation matrix $\Pm$, a non-singular matrix $\Sm \in \mathbb{F}_{2}^{(n-k) \times (n-k)}$ which define the secret and public key as:
$$
sk \leftarrow (\Hsec,\Pm,\Sm) \mbox{ } ; \mbox{ }  pk \leftarrow \Hpub \mbox{ where } \Hpub \eqdef \Sm \Hm \Pm 
$$

\begin{remark} Let $\Csec$ be the code defined by $\Hsec$, then $\Hpub$ defines the following code:
	$$
	\Cpub = \{ \cv \Pm: \cv \in \Csec\}.
	$$
\end{remark}

We also select a cryptographic hash function $\hash : \{0,1\}^{*} \rightarrow \mathbb{F}_{2}^{n-k}$
and a parameter $\lambda_{0}$ for the random salt $\rv$. The algorithms \Sgnsk\ and \Vrfypk\ are defined as follows
\begin{center}
  \begin{tabular}{l@{\hspace{3mm}}|@{\hspace{3mm}}l}
    $\Sgnsk(\mv)\!\!: \qquad \qquad \qquad$ & $\Vrfypk(\mv,(\ev',\rv))\!\!:$ \\
    $\quad \rv \Unif \{ 0,1 \}^{\lambda_{0}}$ &$\quad \sv \leftarrow \hash(\mv ,\rv)$ \\
    $\quad \sv \leftarrow \hash(\mv ,\rv)$ & $\quad \texttt{if } \Hpub \ev'^T = \sv^{T} \texttt{ and } |\ev'| = w \texttt{ return } 1$\\ 
    $\quad\ev \leftarrow D_{\Hsec,w}(\Sm^{-1}\sv^{T})$ &$\quad \texttt{else return } 0 $ \\
    $\quad \texttt{return}(\ev\Pm,\rv)$& \\			
  \end{tabular} 
\end{center}

\begin{remark} We add a salt in the scheme in order to have a tight security proof.
\end{remark}

The correction of the verification step  (i.e. that the pair $(\ev \Pm,\rv)$ passes the verification step) follows from the fact that 
by definition of $D_{\Hsec,w}(\Sm^{-1}\sv^{T})$ we have $\Hsec\ev^{T} = \Sm^{-1}\sv^{T}$.
Therefore  $\Hpub (\ev\Pm)^{T} = (\Hpub \Pm^{T}) \ev^{T}= \Sm \Hsec \ev^{T} = \Sm \Sm^{-1}\sv^{T}=\sv^{T}$.
We also have $|\ev \Pm| = |\ev| = w$.

To summarize, a valid signature of a message $\mv$ consists of a pair $(\ev,\rv)$ such that $\Hpub \ev^{T} = \hash(\mv,\rv)^{T}$ with $\ev$ of Hamming weight $w$.

\subsection{Source-distortion codes and decoders}
\label{subsec:2.2}

Source-distortion theory is a  branch of information theory 
which deals with obtaining a family of codes, with an associated set of parity-check matrices $\Hm \in \Fc$, of the smallest possible dimension which can be used in our setting
(\textit{i.e.} for which we can invert $f_{\Hm,w}$). Recall that a linear code is a vector space and the dimension of the code is defined as the dimension of this vector space.
For a linear code specified by a full rank parity-check matrix of size $r \times n$, the dimension $k$ of the code is equal to $n-r$.
It is essential to have the smallest possible dimension in our cryptographic application, since this makes the associated problem P1 harder: the smaller $n-r$ is, the bigger $r$ is and the further 
away $w$ can be from $r/2$ (where solving P1 becomes easy). This kind of codes is used for performing lossy coding of a source. Indeed 
assume that we can perform this task, then this means that for every binary word $\yv$, we compute $\sv^{T} \eqdef \Hm \yv^{T}$, we find $\ev$ of Hamming weight $w$ such that $\Hm \ev^{T} = \sv^{T}$ which leads to deduce a codeword $\cv \eqdef \yv - \ev$ which is at distance $w$ from $\yv$.
The word $\yv$ is compressed with a compact description of $\cv$. Since the dimension of the code is $n-r$ we just need  $n-r$ bits to store a description of $\cv$.
We have replaced here $\yv$ with  a word which is not too far away from it. Of course, the smaller $n-r$ is,  the smaller the compression rate 
$\frac{n-r}{n}$ is. There is some loss by replacing $\yv$ by $\cv$ since we are in general close to $\yv$ but not equal to it.

In this way, finding a close codeword $\cv$ of a given word $\yv$ is equivalent to find for the syndrome $\Hm \yv^{T}$ a low weight ``error'' $\ev$ such that $\Hm \ev^{T} = \Hm \yv^{T}$.
For our purpose it will be more convenient to adopt the error and syndrome viewpoint than the codeword viewpoint.
To stress the similarity with error-correction we will 
call the function which associates to a syndrome $\sv$ such an $\ev$ a source-distortion decoder.

\begin{definition}[Source Distortion Decoder]\label{def:sddeco} 
Let $n$, $k \leq n$ be
  integers and let $\mathcal{F}$ be a family of parity-check matrices (which define binary linear codes of
  length $n$ and dimension $k$). A source distortion decoder for
  $\mathcal{F}$ is a probabilistic algorithm $D$:
  \begin{displaymath}
    \begin{array}{lccc}
      D : & \Fc \times \mathbb{F}_{2}^{n-k} & \longrightarrow & \mathbb{F}_{2}^{n} \\
                & (\Hm,\sv) &\longmapsto & \ev 
    \end{array}
  \end{displaymath}
  such that $\Hm\ev^{T}=\sv^{T}$.  When the weight of the error is
  fixed, we call it a decoder of fixed distortion $w$ and we denote it
  by $D_{w}$.  We say that the distortion $w$ is achievable if
  there exists a family of codes with a decoder of fixed distortion
  $w$.
\end{definition} 
This discussion raises a first question: for given $n$ and $k$, what
is the minimal distortion $w$ which is achievable? We know from
Shannon's rate-distortion theorem that the minimal $w$ is given by the
Gilbert-Varshamov bound $d_{\textup{GV}}(n,k)$ which follows:

\begin{definition}[Gilbert-Varshamov's bound]\label{dfn:gv} For given integers $n$ and $k$ such that $k \leq n$, the Gilbert-Varshamov bound 
  $d_{\textup{GV}}(n,k)$ is given by:
  \begin{displaymath}
    d_{\textup{GV}}(n,k) \eqdef n h^{-1}\left(1- k/n\right)
  \end{displaymath}
  where $h$ denotes the binary entropy: $h(x) = - x \log_2 x -(1-x) \log_2(1-x)$ and $h^{-1}$ its inverse defined on $[0,1]$ and whose
  range is $[0,\frac{1}{2}]$.
\end{definition}

\subsubsection{Achieving distortion $w=(n-k)/2$ with the Prange technique.}
The study of random codes shows that they achieve the Gilbert-Varshamov source-distortion bound in average. Nevertheless we do not know for them an efficient source-distortion algorithm. However, as the following proposition shows, it is not the case when the distortion $w$ is higher. When $w=(n-k)/2$ there is a very efficient decoder using the Prange technique \cite{P62} for decoding. To explain it consider a parity-check matrix $\Hm$ which defines a linear code $\Cc$ of dimension $k$ and length $n$.
We want to find for a given $\sv \in \mathbb{F}_{2}^{n-k}$ an error $\ev$ of low weight such that $\Hm\ev^{T}= \sv^{T}$.
$\Hm$ is a full-rank matrix and it therefore contains an invertible submatrix $\Am$ of size $(n-k)\times (n-k)$. 
We choose a set of positions $I$ of size $n-k$ for which $\Hm$ restricted to these positions is a full rank matrix. For simplicity assume that this matrix is in 
the first $n-k$ positions: $\Hm = \begin{pmatrix}  \Am | \Bm\end{pmatrix}$. We look for an $\ev$ of the form 
$\ev = (\ev',\mathbf{0})$ where $\ev' \in \mathbb{F}_{2}^{n-k}$. We should therefore have $\sv^T= \Hm \ev^T= \Am {\ev'}^T$, that is 
${\ev'}^T = \Am^{-1}\sv^T$. The expected weight of $\ev'$ is $\frac{n-k}{2}$ and it is easily verified that by randomly picking a 
random set $I$ of size $n-k$ we have to check a polynomial number of them until finding an ${\ev'}^T$ of weight exactly $(n-k)/2$. 

\paragraph{Notation.}
  We denote by $\Psd_{(n-k)/2}$ this fixed distortion decoder and by $\Psd$ the decoder which picks a random subset until finding one 
  for which $\Hm$ restricted to the columns corresponding to $I$ is invertible and computes $\ev'$ as explained above.  $\Psd$  does not necessarily output an error of weight $(n-k)/2$. 
  \newline

From the previous discussion we easily obtain 
\begin{proposition}[Generic Source Distortion Decoder]\label{prop:g}~\\
	The decoder $\Psd_{(n-k)/2}$ runs in polynomial time on average over full rank $(n-k)\times n$ matrices.
\end{proposition}
When we consider in general the family of random parity-check matrices (which define random linear codes) we speak about generic source-distortion decoders as there is no structure, except linearity of the code they define. 
In contrast to the distortion $(n-k)/2$, the only algorithms we know for linear codes for smaller values of $w$  are all exponential in the distortion. 
This is illustrated by Figure \ref{fig:compExp} where we give the exponents (divided by the length $n$)  of the complexity in base $2$ as a function of the distance, for the fixed rate $R = k/n = 0.5$, of the best generic fixed-$w$ source-distortion decoders. As we see, the normalized exponent is $0$ for 
distortion $(n-k)/2$ and the difficulty increases as $w$ approaches the Gilbert-Varshamov bound (which is equal approximately to $0.11 n$ in this case). 

\begin{figure}
  \begin{center} 
    \caption{Normalized exponents in base 2 of the best generic fixed-$w$ source distortion decoders.}
    \includegraphics[scale=0.4]{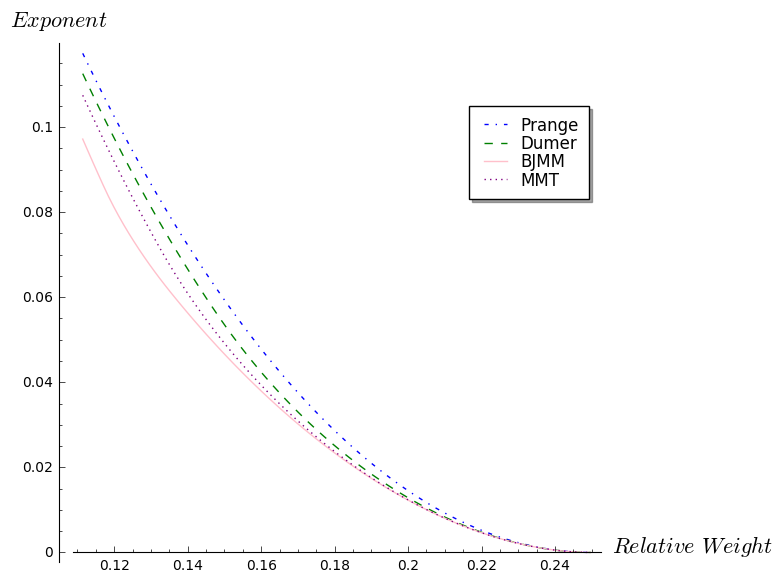}
    \label{fig:compExp}
  \end{center}
\end{figure}

\subsubsection{Decoding Errors and Erasures Simultaneously.}
In the following problem, the word $\xv$, more precisely its support, is called the {\em erasure pattern}.
\begin{problem}[Decoding Errors and Erasures]~\\
	\begin{tabular}{ll}
		Instance: & $\Hm \in \mathbb{F}_{2}^{(n-k) \times n}$,
		$\sv \in \mathbb{F}_{2}^{n-k}$, $\xv \in \mathbb{F}_{2}^{n}$, $\nu$
		integer \\
		Output: & $\ev \in \mathbb{F}_{2}^{n}$ such that $\Hm\ev^{T} = \sv^{T}$ and $\left|\supp(\ev)\setminus\supp(\xv)\right|=\nu$
	\end{tabular}
\end{problem}
In fact, the weight of the solution $\ev$ is constrained outside of the erasure pattern $\xv$. Within the erasure pattern the coordinates of $\ev$ can take any value.
For the sake of simplicity, we will overload the notation and denote $D_\nu(\Hm,\sv,\xv)$ a solution of the above problem whereas $D_\nu(\Hm,\sv)$ denotes the (erasure-less) decoding of $\nu$ errors. The problem of erasure decoding appears very naturally in coding theory, including in source-distortion problem. We may reduce the error and erasure decoding to an error only decoding in a smaller code.
\begin{proposition}
	\label{prop:syndPuncAlt} 
	Let $\Hm\in\F_2^{(n-k)\times n}$ and $\xv\in\F_{2}^n$ be such that the $\wt{\xv}=\rho$ columns of $\Hm$ indexed by $\supp(\xv)$ are independent. For any $\sv\in\F_2^{n-k}$ we can derive $\ev=D_\nu(\Hm,\sv,\xv)$ from $\ev''=D_\nu(\Hm'',\sv'')$ in polynomial time where
	\begin{enumerate}[(i)]
		\item $\Hm''\in\F_2^{(n-k-\rho)\times(n-\rho)}$ can be derived in polynomial time from $\Hm$ and $\xv$,
		\item $\sv''\in\F_2^{n-k-\rho}$ can be derived in polynomial time from $\Hm$, $\xv$, and $\sv$.
	\end{enumerate}
\end{proposition}
\begin{proof}
	Without loss of generality, we assume that the `1's in $\xv$ come first, $\xv=(1\cdots1,0\cdots0)$. A Gaussian elimination on $\Hm$ using the first $\rho$ positions as pivots yields
	$$ \Sm\Hm = \left(\begin{array}{c|clc}
	~~\Imat_{\rho}~~ & ~~ & \Hm' & ~~ \\ \hline
	\zero & & \Hm'' &
	\end{array}\right) $$
	for some non-singular matrix $\Sm$. Let $(\sv',\sv'')=\sv\Sm^T$ with $\sv'\in\F_2^\rho$ and $\sv''\in\F_2^{n-k-\rho}$, $\ev''=D_\nu(\Hm'',\sv'')$, $\ev'=\sv'+\ev''\Hm'^T$, and $\ev=(\ev',\ev'')$. We easily check that $\Hm\ev^T=\sv^T$ and $|\supp(\ev)\setminus\supp(\xv)|=|\supp(\ev'')|=\nu$, thus $\ev=D_\nu(\Hm,\sv,\xv)$. All operations, except possibly the call to $D_\nu$, are polynomial time.\qed
\end{proof}
The reduction of the above proposition applies to $\Psd$. Given $(\Hm,\sv,\xv)$ and using the notation of the proof, we set $\ev''=\Psd(\Hm'',\sv'')$ and we denote $\ev=(\ev',\ev'')=\Psd(\Hm,\sv,\xv)$ the corresponding error. With fixed distortion we have $\ev''=\Psd_{(n-k-|\xv|)/2}(\Hm'',\sv'')$ and we denote $\ev=\Psd_{(n-k-|\xv|)/2}(\Hm,\sv,\xv)$.

Finally, let us point out that if $\Hm$ is the parity check matrix of a binary linear $[n,k]$-code $\Cc$, the matrix $\Hm''$ that appears in Proposition~\ref{prop:syndPuncAlt} is the parity-check matrix of the punctured code in $I=\supp(\xv)$ as defined below:
\begin{definition}[Punctured code]
	Consider a code $\Cc$ of length $n$. The punctured code $\punc_{I}(\Cc)$ in a set of positions
	$I \subset \{1,\dots,n\}$ is a code of length $n -|I|$ defined as
	\begin{displaymath}
	\punc_{I}(\Cc) \eqdef \{\cv_{\bar{I}}:\cv \in \Cc\}
	\end{displaymath}
	where $\bar{I} = \{1,\dots,n\}\setminus I$.
	\end{definition}
Therefore, what Proposition~\ref{prop:syndPuncAlt} says is that decoding $\nu$ errors and $\rho$ erasures in an $[n,k]$-code is essentially the same thing as decoding $\nu$ errors in an $[n-\rho,k]$-code.

\subsection{The $(U,U+V)$ Code Family and Its Decoding}
Source-distortion theory has found over the years several families of codes with an efficient source-distortion algorithm which achieves 
asymptotically the Gilbert-Varshamov source-distortion bound, one of the most prominent ones being probably the Arikan polar codes \cite{A09} (see \cite{K09b}). The naive way would be to build our signature on such a code-family and hoping that permuting the code positions and 
publishing a random parity-check matrix of the permuted code would destroy all the structure used for decoding. 
All known families of codes used in this context have low weight codewords and this can be used to mount an attack. We will proceed differently 
here and introduce in this setting the $(U,U+V)$ codes mentioned in the introduction. The point is that they (i) have very little 
structure, (ii) have a very simple source-distortion decoder which is more powerful than the generic source decoder,
(iii) they do not suffer from low weight codewords as was the case with the aforementioned families.
It will be useful to recall here that
\begin{definition}
  [$(U,U+V)$-Codes]
  Let $U$, $V$ be linear binary codes of length $n/2$ and dimension $k_{U}$, $k_{V}$. We define the subset of $\mathbb{F}_{2}^{n}$:
  \begin{displaymath}
    (U,U+V) \eqdef \{ (\uv,\uv+\vv) \mbox{ such that } \uv \in U \mbox{
      and } \vv \in V \}
  \end{displaymath}
  which is a linear code of length $n$ and dimension $k = k_{U} + k_{V}$. The resulting code is of minimum
  distance $\min(2d_U,d_V)$ where $d_U$ is the minimum distance of $U$ and 
  $d_V$ is the minimum distance of $V$.
   A parity-check matrix of such a code is given by
  $$
  \begin{pmatrix}
  \Hm_{U} & \mathbf{0} \\
  \Hm_{V} & \Hm_{V} \\
  \end{pmatrix}
  $$
  where $\Hm_{U} \in \mathbb{F}_{2}^{(n/2 - k_{U})\times n/2}$ (resp. $\Hm_{V}\in \mathbb{F}_{2}^{(n/2 - k_{V})\times n/2}$) is a parity-check matrix of $U$ (resp. $V$).
\end{definition}
We are now going to present a source-distortion for a $(U,U+V)$ code. 
\newline

We can use the generic source-distortion decoder of Proposition \ref{prop:g} for source distortion decoding a $(U,U+V)$ code. 
Assume that we have a $(U,U+V)$ code of length $n$ of parity-check matrix $\Hsec \eqdef \begin{pmatrix}
\Hm_{U} & \mathbf{0} \\ 
\Hm_{V} & \Hm_{V} \\
\end{pmatrix}$ where $\Hm_{U},\Hm_{V}$ are random and a syndrome $\sv = (\sv_1,\sv_2)\in\F_2^{n/2-k_{U}}\times\F_2^{n/2 - k_{V}}$ that we want to decode. Let us first remark that, for $\ev = (\ev_{1},\ev_{2})\in\F_2^{n/2}\times\F_2^{n/2}$,
$$
\Hsec \ev^{T} = \begin{pmatrix}
\Hm_{U} \ev_{1}^{T} \\
\Hm_{V} (\ev_{1} + \ev_{2})^{T}
\end{pmatrix} = \begin{pmatrix}
\sv_{1}^{T} \\
\sv_{2}^{T} 
\end{pmatrix} \iff \Hm_{U} \ev_{1}^{T} = \sv_{1}^{T} \mbox{ and } \Hm_{V} (\ev_{1} + \ev_{2})^{T} = \sv_{2}^{T}$$
In this way, we first decode $\sv_2$ in $V$ to find $\ev_{V} \eqdef \ev_{1} + \ev_{2}$. That is $\ev_V = \Psd_{(n/2-k_{V})/2}(\Hm_{V} ,\sv_{2})$ with Prange's polynomial time fixed distortion algorithm. We next decode $\sv_1$ in $U$ using $\ev_V$ as an erasure pattern. The idea here is that $\ev_V$ covers a large part of $\ev_1$ leaving us with an error which is, hopefully, easier to find. We compute $\ev_U = \Psd_{(n/2-k_{U}-|\ev_{V}|)/2}(\Hm_{U},\sv_{1},\ev_{V})$ with Prange's polynomial time fixed distortion algorithm. We claim that $\ev=(\ev_1,\ev_2)=(\ev_U,\ev_V+\ev_U)$ verifies $\Hsec\ev^T=\sv^T$ and has weight $n/2-k_U$. The procedure is described in Algorithm \ref{algo:uvsdd}.

\begin{algorithm}
	{\bf Parameter:}  a $(U,U+V)$ code of length $n$ and dimension $k = k_{U} + k_{V}$\\
	{\bf Input:} $(\sv_1,\sv_2)$ with $\sv_1 \in \mathbb{F}_{2}^{n/2-k_{U}}$, $\sv_{2} \in \mathbb{F}_{2}^{n/2- k_{V}}$\\
	{\bf Output:} $\ev \in \mathbb{F}_{2}^{n}$\\
	{\bf Assumes:} $2k_U - k_V \leq n/2$.
	
	\caption{\texttt{$UV$-$\mathbf{sddV}1$} : \textbf{$(U,U+V)-$Source Distortion Decoder} \label{algo:uvsdd}}
	\begin{algorithmic}[1]
		\State $\ev_{V} \leftarrow \Psd_{(n/2-k_{V})/2}(\Hm_{V} ,\sv_{2})$
		\State $\nu \leftarrow (n/2 - k_U - \wt{\ev_V})/2$
		\State $\ev_{U} \leftarrow \Psd_{\nu}(\Hm_{U} ,\sv_{1}, \ev_V)$
		\State \Return $(\ev_{U},\ev_{U} + \ev_{V})$
	\end{algorithmic}
\end{algorithm}

\begin{proposition} The algorithm \texttt{$UV$-$\mathbf{sddV}1$} is a fixed-$(n/2-k_{U})$ source-distortion decoder which works in polynomial average-time
  when $2k_U - k_V \leq n/2$.	
\end{proposition}
\begin{proof}
	First remark that both calls to $\Psd$ are made for a distortion level that is achieved in polynomial time. It only remains to prove that the output has the expected weight.
	We have $\wt{\ev_V}=(n/2-k_V)/2$. The word $\ev_U$ splits in two disjoint parts $\ev_U'$ whose support is $\supp(\ev_U)\cap\supp(\ev_V)$ and $\ev_U''$ whose support is $\supp(\ev_U)\setminus\supp(\ev_V)$. By construction, the second call to Prange corrects exactly $\nu\eqdef(n/2 - k_U - \wt{\ev_V})/2$ errors, this is also the weight of $\ev_U''$. Finally we can write $\ev_1=\ev_U=\ev_U'+\ev_U''$ and $\ev_2=\ev_U+\ev_V=(\ev_U'+\ev_V)+\ev_U''$ with $\supp(\ev'_U)\subset\supp(\ev_V)$ and $\supp(\ev''_U)\cap\supp(\ev_V)=\emptyset$. We derive that $\wt{\ev_1}=\wt{\ev_U'}+\wt{\ev_U''}$,  $\wt{\ev_2}=\wt{\ev_U'+\ev_V}+\wt{\ev_U''}$, and $\wt{\ev_U'+\ev_V}+\wt{\ev_U'}=\wt{\ev_V}$. And finally $\wt{\ev}=\wt{\ev_1}+\wt{\ev_2}=\wt{\ev_V}+2\nu=n/2-k_U$.\qed
\end{proof}
We can now choose the parameters $k_{U}$ and $k_{V}$ in order to minimize the distortion $n/2-k_{U}$ for a fixed dimension $k = k_{U} + k_{V}$ of the code. Let us define the relative error weight of $\ev \in \F_{2}^{n}$ as $\frac{|\ev|}{n}$. Figure \ref{fig:distOptSgn} compares the relative error weight we obtain with the algorithm $UV$-$\mathbf{sddV}1$ to $\frac{1}{n}(n-k)/2$ which corresponds to what is achieved by the generic 
decoder and to the optimal relative Gilbert-Varshamov relative weight $h^{-1}(1-R)$ where $R$ denotes the rate of the code defined as $k/n$. As we see there is a non-negligible gain. Nevertheless, $UV$-$\mathbf{sddV}1$ approximates to a fixed distance in each step of its execution which leads to correlations between some bits that can be used to recover the structure of the secret key. In order to fix this problem and as it is asked in our proof of security, we will present a modified version of $UV$-$\mathbf{sddV}1$ in \S\ref{sec:statDist} which uses a rejection sampling method to simulate uniform outputs. 
This comes at the price of slightly increasing the weight of the error output by the decoder.
\begin{figure}[htbp]
  \begin{center} 
    \caption{Comparison of the Optimal Signature Distance, the Gilbert-Varshamov Bound and Generic Distance}
    \includegraphics[scale=0.4]{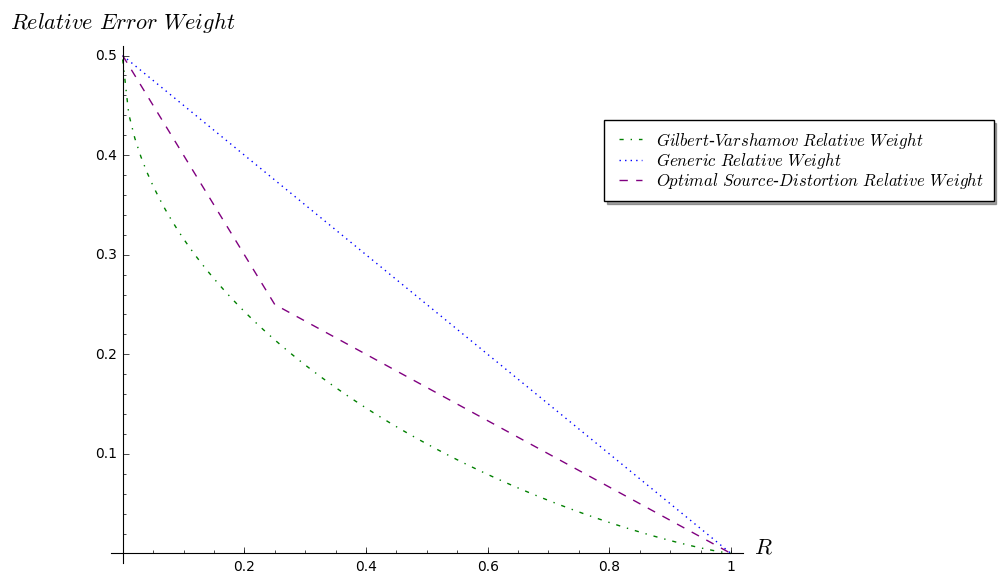}
    \label{fig:distOptSgn}
  \end{center}
\end{figure}

\subsection{$\UV$ codes and cryptography}

It is not the first time that $\UV$ codes are suggested for a cryptographic use.
 This was already considered for constructing a McEliece cryptosystem in \cite[p.225-228]{KKS05} or more recently in \cite{PMIB17} (it is namely a particular case of a generalized 
concatenated code). However both papers did not consider the improvement in the error correction 
performance that comes with the $\UV$-construction if a decoder that uses soft information is used.
For instance \cite{KKS05} studies only the case of hard-decision decoder for Goppa codes and concludes that the obtained code has a worse error correction capability than the original Goppa code and therefore worse public-key sizes. This is actually a situation which really depends on 
the code family (call it $\Fc$) and its decoder. If $\Fc$ is the family of (generalized) Reed-Solomon codes with the 
Koetter-Vardy decoder which is able to cope with soft information on the symbols, then 
the results go the other way round (at least in a certain range of rates) \cite{MT16a}.
In this case, a
$\UV$ code based on generalized Reed-Solomon codes, has for certain rates better error-correction capacity when 
decoded with the Koetter-Vardy decoder 
than a generalized Reed-Solomon of the same rate decoded with the same decoder.
This allows in principle to decrease the public key-size. Our signature scheme where $\Fc$ is the family of linear codes and
the associated (source-distortion) decoder is the Prange decoder is another example of this kind.

In order to understand what soft-decoding has to do in this setting, it is helpful to recall a few points from coding theory. 
What we are going to review here is how a $\UV$ code is decoded in the polar code construction \cite{A09} 
or in the case of Reed-Muller codes \cite{DS06}. Consider a binary linear code of length $n$ defined by a parity-check matrix $\Hm$. In hard decoding, we aim at recovering the error $\ev$ of minimum weight satisfying a given syndrome $\Hm \ev^T = \sv^T$. In soft decoding, we know all probabilities $\prob(e_i=1)$ and want to 
find the error which maximizes $\prob(\ev|\Hm \ev^T = \sv^T)$. Both decoding perform the same task in the case of a binary symmetric channel of crossover probability $p \leq \frac{1}{2}$ (this means that $\prob(e_i=1)=p$ for all $i$).

It turns out that soft-decoding is the natural scenario for decoding a $\UV$ code when we decode 
 the $V$ component first and then the $U$ component, even if one wants to perform
hard decoding of the whole $\UV$ code. This really amounts to find the error of minimum weight $\ev=(\ev_U,\ev_U+\ev_V)$ such that 
\begin{equation}
\label{eq:UVequation}
\begin{pmatrix}
\Hm_U & 0\\
\Hm_V & \Hm_V
\end{pmatrix}
\begin{pmatrix}
\ev_U^T
\\
\ev_U^T + \ev_V^T
\end{pmatrix}
= 
\begin{pmatrix}
\sv_1^T
\\
\sv_2^T
\end{pmatrix}.
\end{equation}
We will assume that the error model is a binary symmetric channel of crossover probability $p$, meaning that
\begin{eqnarray}
\prob(\ev_U(i)=1)& = & p \label{eq:prevU}\\
\prob(\ev_U(i)+\ev_V(i)=1) & = & p \label{eq:prevV}
\end{eqnarray}
where $\ev_U(i)$ and $\ev_V(i)$ denote the $i$-th coordinate of $\ev_U$ and $\ev_V$ respectively.
Recall that hard decoding $\ev$ really amounts to soft-information decoding $\ev$ with respect to this error model.
Decoding can now be done through the following steps.

\noindent
{\bf Step 1.}
We observe that \eqref{eq:UVequation} implies that $\Hm_V \ev_V^T = \sv_2^T$. Recovering the $V$ component amounts here 
to hard decode $\ev_V$. The rationale behind this is that the channel model of $\ev_V(i)$ is a binary symmetric channel of crossover probability $2p(1-p)$. This can be verified by observing that
$\ev_V = \ev_U+(\ev_U + \ev_V)$ with $\ev_U$ and $\ev_U+\ev_V$ being independent random variables 
whose components are i.i.d. with probability distributions given by \eqref{eq:prevU} and \eqref{eq:prevV}.
From this we deduce that the components $\ev_V(i)$ are i.i.d. with 
$\prob(\ev_V(i)=1)=2p(1-p)$.
Let us now assume that we have decoded $\ev_V$ correctly.

\noindent
{\bf Step 2.} Recovering $\ev_U$ can in principle be done in two different ways.
The first one uses \eqref{eq:UVequation} directly from which we deduce
\begin{equation}
\label{eq:u1}
\Hm_U \ev_U^T  =  \sv_1^T.
\end{equation}
 Now that we know $\ev_V$ we could also notice that
\begin{equation}
\label{eq:u2}
\Hm_U(\ev_U+\ev_V)^T = \Hm_U \ev_U^T + \Hm_U\ev_V^T= \sv_1^T + \Hm_U \ev_V^T
\end{equation}
and we know here the right-hand term. 
There are therefore two ways to recover $\ev_U$
\begin{itemize}
\item[Method 1] We perform  hard decoding of the 
syndrome $\sv_1^T$ and find the $\ev_U$ of minimum weight satisfying \eqref{eq:u1}.
\item[Method 2] We perform
hard decoding of the syndrome $\sv_1^T + \Hm_U \ev_V^T$ by finding the vector $\xv$ of minimum weight
satisfying $\Hm_U \xv^T = \sv_1^T + \Hm_U \ev_V^T$ and let $\ev_U = \xv + \ev_V$.
\end{itemize}
In \cite{KKS05} it is suggested to perform both decodings and to choose for computing $\ev_U$ the
decoding which gives the smallest error weight (the decoding is not explained in terms of syndromes
there, but expressing their decoding in terms of syndrome decoding amounts to the decision rule that we have just given).
It is clear that some amount of information is lost during this process. This can be seen by noticing that once we know 
$\ev_V$, we have a much finer knowledge on $\ev_U$. It is readily seen that we can now use $\prob(\ev_U(i)=1|\ev_V(i))$ instead of 
$\prob(\ev_U(i)=1)$. This calculation follows from the fact that $\ev_U(i)$ and $\ev_U(i)+\ev_V(i)$ are independent and we know 
the distribution of these two random variables. A straightforward calculation leads to
\begin{eqnarray}
\prob(\ev_U(i)=1|\ev_V(i)=0) & = & \frac{p^2}{(1-p)^2+p^2} 
\label{eq:probability1} \\
\prob(\ev_U(i)=1|\ev_V(i)=1) & = & \frac{1}{2}.
\label{eq:probability2}
\end{eqnarray}
In other words, when $\ev_V(i)=0$, we can view $\ev_U(i)$ as an error originating from a binary symmetric channel
of crossover probability  $\frac{p^2}{(1-p)^2+p^2}$ 
(which is much smaller than $p$) and when $\ev_V(i) = 1$ we may consider that the position has just been erased.
 A decoder for $U$ which uses this soft information  has potentially 
much better performances than the previous hard decoder. In fact, in this case we just need a decoder which decodes 
errors and erasures. When the alphabet is non binary, the channel model is  slightly more complicated:
this is why the Koetter-Vardy soft decoder is used in \cite{MT16a} and not just an error and erasure 
decoder of generalized Reed-Solomon codes.
By using the noise model corresponding to the probability computations \eqref{eq:probability1} and \eqref{eq:probability2}  we obtain a much less noisy model than the original 
binary symmetric channel. This can be checked by a capacity calculation which is in a sense
a measure of the noise of transmission channel (the capacity is a decreasing function of the noise
level in some sense). 
The capacity of the binary symmetric channel of crossover probability $p$ is $1-h(p)$ whereas
it is $1 - 2h(p)+h(2p(1-p))$ for the noise model corresponding to the probability computations \eqref{eq:probability1} and \eqref{eq:probability2}. 
We have represented these two capacities in Figure \ref{fig:capacities} 
and it can be verified there that the new noise model has a much larger capacity than the original channel.

\begin{figure}
\caption{Capacity of the original binary symmetric channel vs. capacity of the channel
model corresponding to the probability computations \eqref{eq:probability1} and \eqref{eq:probability2}. \label{fig:capacities}} 
\centering
\includegraphics[height=6cm]{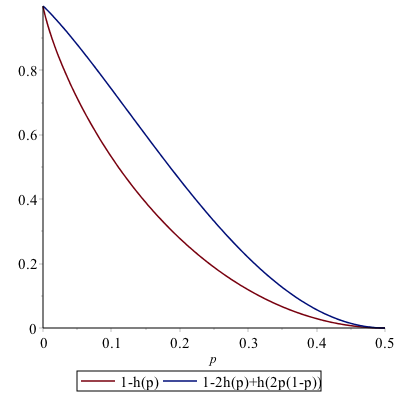}
\end{figure}

This discussion explains why in the binary setting we would really like to use a decoder for $U$
which is able to correct errors on the positions where $\ev_V(i)=0$ and erasures on the positions
where $\ev_V(i)=1$. In our context where we perform source-distortion decoding the situation is actually similar.
 Our strategy works here because the Prange decoder has a natural and powerful extension to the error/erasure scenario.
It is natural to expect that a family of codes and associated decoders which are powerful in the erasure/error scenario behave better 
when used in a $\UV$ construction and decoded as above, than the original family of codes. Our strategy for obtaining a signature scheme really builds 
upon this approach : a $\UV$ code decoded as above with the Prange decoder has better distortion than the Prange 
decoder used directly on a linear code with the same length and dimension as the $\UV$-code. The trapdoor here for obtaining the better
distortion is only the $\UV$ structure, but we can afford to have random linear codes for $U$ and $V$.

 \section{Security Proof}
\label{sec:securityProof} 

We give in this section a security proof of the signature scheme $\cS_{\text{code}}$.
This proof is in the spirit of the security proof of the FDH
signatures in the random oracle model (see \cite{BR93}). However in order to have a tight security reduction we were inspired by the proof of
\cite{C02}. Our main result is to reduce the security to two major
problems in code-based cryptography.

\subsection{Basic tools}

\subsubsection{Basic definitions.}

A function $f(n)$ is said to be negligible if for all polynomials $p(n)$, $|f(n)| < p(n)^{-1}$ for all sufficiently large $n$. The statistical distance between two discrete probability distributions over a same space $\mathcal{E}$ is defined as:
$$
    \rho(\cD^0,\cD^1) \eqdef \frac{1}{2} \sum_{x \in \mathcal{E}} |\cD^0(x)-\cD^1(x) |.
$$

We will need the following well known property for the statistical distance which can be easily proved by induction.
\begin{proposition}
  \label{prop:product}
  Let $(\cD^0_1,\dots,\cD^0_n)$ and $(\cD^1_1,\dots,\cD^1_n)$ be two $n$-tuples of discrete probability distributions where 
  $\cD^0_i$ and $\cD^1_i$ are distributed over a same space $\cE_i$. 
We have for all positive integers $n$:
  \begin{displaymath}
    \rho\left(\cD^0_1 \otimes \dots \otimes \cD^0_n,\cD^1_1 \otimes \dots \otimes \cD^1_n \right) \leq \sum_{i=1}^n \rho(\cD^0_i,\cD^1_i).
  \end{displaymath}
\end{proposition}

  A {\em distinguisher} between two distributions $\mathcal{D}^{0}$ and
$\mathcal{D}^{1}$ over the same space $\mathcal{E}$ is a randomized
algorithm which takes as input an element of $\mathcal{E}$ that
follows the distribution $\mathcal{D}^{0}$ or $\mathcal{D}^{1}$ and
outputs $b \in \{0,1\}$. It is  characterized by its advantage:
\begin{displaymath}
Adv^{\mathcal{D}^{0},\mathcal{D}^{1}}(\cA) \eqdef
\mathbb{P}_{\xi \sim \mathcal{D}^{0}}\left( \cA(\xi) \mbox{
	outputs } 1 \right) - \mathbb{P}_{\xi \sim
	\mathcal{D}^{1}}\left(\cA(\xi) \mbox{ outputs } 1
\right).
\end{displaymath}

We call
this quantity the {\em advantage} of $\cA$ against
$\mathcal{D}^{0}$ and $\mathcal{D}^{1}$.

\begin{definition}
  [Computational Distance and Indistinguishability]
  The computational distance between two distributions $\mathcal{D}^{0}$ and $\mathcal{D}^{1}$ in time $t$ is: 
  \begin{displaymath}
    \rho_{c}\left( \mathcal{D}^{0},\mathcal{D}^{1}\right)(t) \eqdef
    \mathop{\max}\limits_{ |\cA| \leq t} \left\{
      Adv^{\mathcal{D}^{0},\mathcal{D}^{1}}(\cA) \right\}
  \end{displaymath}
  where $|\cA|$ denotes the running time of $\cA$ on
  its inputs.
  
  The ensembles $\mathcal{D}^{0}=(\mathcal{D}^{0}_{n})$ and
  $\mathcal{D}^{1} = (\mathcal{D}_{n}^{1})$ are computationally
  indistinguishable in time $(t_{n})$ if their computational distance
  in time $(t_{n})$ is negligible in $n$.
\end{definition}
In other words, the computational distance is the best advantage that any adversary could get in bounded time. 

\subsubsection{Digital signature and games.}
Let us recall the concept of signature schemes, the security model
that will be considered in the following and to recall in this context
the paradigm of games in which we give a security proof of
our scheme.

\begin{definition}
  [Signature Scheme]A signature scheme $\cS$ is a triple of algorithms
  $\Gen$, $\Sgn$, and $\Ver$ which are
  defined as:
  \begin{itemize}
  \item The key generation algorithm $\Gen$ is a probabilistic
    algorithm which given $1^{\lambda}$, where $\lambda$ is the
    security parameter, outputs a pair of matching public and private
    keys $(pk,sk)$;
  \item The signing algorithm is probabilistic and takes as input a
    message $\mv \in \{0,1\}^{*}$ to be signed and returns a signature
    $\sigma = \Sgnsk(\mv)$;
  \item The verification algorithm takes as input a message $\mv$ and
    a signature $\sigma$. It returns $\Vrfypk(\mv,\sigma)$
    which is $1$ if the signature is accepted and $0$ otherwise. It
    is required that \ $\Vrfypk(\mv,\sigma)=1$ if  $\sigma = \Sgnsk(\mv)$.
  \end{itemize}
\end{definition}
For this kind of scheme, one of the strongest security notion is {\em
  existential unforgeability under an adaptive chosen message attack}
(EUF-CMA). In this model the adversary has access to all signatures
of its choice and its goal is to produce a valid forgery. A valid
forgery is a message/signature pair $(\mv,\sigma)$ such that
$\Vrfypk(\mv,\sigma)=1$ whereas the signature of $\mv$ has never been
requested by the forger. More precisely, the following definition
gives the EUF-CMA security of a signature scheme:

\begin{definition}
  [EUF-CMA Security] Let $\cS$ be a signature scheme.\\  A forger $\cA$
  is a $(t,\qhash,\qsig,\varepsilon)$-adversary in \textup{EUF-CMA} against
  $\cS$ if after at most $\qhash$ queries to the hash oracle, $\qsig$
  signatures queries and $t$ working time, it outputs a valid forgery
  with probability at least $\varepsilon$.
  We define the \textup{EUF-CMA} success probability against $\cS$ as:
  \begin{displaymath}
    Succ_{\cS }^{\textup{EUF-CMA}}(t,\qhash,\qsig) \eqdef
    \max \left( \varepsilon \mbox{} | \mbox{it exists a }
      (t,\qhash,\qsig,\varepsilon) \mbox{-adversary} \right).
  \end{displaymath}
  The signature scheme $\cS$ is said to be
  $(t,\qhash,\qsig)$-secure in \textup{EUF-CMA} if the above success
  probability is a negligible function of the security parameter
  $\lambda$.
  
\end{definition}

\subsubsection{The game associated to our code-based signature scheme.}
The modern approach to prove the security of cryptographic schemes is to relate the security of its
primitives to well-known problems that are believed to be hard by
proving that breaking the cryptographic primitives provides a mean
to break one of these hard problems.
In our case, the security of the signature scheme is defined as a game with an adversary that has access to hash and sign oracles.
It will be helpful here to be more formal and to define more precisely the games we will consider. They are  games between two players, an
{\em adversary} and a {\em challenger}. In a game $G$, the challenger executes three kind of procedures:
\begin{itemize}
\item an initialization procedure {\tt Initialize} which is called
  once at the beginning of the game.
\item oracle procedures which can be requested at the will of the
  adversary. In our case, there will be two, {\tt Hash} and {\tt
    Sign}. The adversary $\cA$ which is an algorithm may call {\tt
    Hash} at most $\qhash$ times and {\tt Sign} at most $\qsig$ times.
\item a final procedure {\tt Finalize} which is executed once $\cA$
  has terminated. The output of $\cA$ is given as input to this
  procedure.
\end{itemize}
The output of the game $G$, which is denoted $G(\cA)$, is the output
of the finalization procedure (which is a bit $b \in \{0,1\}$). The
game $G$ with $\cA$ is said to be successful if $G(\cA)=1$. The
standard approach for obtaining a security proof in a certain model is
to construct a sequence of games such that the success of the first
game with an adversary $\cA$ is exactly the success against the model
of security, the difference of the probability of success between two
consecutive games is negligible until the final game where the
probability of success is the probability for $\cA$ to break one of
the problems which is supposed to be hard. In this way, no adversary
can break the claim of security with non-negligible success unless it
breaks one of the problems that are supposed to be hard.

\begin{definition}[challenger procedures in the EUF-CMA Game]
  The challenger procedures for the \textup{EUF-CMA} Game corresponding to
  $\cS_{\text{code}}$ are defined as:
  \begin{center} \tt
    \begin{tabular}{|l|l|l|l|}
      \hline
      \underline{proc Initialize$(\lambda)$} & \underline{proc Hash$(\mv,\rv)$} & \underline{proc Sign$(\mv)$} & \underline{proc Finalize$(\mv,\ev,\rv)$} \\ 
      $(pk,sk) \leftarrow \Gen(1^{\lambda})$ & {return} $\hash(\mv,\rv)$ & $\rv \Unif \{0,1\}^{\lambda_{0}}$ & $\sv \leftarrow {\tt Hash}(\mv,\rv)$ \\
      $\Hpub\leftarrow pk$ & & $\sv \leftarrow$ \texttt{Hash}$(\mv,\rv)$  & return \\
      $(\Hsec,\Pm,\Sm) \leftarrow sk$ & & $\ev \leftarrow D_{\Hsec,w}(\Sm^{-1}\sv^{T})$ & $\Hpub \ev^{T} = \sv^{T} \wedge \wt{\ev} = w$ \\
      return $\Hpub$ & & return $(\ev\Pm,\rv)$ &   \\
      \hline  
    \end{tabular}
  \end{center}
\end{definition}
\subsection{Code-Based Problems}
\label{subsec:cbProb} 
We introduce in this subsection the code-based problems that will be
used in the security proof. The first is Decoding One Out of
Many (DOOM) which was first considered in \cite{JJ02} and later
analyzed in \cite{S11}. We will come back to the best known algorithms
to solve this problem as a function of the distance $w$ in
\S\ref{sec:attsourcedistortion}.

\begin{problem}[\textup{DOOM} -- Decoding One Out of Many]~\\
  \label{prob:DOOM}
  \begin{tabular}{ll}
    Instance: & $\quad\Hm \in \mathbb{F}_{2}^{(n-k) \times n}$,
                $\sv_{1},\cdots,\sv_{q} \in \mathbb{F}_{2}^{n-k}$, $w$
                integer \\
    Output: & $\quad(\ev,i) \in \mathbb{F}_{2}^{n} \times \llbracket 1,q \rrbracket$ such that $|\ev| = w$ and $\Hm\ev^{T} = \sv_{i}^{T}$.
  \end{tabular}
\end{problem}

\begin{definition}[One-Wayness of DOOM]
  We define the success of an algorithm $\cA$ against \DOOM\ with the parameters $n,k,q,w$ as:
  \begin{align*} 
    Succ_{\DOOM}^{n,k,q,w}\left( \cA \right) = \prob \big( \cA&\left( \Hm,\sv_{1},\cdots,\sv_{q} \right) \mbox{solution  of } \DOOM \big)
  \end{align*} 
  where $\Hm$ is chosen uniformly at random in $\F_2^{(n-k)\times n}$, the $\sv_i$'s are chosen uniformly at random in $\F_2^{n-k}$ and
the probability is taken over these choices of $\Hm$, the $\sv_i$'s and the internal coins of $\cA$.

  The computational success in time $t$ of breaking \DOOM\ with the parameters $n,k,q,w$ is then defined as:
  \begin{displaymath}
    Succ_{\DOOM}^{n,k,q,w}(t) = \mathop{\max}\limits_{|\cA|\leq t} \left\{
      Succ_{\DOOM}^{n,k,q,w}\left( \cA \right) \right\}.
  \end{displaymath}
\end{definition}

Another problem will appear in the security proof: distinguish
random codes from a code drawn uniformly at random in the family used for public keys
in the signature scheme. 

\begin{remark} 
We will show in \S\ref{ss:NPcomplete} (see Theorem \ref{theo:UVNP}) that the associated decision problem is NP-complete.
\end{remark}

\noindent We will denote in the rest of the article by $\Hpub$ the random matrix chosen as the public parity-check matrix of our scheme.
Let us recall that it is obtained as
\begin{equation}\label{eq:def_Hpub}
\Hpub = \Sm \Hsec \Pm
 \mbox{ with } \Hsec = \begin{pmatrix}
 \Hm_{U} & \mathbf{0} \\
 \Hm_{V} & \Hm_{V} 
 \end{pmatrix},
\end{equation} 
where $\Sm$ is chosen uniformly at random among the invertible binary matrices of size $(n-k)\times (n-k)$, $\Hm_U$ is chosen uniformly at random among the binary matrices
of size $(n/2-k_U)\times n/2$, $\Hm_V$ is chosen uniformly at random among the binary matrices
of size $(n/2-k_V)\times n/2$ and $\Pm$ is chosen uniformly at random among the permutation matrices of size $n \times n$.
The distribution of the random variable $\Hpub$ is denoted by $\Dpub$.
On the other hand
$\Drand$ will denote the uniform distribution over the parity-check matrices of all $\lbrack n,k\rbrack$-codes with $k = k_{U} + k_{V}$.

We will discuss about the difficulty of the task to distinguish $\Dpub$
and $\Drand$ in \S\ref{sec:keyAttack}. 
It should be noted that the syndromes associated to matrices $\Hpub$ are indistinguishable in a very strong sense from random syndromes 
as the following proposition shows

\begin{restatable}{proposition}{propoDist}
	\label{prop:statDist} 
	Let $\Dsw{\Hm}$ be the distribution of the syndromes $\Hm \ev^T$ when $\ev$ is drawn uniformly at random among the binary vectors of weight $w$ and $\Uc$ be the uniform distribution over the syndrome space $\F_2^{n-k}$.
We have
	$$
	\esp_{\Hpub} \left( \rho(\Dsw{\Hpub}, \Uc) \right) \leq \frac{1}{2} \sqrt{\varepsilon}
	$$
	with
	$$
\varepsilon = \frac{2^{n-k}}{\binom{n}{w}} + \frac{2^{n/2-k_U} \binom{n/2}{w/2}}{\binom{n}{w}}
+ \sum_{\substack { j \in \{0,\dots,w\} \\ j \equiv w \pmod{2} }} \frac{ 2^{2j+n/2-k_V}\binom{n/2}{(w-j)/2}^2 \binom{n/2-(w-j)/2}{j}}{ \binom{n/2}{j} \binom{n}{w}^2}^2.
		$$
	
\end{restatable}

\begin{remark}
	In the paradigm of code-based signatures we have $w$ greater than
	the Gilbert-Varshamov bound, which gives $2^{n-k} \ll \binom{n}{w}$ and for the set of parameters we present in \S\ref{sec:parameters}, $\varepsilon \lll \frac{1}{2^{\lambda}}$ with $\lambda$ the security parameter.
\end{remark}

\subsection{EUF-CMA Security Proof} 
\label{sec:securityProof3}
This subsection is devoted to our main theorem and its proof. Let us first introduce some notation that will be used. 
We will denote by $\mathcal{D}_{w}$ the distribution $
\left\{ D_{\Hsec,w}(\sv) : \sv \Unif \mathbb{F}_{2}^{n-k} \right\}
$ where $D_{\Hsec,w}$ is the source-distortion decoder used in the signature scheme. Recall that $\cU_{w}$ is the uniform distribution over $S_{w}$ (which is the
set of words of weight $w$ in $\mathbb{F}_{2}^{n}$), $\Dpub$ is the distribution of public keys, $\Drand$ is the uniform distribution over parity-check matrices of all $\lbrack n,k \rbrack$-codes and $\cS_{\textup{code}}$ is our signature scheme defined in \S\ref{sec:UVsgnScheme1} with the family of $(U,U+V)$ codes.

\begin{theorem}[Security Reduction]
  \label{theo:secRedu}
  Let $\qhash$ (resp. $\qsig$) be the number of queries to the hash
  (resp. signing) oracle. We assume that
  $\lambda_{0} = \lambda + 2\log_{2}(\qsig)$  where $\lambda$ is the security parameter of the signature scheme. We have in the random oracle model \textup{(ROM)} for all time $t$:
  \begin{multline*}
    Succ_{\cS_{\textup{code}}}^{\textup{EUF-CMA}}(t,\qhash,\qsig) \leq
    2 Succ_{\DOOM}^{n,k,\qhash,w}(t_{c}) + \frac{1}{2}\qhash\sqrt{ \varepsilon } \\ + \qsig \rho\left( \mathcal{D}_{w},\mathcal{U}_{w} \right) + \rho_{c} \left( \Drand,\Dpub \right)(t_{c}) + \frac{1}{2^{\lambda}}
  \end{multline*}
  where $t_{c} = t + O \left( \qhash  \cdot n^{2} \right)$ and $\varepsilon$ given in Proposition \ref{prop:statDist}.
\end{theorem}

\begin{proof}
  Let $\cA$ be a $(t,\qsig,\qhash,\varepsilon)$-adversary
  in the EUF-CMA model against $\cS_{\text{code}}$ and let  $(\Hm_{0},\sv_{1},\cdots,\sv_{\qhash})$  be drawn uniformly at random among all instances of $\DOOM$ for parameters $n,k,\qhash,w$. We stress here that syndromes $\sv_{j}$ are random and independent vectors of $\mathbb{F}_{2}^{n-k}$. We write
  $\mathbb{P}\left( S_{i} \right)$ to denote the probability of
  success for $\cA$ of game $G_{i}$. Let
      \bigskip
  
  {\bf Game $0$} is the EUF-CMA game for $\cS_{\textup{code}}$.
  \bigskip

  {\bf Game $1$} is identical to Game $0$ unless the following failure event $F$ occurs: there is a collision in a signature query ({\em i.e.} two signatures queries for a same message $\mv$ lead to the same salt $\rv$). By using the difference lemma
  (see for instance \cite[Lemma 1]{S04a}) we get:
  \begin{displaymath}
    \mathbb{P}\left( S_{0} \right) \leq \mathbb{P}\left( S_{1} \right) + \mathbb{P} \left( F \right). 
  \end{displaymath}

  The following lemma (see \ref{lemm:failEvent} for a proof) shows that in our case as $\lambda_{0} = \lambda + 2\log_{2}(\qsig)$, the probability of the event $F$ is negligible.
  \begin{lemma}
    \label{lemm:fEvent}
    For $\lambda_{0} = \lambda + 2\log_{2}(\qsig)$ we have:
    \begin{displaymath}
      \mathbb{P}\left( F \right) \leq \frac{1}{2^{\lambda}}.
    \end{displaymath}
  \end{lemma}

  {\bf Game $2$} is modified from Game $1$ as follows:\smallskip
  \begin{center}{\tt
    \begin{tabular}{|l|l|}
      \hline
      \underline{proc Hash$(\mv ,\rv)$} & \underline{proc Sign$(\mv )$} \\
       if $\rv \in \listM$ & $\rv\gets \listM$.next$()$ \\
      \quad $\ev_{\mv ,\rv}\Unif S_w$ & $\sv\gets$ Hash$(\mv ,\rv)$ \\
      \quad return $\ev_{\mv ,\rv}\Hm^{T}_{\textup{pub}}$ & $\ev \gets D_{\Hsec,w}(\Sm^{-1}\sv^{T})$ \\
      else & return $\left(\ev\Pm,\rv\right)$ \\
      \quad $j\gets j+1$ & \\
 	  \quad return $\sv_j$ & \\
      \hline
    \end{tabular}}
	\newlength{\mylen}\settowidth{\mylen}{\tt
		    \begin{tabular}{|l|l|}
			\hline
			\underline{proc Hash$(\mv ,\rv)$} & \underline{proc Sign$(\mv )$} \\
			if $\rv \in \listM$ & $\rv\gets \listM$.next$()$ \\
			\quad $\ev_{\mv ,\rv}\Unif S_w$ & $\sv\gets$ Hash$(\mv ,\rv)$ \\
			\quad return $\ev_{\mv ,\rv}\Hm^{T}_{\textup{pub}}$ & $\ev \gets D_{\Hsec,w}(\Sm^{-1}\sv^{T})$ \\
			else & return $\left(\ev\Pm,\rv\right)$ \\
			\quad $j\gets j+1$ & \\
			\quad return $\sv_j$ & \\
			\hline
		\end{tabular}}
	\addtolength{\mylen}{-2\mylen}
	\addtolength{\mylen}{\linewidth}
	\addtolength{\mylen}{-15pt}\hfill
	\begin{minipage}{\mylen}
	 To each message $\mv $ we associate a list $\listM$ containing  $\qsig$
	random elements of $\F_{2}^{\lambda_0}$. It is constructed the
	first time it is needed.  The call $\rv \in \listM$
	returns true if and only if $\rv$ is in the list. The call
	$\listM.{\tt next}()$ returns elements of $\listM$
	sequentially. The list is large enough to satisfy all queries.
	\end{minipage}
  \end{center}
	The {\tt Hash} procedure now creates the list $\listM$ if needed, then, if $\rv\in\listM$ it returns $\ev_{\mv,\rv}\Hm^{T}_{\textup{pub}}$ with $\ev_{\mv,\rv} \Unif S_{w}$. This leads to a valid signature $(\ev_{\mv,\rv},\rv)$ for $\mv$. The error value is stored. If $\rv\not\in\listM$ it outputs one of $\sv_j$ of
the instance $(\Hm_{0},\sv_{1},\ldots,\sv_{\qhash})$ of the {DOOM}
problem. The {\tt Sign} procedure is unchanged, except for $\rv$
which is now taken in $\listM$.  The global index $j$ is set to 0 in {\tt
	proc Initialize}.

  We can relate this game to the previous one through the following lemma.
  \begin{restatable}{lemma}{lemdistribi}
    \label{lem:distribi}
    $$
      \mathbb{P}(S_{1})\leq \mathbb{P}(S_{2}) + \frac{\qhash}{2}\sqrt{\varepsilon} \mbox{ where } \varepsilon \mbox{ is given in Proposition \ref{prop:statDist}.} 
    $$
  \end{restatable}
  The proof of this lemma is given in Appendix \ref{lem:distrib} and relies
  among other things on the following points:
  \begin{itemize} 
  \item Proposition \ref{prop:product};
  \item Syndromes produced by matrices $\Hm_{\text{pub}}$ with errors of weight $w$ have average statistical distance from 
the uniform distribution over $\mathbb{F}_{2}^{n-k}$  at most $\frac{1}{2} \sqrt{\varepsilon}$ (see Proposition \ref{prop:statDist}). This follows from 
a lemma which is a variation of the leftover hash lemma  (see \cite{BDKPPS11}) and which can be expressed as follows.
\item \begin{restatable}{lemma}{lemleftoverHash}\label{lem:leftoverHash}
Consider a finite family $\Hc = (h_i)_{i \in I}$ of functions from a finite set $E$ to a finite set $F$.
Denote by $\varepsilon$ the bias of the collision probability, i.e. the quantity such that
$$
\prob_{h,e,e'}(h(e)=h(e')) = \frac{1}{|F|} (1 + \varepsilon)
$$
where $h$ is drawn uniformly at random in $\Hc$, $e$ and $e'$ are drawn uniformly at random in $E$. Let $\Uc$ be the uniform distribution over $F$ and $\Dc(h)$ be the
distribution of the outputs $h(e)$ when $e$ is chosen uniformly at random in $E$.
We have
$$
\esp_h \left\{ \rho(\Dc(h),\Uc) \right\} \leq \frac{1}{2} \sqrt{\varepsilon}.
$$
\end{restatable}

\begin{remark}
In the leftover hash lemma, there is the additional assumption that $\Hc$ is a universal family of hash functions, meaning that 
for any $e$ and $e'$ distinct in $F$, we have $\prob_h(h(e)=h(e'))=\frac{1}{|F|}$. This assumption allows to have a general bound on the bias $\varepsilon$. In our case, where the $h$'s are hash functions defined as
$h(e) = \Hpub \ev^T$, $\Hc$ does not form a universal family of hash functions (essentially because the distribution of the $\Hpub$'s 
is not the uniform distribution over $\F_2^{(n-k)\times n}$). However in our case we can still bound $\varepsilon$ by a direct computation.
This lemma is proved in Appendix \S\ref{lem:distrib}.
\end{remark}

  \end{itemize}

  {\bf Game $3$} differs from Game $2$ by changing in {\tt proc
    Sign} calls ``$\ev \gets D_{\Hsec,w}(\Sm^{-1}\sv^{T})$'' by
  ``$ \ev \gets \ev_{\mv ,\rv}$'' and ``return $(\ev\Pm,\rv)$'' by
  ``return $(\ev,\rv)$''. 
  Any signature $(\ev,\rv)$ produced by
  {\tt proc Sign} is valid.  
  The error $\ev$ is drawn according to
  the uniform distribution $\mathcal{U}_{w}$ while previously it was
  drawn according to the source distortion decoder distribution, that
  is $\mathcal{D}_{w}$.
  By using Proposition \ref{prop:product} it follows that
  \begin{displaymath}
    \mathbb{P} \left( S_{2} \right) \leq \mathbb{P} \left( S_{3}
    \right) + \qsig \rho\left( \mathcal{U}_{w},\mathcal{D}_{w} \right).
  \end{displaymath}

  {\bf Game $4$} is the game where we replace the public matrix
  $\Hpub$ by $\Hm_{0}$. 
  In this way we will force the adversary to
  build a solution of the \DOOM\ problem. Here
  if a difference is detected between games it gives a
  distinguisher between distributions $\Drand$
  and $\Dpub$:
  \begin{displaymath}
    \mathbb{P} \left( S_{3} \right) \leq \mathbb{P} \left( S_{4} \right) + \rho_{c} \left( \Dpub,\Drand \right)\left(t_{c} \right).   
  \end{displaymath}

     We show in appendix how to emulate the lists $\listM$ in such a way
   that list operations cost, including its construction, is at most
   linear in the security parameter $\lambda$. Since $\lambda\le n$, it
   follows that the cost to a call to {\tt proc Hash} cannot exceed
   $O(n^2)$ and the running time of the challenger is
   $t_{c} = t +  O\left( \qhash \cdot n^{2} \right)$.
   \bigskip

  {\bf Game $5$} differs in the finalize procedure.
  \begin{center} {\tt
      \begin{tabular}{|l|}
        \hline
        \underline{proc Finalize}$(\mv ,\ev,\rv)$ \\
        $\sv\gets$ Hash$(\mv,\rv)$ \\
        $b \leftarrow \Hm_{\text{pub}}\ev^{T} = \sv^{T} \wedge |\ev| = w$ \\  
        return $b \wedge \rv \notin \listM$ \\
        \hline
      \end{tabular}}
	\settowidth{\mylen}{\tt
\begin{tabular}{|l|}
        \hline
        \underline{proc Finalize}$(\mv ,\ev,\rv)$ \\
        $\sv\gets$ Hash$(\mv ,\rv)$ \\
        $b \leftarrow \Hm_{\text{pub}}\ev^{T} = \sv^{T} \wedge |\ev| = w$ \\  
        return $b \wedge \rv \notin \listM$ \\
        \hline
      \end{tabular}}
    \addtolength{\mylen}{-2\mylen}
    \addtolength{\mylen}{\linewidth}
    \addtolength{\mylen}{-15pt}\hfill
    \begin{minipage}{\mylen}
      We assume the forger outputs a valid signature $(\ev,\rv)$ for
      the message $\mv $. The probability of success of Game $5$ is the
      probability of the event ``$S_4 \wedge(\rv\not\in\listM)$''.
    \end{minipage}
  \end{center}

  If the forgery is valid, the message $\mv $ has never been queried by
  {\tt Sign}, and the adversary never had access to any
  element of the list $\listM$. This way, the two events are
  independent and we get:
  \begin{displaymath}
    \mathbb{P} \left( S_{5} \right) = (1 - 2^{-\lambda_{0}})^{\qsig} \mathbb{P} \left( S_{4} \right).
  \end{displaymath}
  As we assumed $\lambda_{0}= \lambda + 2\log_{2}(\qsig) \geq  \log_{2}(\qsig^{2})$, we have:
  \begin{displaymath}
    \left( 1 - 2^{-\lambda_{0}} \right)^{\qsig} \geq \left( 1 - \frac{1}{\qsig^{2}}\right)^{\qsig} \geq \frac{1}{2}.
  \end{displaymath}
  Therefore
  \begin{equation}\label{eq:lower_bound_p5}
    \mathbb{P}\left( S_{5} \right) \geq \frac{1}{2}\mathbb{P}\left( S_{4} \right). 
  \end{equation}
  The probability $\mathbb{P} \left( S_{5} \right)$ is then exactly the probability for $\cA$ to output $\ev_{j} \in S_{w}$ such that $\Hm_{0}\ev^{T}_{j} = \sv_{j}^{T}$ for some $j$ which gives
  \begin{align}\label{eq:upper_bound_p5}
    \mathbb{P} \left( S_{5} \right) \leq\ Succ_{\DOOM}^{n,k,\qhash,w}(t_{c}).
  \end{align} 
  \eqref{eq:lower_bound_p5} together with \eqref{eq:upper_bound_p5} imply that 
  \begin{displaymath}
    \prob(S_4) \leq 2\cdot \ Succ_{\DOOM}^{n,k,\qhash,w}(t_{c}).
  \end{displaymath}
  This  concludes the proof of Theorem \ref{theo:secRedu} by combining this together with all the bounds obtained for each of the previous games.
\end{proof} 

 \section{Achieving the Uniform Distribution of the Outputs}
\label{sec:statDist} 

\subsection{Rejection Sampling Method} 

In our security proof, we use the fact that the distribution of the
outputs of the $(U,U+V)$ decoder is close to the uniform distribution
on the words of weight $w$. We will show how to modify a little bit
the decoder by performing some moderate rejection sampling in order to
meet this property. Note that ensuring such a property is actually not
only desirable for the security proof, it is also more or less
necessary since there is an easy way to attack the signature when it
is based on the decoder $UV$-$\mathbf{sddV}1$. Indeed, it is readily
verified that with this decoder the probability $\prob(e_i=1,e_j=1)$
we have on the output $\ev$ of the decoder for certain $i$ and $j$ is
larger than the same probability for a random word $\ev$ of weight
$w$.  The pairs $(i,j)$ which have this property correspond to the
image by the permutation $\Pm$ of pairs of the form $(x,x+n/2)$ or
$(x+n/2,x)$.  In other words, signatures leak information in this case
and this can be used to recover completely the permuted $(U,U+V)$
structure of the code.

To explain the rejection method, let us introduce some notation.
Let $\ev \in \mathbb{F}_{2}^{n}$,
\begin{eqnarray*}
	w_{1}(\ev) & \eqdef & \left|\left\{ i \in \{1, \cdots, n/2\} \mbox{ } : \mbox{ } e_{i} \neq e_{i+n/2} \right\}\right|, \\
	w_{2}(\ev) & \eqdef & \left|\left\{ i \in \{1, \cdots, n/2\}  \mbox{ } : \mbox{ } e_{i} = e_{i+n/2} = 1 \right\}\right|.
\end{eqnarray*}

The problem is that algorithm $UV$-$\mathbf{sddV}1$ outputs errors $\ev$ for which $w_1(\ev)$ and
$w_2(\ev)$ are constant:
$w_1(\ev) = (n/2-k_V)/2$ and
$w_2(\ev)=(n/2-k_U-w_1(\ev))/2 = n/8-k_U/2+k_V/4$.
Obviously uniformly distributed errors $\ev$ in $S_w$ do not have this behavior.
Our strategy to attain this uniform distribution on the outputs $\ev$ is to change a little bit 
the source-distortion decoder for $V$ in order to attain variable weight errors 
which are such that the weight of $\ev_V$ (which corresponds to $w_1(\ev)$)  have
the same distribution as $w_1(\ev')$ where $\ev'$ is a random error of weight $w$ which is
uniformly distributed. This can be easily done by rejection sampling as in 
Algorithm \ref{alg:2algsgn}. Recall that $D$ is a Source Distortion Decoder (see Definition \ref{def:sddeco} in \S\ref{subsec:2.2}).

\begin{algorithm} 
	{\bf Parameter:}  a $(U,U+V)$ code of length $n$ \\
	{\bf Inputs:} $\cdot$ $(\sv_1,\sv_2)$ with $\sv_1 \in \mathbb{F}_{2}^{n/2-k_{U}}$, $\sv_{2} \in \mathbb{F}_{2}^{n/2 - k_{V}}$\\
	$\cdot$ no-rejection probability vector $\xv=(x_i)_{0 \leq i \leq n-k_V} \in [0,1]^{n-k_V}$\\
	{\bf Output:} $\ev \in \mathbb{F}_{2}^n$ with $|\ev|=w$.\\
	{\bf Assumes:} $2k_U - k_V \leq n/2$.
	\caption{$UV$-$\mathbf{sddV}2$ : \textbf{$(U,U+V)-$Source Distortion Decoder}}
	\label{alg:2algsgn}
	\begin{algorithmic}[1]
		\Repeat
		\State $\ev_{V} \leftarrow D(\Hm_{V} ,\sv_{2})$ \label{ins:evv}
		\State $p \Unif [0,1]$
		\Until $|\ev_{V}| \leq w$, $w-|\ev_{V}|  \equiv 0 \pmod{2}$ and $p \leq x_{|\ev_{V}|}$
		\State $\ev_{U} \leftarrow D_{(w-|\ev_{V}| )/2}(\Hm_{U},\sv_{1},\ev_{V})$ 
				\State \Return $(\ev_{U},\ev_{U}+\ev_{V})$
	\end{algorithmic}
\end{algorithm}

From now on we consider two random variables : $\ev$ which is the output of Algorithm 
\ref{alg:2algsgn} and $\ev'$ which is a uniformly distributed error of weight $w$. It is easily verified that $w_{1}(\ev) = |\ev_{V}|$ and $w_{2}(\ev) = (w-|\ev_{V}|)/2$.
Moreover, it turns out that it is not only necessary in order to achieve uniform distribution on the
output to enforce that $w_1(\ev)$ follows the same law as $w_1(\ev')$, this is also sufficient.
To check this, let us introduce some additional notation. For
$i \in \{1,2\}$ we define the quantities
$$
p_{i}^{sdd}(j) \eqdef \mathbb{P}_{\ev} \left( w_{i}(\ev) = j\right) \quad ; \quad 
p_{i}^{u}(j) \eqdef \mathbb{P}_{\ev'} \left( w_{i}(\ev')
= j \right)
$$
We will also say that a source distortion decoder $D$ {\em behaves uniformly} for a
parity-check matrix $\Hm$ if
$\mathbb{P}_{\sv,\theta} \left( \ev = D(\Hm,\sv)
\right)$ only depends on the weight $|\ev|$ (here $\theta$ denotes the internal randomness of algorithm $D$).

In such a case, the
no-rejection vector $\xv$ can be chosen so that the output of
Algorithm \ref{alg:2algsgn} is uniformly distributed as shown by the
following theorem.

\begin{restatable}{theorem}{propostatDec}
	\label{th:statDec}
	If the source decoder $D$ used in Algorithm \ref{alg:2algsgn}
	behaves uniformly for $\Hm_V$ and uniformly for $\Hm_{U}''$ which is obtained from $(\Hm_{U},\ev_{V})$ in Proposition \ref{prop:syndPuncAlt} (see \S\ref{subsec:2.2}) for
	all error patterns $\ev_V$ obtained as
	$\ev_V = D(\Hm_{V} ,\sv_{2})$, we have:
	\begin{displaymath}
	\rho\left( \mathcal{D}_{w}, \mathcal{U}_{w} \right) = \rho \left(
	p_1^{sdd},p_1^{u} \right)
	\end{displaymath}
	where $\cD_w$ is the output distribution of Algorithm \ref{alg:2algsgn}. Then,  output of Algorithm \ref{alg:2algsgn} is the uniform
	distribution over $S_w$ if in addition two executions of $D$
	are independent and the no-rejection probability vector $\xv$ is
	chosen for any $i$ in $\{0,\dots,w\}$ as
	\begin{displaymath}
	x_i = \frac{1}{\Mrs} \frac{p_1^u(i)}{p(i)} \text{ if $w \equiv i \pmod {2}$}\;\;x_i=0 \;\text{ otherwise}
	\end{displaymath}
	with $p(i) \eqdef \prob_{\sv,\theta}(|D(\Hm_{V},\sv)|=i)$ and $\Mrs \eqdef \mathop{\sup}\limits_{\substack{0 \leq i \leq w \\ i \equiv w \pmod {2} }}\frac{p_1^u(i)}{p(i)}$. 
\end{restatable}

\subsection{Application to the Prange source distortion decoder.} 
The Prange source decoder (defined in \S\ref{subsec:2.2}) is extremely
close to behave uniformly for almost all linear codes. To keep this
paper within a reasonable length we just provide here how the relevant
distribution $p(i)$ is computed.
\begin{proposition}
	[Weight Distribution of the Prange Algorithm]
	\label{prop:weightDistribPrange}~\\
	Let $p(i) = \sum_{\ev : |\ev|=i} \mathbb{P}_{\sv,\theta} \left( \ev = \Psd(\Hm,\sv) \right)$.
	For all $w,k,n \in \mathbb{N}$ with $k \leq n$, $w \leq n-k$, all parity-check matrices of size $(n-k)\times n$,
		we have $p(w) = \frac{\binom{n-k}{w}}{2^{n-k}}$.
\end{proposition}
By using Theorem \ref{th:statDec} with this distribution $p$ we can set up the no-rejection probability vector $\xv$ in Algorithm
\ref{alg:2algsgn}. To have an efficient algorithm it is essential that the parameter $\Mrs$ is as small as possible (it is readily verified 
that the average number of calls in Algorithm
\ref{alg:2algsgn} to $\Psd(\Hm_{V},\sv_{2})$
is $\Mrs$). Let $\ev$ be an error of weight $w$ chosen uniformly at random. This average number of calls can be chosen to be small by
imposing that the distributions of $w_1(\ev)$ and $|D(\Hm_{V},\sv_2)|$ to have the same expectation. 
The expectation of $w_1(\ev)$ is approximately $w\left(1 - \frac{w}{n}\right)$ and the expectation of $|D(\Hm_{V},\sv_2)|$ is $(n/2-k_V)/2$. We choose therefore 
$k_V$ such that
\begin{displaymath}
(n/2-k_V)/2 \approx w\left(1 - \frac{w}{n}\right).
\end{displaymath}
Thanks to this property, $k_{V}$ is chosen to ``align'' both distributions and in this way $M_{rs}$ is small. This rejection sampling method comes at the price of slightly increasing the weight the decoder can output as it is shown in Figure \ref{fig:distOptR}.
It is easy to see that the optimal choice of the parameters $k_U,k_V,w$ minimizing $\Mrs$ for given $n$, $R \eqdef k/n$
leads to the following choice: 
$$
w = \lfloor  n  \frac{3-\sqrt{1+8R}}{4} \rfloor,\quad k_{U} = n/2-w,\quad  
k_V = \lfloor  n/2 - 2w\left( 1 - \frac{w}{n} \right) \rfloor.
$$

For instance for  $n= 2000$, $k=1000$, we have $w=382$, $k_U=618$, $k_V=382$ and
$\Mrs \approx 2.54$. Recall that the relative weight of an error $\ev \in \F_{2}^{n}$ is defined as $\frac{|\ev|}{n}$. 
Figure \ref{fig:distOptR} gives the relative error weight as a function of $R$ of Algorithm $UV$-$\mathbf{sddV}2$ (with rejection sampling), Algorithm $UV$-$\mathbf{sddV}1$ (without rejection sampling), the relative weight which is achieved by a generic decoder $\frac{1-R}{2}$ and the relative Gilbert-Varshamov bound $h^{-1}(1-R)$. For instance with $R = 0.5$ we have $w = \lfloor 0.1909n \rfloor$. 
\begin{figure}[htbp]
	\begin{center} 
		\caption{Comparison of the Optimal Signature Distortion with or without the Rejection Sampling Method, the Gilbert-Varshamov Bound and the Generic Distortion}
		\includegraphics[scale=0.4]{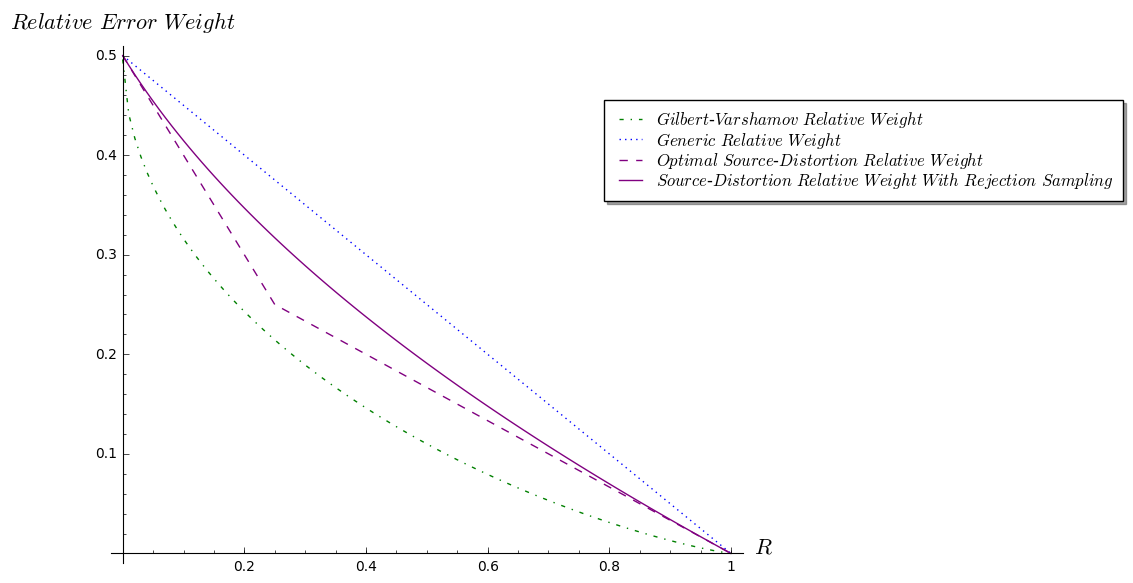}
		\label{fig:distOptR}
	\end{center}
\end{figure}

 \newcommand{\WF}{\mathrm{WF}}

\section{Best Known Algorithms for Solving the DOOM Decoding Problem } 
\label{sec:attsourcedistortion}
We consider here the best known techniques for solving Problem
\ref{prob:DOOM}.
\begin{restate}{problem}{prob:DOOM}[\textup{DOOM} -- Decoding One Out of Many]~\\
  \begin{tabular}{ll}
    Instance: & $\quad\Hm \in \mathbb{F}_{2}^{(n-k) \times n}$,
                $\sv_{1},\cdots,\sv_{q} \in \mathbb{F}_{2}^{n-k}$, $w$
                integer \\
    Output: & $\quad(\ev,i) \in \mathbb{F}_{2}^{n} \times \llbracket
              1,q \rrbracket$ such that $|\ev| = w$ and $\Hm\ev^{T} =
              \sv_{i}^{T}$.
  \end{tabular}
\end{restate}
When $q=1$, Problem~\ref{prob:DOOM} is known as the Syndrome Decoding
(SD) problem. Information Set Decoding (ISD) is the best known
technique to solve SD, it can be traced back to Prange \cite{P62}.  It
has been improved in \cite{S88,D91} by introducing a birthday
paradox. The current state-of-the-art can be found in
\cite{MMT11,BJMM12,MO15}. The DOOM problem was first considered in
\cite{JJ02} then analyzed in \cite{S11} for Dumer's variant of
ISD.

Existing literature usually assumes that there is a unique solution to
the problem. This is true when $w$ is smaller than the
Gilbert-Varshamov bound (see Definition~\ref{dfn:gv}).  When $w$ is
larger, as it is the case here, we speak of {\em source-distortion
  decoding}, the number of solutions grows as $M=\binom{n}{w}/2^{n-k}$
and the cost analysis must be adapted. Considering multiple instances,
as in DOOM above, also alters the cost analysis.

From this point and till the end of this section, the parameters
$n,k,w$ are fixed.

\subsection{Why Does DOOM Strengthen the Security Proof?}
An attacker may produce many, say $q$, favorable messages and hash
them to obtain $\sv_1,\ldots{},\sv_q$ submitted to a solver of
Problem~\ref{prob:DOOM} together with the public key $\Hm$ and the
signature weight $w$. The output of the solver will produce a valid
signature for one of the $q$ messages.  In the security reduction, the
assumption related to DOOM is precisely the same, that is assuming key
indistinguishability and a proper distribution of the signatures, the
adversary has to solve an instance of DOOM as described above and the
reduction is tight in this respect.

The usual Full Domain Hash (FDH) proof for existential forgery would
use SD rather than DOOM and to guaranty a security parameter
$\lambda$, the cost of SD, denoted $\WF$, has to be at least
$q2^\lambda$ where $q\le2^\lambda$ is the number of hash queries.
This would require code parameters $(n,k,w)$ such that
$\WF\ge2^{2\lambda}$. Instead we only require the cost of DOOM to be
at least $2^\lambda$, and even though DOOM is easier than SD, this
will provide a tighter bound and allow smaller parameters.

We denote by $\WF^{1-\delta}$ the workfactor of DOOM when $q$ can be as
large as allowed. It is shown in \cite{S11} that solving DOOM with ISD
with $q$ instances cannot cost less than $\WF/\sqrt{q}$, corresponding
to $\delta=0.33$ and a choice of parameters such that
$\WF\ge2^{1.5\lambda}$. In practice, the situation is more
favorable. When decoding codes of rate $k/n=1/2$ at the
Gilbert-Varshamov bound ($w=0.11\,n$), for Dumer's variant of ISD we
get $\delta\le0.25$ and $\WF\ge2^{1.32\lambda}$. When $w$ grows, the
situation is even better, for a rate $1/2$ and $w=0.19\,n$ (the
signature parameters) we get $\delta\approx0.07$ and
$\WF\ge2^{1.08\lambda}$.

Finally, this means that using state-of-the-art solutions for DOOM, we
only need to increase the code size by $8\%$ compared with SD's
requirement, whereas the usual proof would require to double the
parameters. The rest of this section is devoted to a detailed analysis
leading to this conclusion.

\subsection{ISD -- Information Set Decoding}
The ISD algorithm for solving DOOM is sketched in Algorithm
\ref{algo:isdoom}.
\begin{algorithm}
  \caption{(generalized) ISD}\label{algo:isdoom}
  \begin{algorithmic}[1]
    \State {\bf input:} $\Hm \in \F_{2}^{(n-k)\times n},
    \sv_1,\ldots{},\sv_q \in \F_{2}^{n-k}, w \mbox{ integer}$ 
    \Loop
    \State pick an $n\times n$ permutation matrix $\Pm$
    \State perform partial Gaussian elimination on $\Hm\Pm$
    \begin{displaymath}
      \setlength{\unitlength}{4mm}
      \begin{picture}(23,3)(-4,0)
        \put(-1.1,1){\makebox(0,1)[r]{$\Um\Hm\Pm=$}}
        \put(-0.6,0){\framebox(2.6,1){0}}
        \put(-0.6,1){\framebox(2.6,2){$\Imat_{n-k-\ell}$}}
        \put(2,0){\framebox(3,1){$\Hm'$}}
        \put(2,1){\framebox(3,2){$\Hm''$}}
        \put(5.5,0){\vector(0,1){1}}
        \put(5.5,1){\vector(0,1){2}}
        \put(5.5,0){\vector(0,-1){0}}
        \put(5.5,1){\vector(0,-1){0}}
        \put(6,0){\makebox(0,1)[l]{$\ell$}}
        \put(6,1){\makebox(0,2)[l]{$n-k-\ell$}}
        \put(11.5,1){\makebox(0,1){$\Um\sv_i^T=$}}
        \put(13,1.4){\framebox(1.3,1.4){${\sv''_i}^T$}}
        \put(13,-0.2){\framebox(1.3,1.6){${\sv'_i}^T$}}
        \put(15,1){\makebox(0,1)[l]{$i=1,\ldots{},q$}}
      \end{picture}
    \end{displaymath}
    \State \label{subisd} compute
    $\cE=\{(\ev',i)\in\F_2^{k+\ell}\times\llbracket1,q\rrbracket\mid \Hm'\ev'^T=\sv'_i,
    \wt{\ev'}=p\}$,
    $\Hm' \in \F_{2}^{\ell\times (k+\ell)}$ \ForAll{$(\ev',i)\in\cE$}
    \State $\ev''\gets \ev' \Hm''^T+\sv''_{i}$ ;
    $\ev\gets (\ev'', \ev')\Pm^{T}$ \If{$\wt{\ev}=w$} {\bf return}
    $(\ev,i)$ \EndIf \EndFor \EndLoop
  \end{algorithmic}
\end{algorithm}
In all variants of ISD, the computation of the set $\cE$ (Instruction
\ref{subisd}) dominates the cost of one loop of
Algorithm~\ref{algo:isdoom}, we denote it by $C_q(p,\ell)$. As it is
described, the loop is repeated until a solution is found. The
standard version corresponds to a single instance, that is $q=1$.
Below we explain how the cost estimate of the algorithm varies in
various situations: when we have a single instance and a single
solution, when the number of solutions increases and when the number
of instances ($q$) increases.  For each value of $n$, $k$, $w$ and $q$, the
algorithm is optimized over the parameters $p$ and $\ell$. The optimal
values of $p$ and $\ell$ will change with the number of solutions and
the number of instances.

\subsubsection{Single Instance and Single Solution.}
We consider a situation where we wish to estimate the cost of the
algorithm for producing one specific solution of
Problem~\ref{prob:DOOM} with $q=1$. In that case, even when $w$ is
large and there are multiple solutions, the solution we are looking
for, say $\ev$, is returned if and only if the permutation $\Pm$ is
such that $\wt{\ev'}=p$ and $\wt{\ev''}=w-p$ where
$(\ev'', \ev')\gets \ev\Pm$. This will happen
with probability $\Pc(p,\ell)$ leading to the workfactor $\WF^{(1)}$
\begin{displaymath}
\Pc(p,\ell) = \frac{\binom{n-k-\ell}{w-p}\binom{k+\ell}{p}}{\binom{n}{w}}, \WF^{(1)}=\min_{p,\ell}\frac{C_1(p,\ell)}{\Pc(p,\ell)},
\end{displaymath}
which is obtained by solving an optimization problem over $p$ and
$\ell$.  The exact expression of $C_1(p,\ell)$ depends on the variant,
for instance, for Dumer's algorithm \cite{D91} we have
$C_1(p,\ell)=\max\left({\sqrt{\binom{k+\ell}{p}}}, {\binom{k+\ell}{p}}{2^{-\ell}}\right)$ up to a small polynomial factor.
For more involved variants \cite{BJMM12,MO15}, the value of
$C_1(p,\ell)$ is, for each $(p,\ell)$, the solution of another
optimization problem.

\subsubsection{Single Instance and Multiple Solutions.}
We now consider a situation where there are $M$ solutions to a
syndrome decoding problem ($q=1$). If $w$ is larger than the
Gilbert-Varshamov bound we expect $M=\binom{n}{w}/2^{n-k}$ else
$M=1$. Assuming each of the $M$ solutions can be independently
produced, the probability that one particular iteration produces (at
least) one of the solutions becomes
$\Pc_M(p,\ell)=1-(1-\Pc(p,\ell))^M$. The corresponding workfactor is
\begin{displaymath}
\WF^{(M)} =
\min_{p,\ell}\frac{C_1(p,\ell)}{\Pc_M(p,\ell)}.
\end{displaymath}
Let $(p_0,\ell_0)$ be the optimal value of the pair $(p,\ell)$ for a
single solution.

\smallskip
\noindent{\em Case 1:} $\Pc(p_0,\ell_0)\le 1/M$. We have
$\Pc_M(p_0,\ell_0)=1-(1-\Pc(p_0,\ell_0))^M\ge\mathrm{e}^{-1}
M\Pc(p_0,\ell_0)$ and thus
\begin{displaymath}
  \WF^{(M)} = \min_{p,\ell}\frac{C_1(p,\ell)}{\Pc_M(p,\ell)} \le \frac{C_1(p_0,\ell_0)}{\Pc_M(p_0
    ,\ell_0)}
  \le\frac{\mathrm{e}}{M}\frac{C_1(p_0,\ell_0)}{\Pc(p_0,\ell_0)} = \frac{\mathrm{e}}{M}\WF^{(1)}.
\end{displaymath}
Also remark that $\Pc_M(p,\ell)\le M\Pc(p,\ell)$ and thus
$\WF^{(M)}\ge\WF^{(1)}/M$.  In other words, up to a small constant
factor, the workfactor for multiple solutions is simply obtained by
dividing the single solution workfactor by the number of solutions.

\smallskip
\noindent{\em Case 2:} $\Pc(p_0,\ell_0)> 1/M$.  In this case the
success probability $\Pc_M(p_0,\ell_0)<M\Pc(p_0,\ell_0)$ and the pair
$(p,\ell)$ that minimizes the workfactor is going to be different. We
observe that the gain is much less than the factor $M$ of Case 1.

In practice, and for the parameters we consider in this work, we are
always in Case 2. In fact, for $k/n=0.5$, with Dumer's algorithm Case
1 only applies when $w/n<0.150$, while the Gilbert-Varshamov bound
corresponds to $w/n=0.110$. With BJMM's algorithm, Case 1 only happens
when $w/n \leq 0.117$. In our signature scheme we have $w/n\approx0.19$ and we
always fall in Case 2, even with a single instance.

\subsubsection{Multiples Instances with Multiple Solutions.}
We now consider the case where the adversary has access to $q$
instances for the same matrix $\Hm$ and various syndromes. This is the
Problem~\ref{prob:DOOM} that appears in the security reduction. For
each instance, we expect $M=\max\left(1,\binom{n}{w}/2^{n-k}\right)$
solutions.

As before, the cost is dominated by Instruction \ref{subisd}, we
denote it by $C_q(p,\ell)$, and the probability of success is
$\Pc_{qM}(p,\ell)=1-(1-\Pc(p,\ell))^{qM}$. The overall cost has to be
minimized over $p$ and $\ell$
\begin{displaymath}
\WF_q^{(M)} =
\min_{p,\ell}\frac{C_q(p,\ell)}{\Pc_{qM}(p,\ell)}.
\end{displaymath}
Indeed how to compute $\cE$, and thus the value of $C_q(p,\ell)$, is
not specified in Algorithm~\ref{algo:isdoom}. This is in fact what
\cite{JJ02,S11} are about. For instance with Dumer's algorithm, we
have \cite{S11}
\begin{displaymath}
C_q(p,\ell) = \max\left({\textstyle \sqrt{q\binom{k+\ell}{p}}},
\frac{q\binom{k+\ell}{p}}{2^\ell}\right),\; q\le\binom{k+\ell}{p}
\end{displaymath}
up to a small polynomial factor. Introducing multiple instances in
advanced variants of ISD has not been done so far and is an open
problem. We give in Table~\ref{tab:expo} the asymptotic exponent for
various decoding distances and for the code rate $0.5$. The third
column gives the largest useful value of $q$. It is likely that BJMM's
algorithm will have a slightly lower exponent when addressing multiple
instances. Note that for Dumer's algorithm in this range of parameters,
the improvement from $\WF^{(M)}$ (single instance) to $\WF_q^{(M)}$
(multiple instances) is relatively small, there is no reason to expect
a much different behavior for BJMM.

\begin{table}[htb]
	\centering
	\begin{tabular}{|c|c||c|c|c||c|}\cline{3-6}
		\multicolumn{2}{c|}{} & \multicolumn{3}{c||}{Dumer} & BJMM \\\hline
		$w/n$ & $\frac{1}{n}\log_2M$ & $\frac{1}{n}\log_2q$ & $\frac{1}{n}\log_2\WF_q^{(M)}$ & $\frac{1}{n}\log_2\WF^{(M)}$ & $\frac{1}{n}\log_2\WF^{(M)}$ \\\hline
		0.11 & 0.0000 & 0.0872 & 0.0872 & 0.1152 & 0.1000  \\
		0.15 & 0.1098 & 0.0448 & 0.0448 & 0.0535 & 0.0486 \\
				0.19 & 0.2015 & 0.0171 & 0.0171 & 0.0184 & 0.0175 \\\hline
	\end{tabular}
	
	\caption{Asymptotic Exponent for Algorithm~\ref{algo:isdoom} for
		$k/n=0.5$}
	\label{tab:expo}\vspace{-1cm}
\end{table}
Finally, let us mention that the best asymptotic exponent among all
known decoding techniques was proposed in \cite{MO15}. However it is penalized by a big polynomial overhead
which makes it more expensive at this point for the sizes considered
here.

\subsection{Other Decoding Techniques.}
As mentioned in \cite{CJ04,FS09}, the Generalized Birthday Algorithm
(GBA) \cite{W02} is a relevant technique to solve decoding problems,
in particular when there are multiple solutions. However, it is
competitive only when the ratio $k/n$ tends to 1, and does not apply
here. We refer the reader to \cite{MS09} for more details on GBA and
its usage.

 \section{Distinguishing a permuted $\UV$ code}\label{sec:keyAttack} 

We discuss in this section how hard it is to decide whether a given linear code is a permuted $\UV$-code or not and give the best algorithm
we have found to perform this task. This algorithm is based on a series of works on related problems
\cite{OT11,LT13,GHPT17}.

\subsection{NP-completeness}
\label{ss:NPcomplete}

	The key security of our scheme primarily relies on the problem of deciding whether a linear code is a permuted $(U,U+V)$ code or not, namely:

		\begin{restatable}{problem}{probUVNP}\textit{(}$\UV$-distinguishing\textit{)}~\\
		\label{prob:UVDist}
		\begin{tabular}{ll}
			Instance: & A binary linear code $\cC$ and an integer $k_{U}$, \\
			Question: & Is there a permutation $\pi$ of length $n$ such that $\pi(\cC)$ is a $\UV$-code \\ 
			& where $\dim(U) = k_{U}$ and $|\Sp(V)| = n/2$? 
		\end{tabular}
	\end{restatable}
Here the support of $V$ is defined as the union of the support of its codewords. More precisely the support of a vector $\xv = (x_{i})_{1 \leq i \leq n} \in \F_{2}^{n}$ is defined as:
	$$
	\Sp(\xv) \eqdef \{ i \in \llbracket 1,n\rrbracket \mbox{ } : \mbox{ } x_{i} \neq 0 \}
	$$
	and if $\cC$ denotes a code, we define its support as:
	$$
	\Sp(\cC) \eqdef \mathop{\bigcup}\limits_{\cv \in \cC} \Sp(\cv). 
	$$
 It turns out that this problem is NP-complete:
	\begin{restatable}{theorem}{theoUVNP} \label{theo:UVNP}
		The $(U,U+V)$-distinguishing problem is \textup{NP}-complete.
	\end{restatable}

	\noindent This theorem is proved in Appendix \S\ref{subsec:proofUVDist}.

	Moreover, even a weaker version of this problem still stays NP-complete. This problem is related to the output of Algorithm \ref{algo:ComputeV} that we will give in 
Subsection \ref{ss:V}.
It is an algorithm that recovers in a permuted $(U,U+V)$-code the positions that belong to the support of the $V$-code.
This allows to reorder the positions of the public code such that the first half is a permutation of the
first $n/2$ positions of the $\UV$-code whereas the second half is a permutation of the $n/2$ last positions of the $\UV$-code.
Note that we just have to reorder the second half so that it corresponds to a valid $\UV$ code (where the new $U$-code is a permutation of the old $U$-code). In other words the problem we have to solve is the following.
\begin{problem}
	Let $U$ and $V$ be two binary linear codes of length $n/2$ and let $\pi$ be a permutation of length $n/2$.
	Let
	$$
	(U,\pi(U)+V) \eqdef \{(\uv,\pi(\uv)+\vv), \uv \in U, \vv \in V\}.
	$$
	Find a permutation $\pi'$ acting on the right-hand part which gives to $(U,\pi(U)+V)$ the structure of a $\UV$-code, i.e. find $\pi'$ a permutation of length $n/2$ such that $(U,\pi'(\pi(U)+V))=(U,U+V')$ for a certain binary linear code $V'$, where
	$$(U,\pi'(\pi(U)+V)) \eqdef \{(\uv,\pi'(\pi(\uv)+\vv)), \uv \in U, \vv \in V\}.$$
\end{problem}

\begin{remark}
	$\pi'=\pi^{-1}$ is a solution to this problem but there might be other solutions of course.
\end{remark}

The decision problem which is related to this search version is the following.
\begin{problem}[Problem P2': weak $\UV$-distinguishing]
	\label{prob:P1'}
	Consider a binary linear code $\Cc$ of length $n$ where $n$ is even. Do there exist two binary linear codes
	$U$ and $V$ of length $n/2$ and a permutation $\pi$ of length $n/2$ such that
	$$
	\UV = \{(\xv,\pi(\yv)) \mbox{ } : \mbox{ } (\xv,\yv) \in \Cc, \xv \in \F_2^{n/2}, \yv \in \F_2^{n/2}\}.
	$$
\end{problem}

This problem can be viewed as Problem P2 where we have some side information available where we have been revealed the split of the support 
of the $\UV$ construction in the left and the right part, but the left part and the right part have been
permuted ``internally''.
It turns out that this decision problem is already NP-complete
\begin{restatable}{theorem}{thNPcomplete}
	\label{th:NPcomplete}
	The weak-$\UV$ distinguishing Problem P2' is an NP-complete problem.
 \end{restatable}

The proof of this theorem is also given in Appendix \S\ref{subsec:proofUVDist}.

\subsection{Main idea used in the algorithms distinguishing or recovering the structure  of a $\UV$-code we present here}
A $(U,U+V)$ code where $U$ and $V$ are random seems very close to a random linear code. There is for instance only a very slight difference between the weight distribution of a random linear code and the weight distribution of a random $(U,U+V)$-code of the same length and dimension. This slight difference happens for small and large weights and is due to codewords of the form $(\uv,\uv)$ where $\uv$ belongs to $U$ or codewords of the form
$(\mathbf{0},\vv)$ where $\vv$ belongs to $V$. More precisely, we have the following proposition
\begin{restatable}{proposition}{propdensity}
  \label{prop:density}
  Assume that we choose a $(U,U+V)$ code  by picking
  the parity-check matrices of $U$ and $V$  uniformly at random among the
  binary matrices of size $(n/2-k_U) \times n/2$ and $(n/2-k_V) \times n/2$ respectively.
  Let $a_{(U,U+V)}(w)$,  $a_{(U,U)}(w)$ and $a_{(0,V)}(w)$ be the expected number of codewords of weight 
  $w$ that are respectively in the $(U,U+V)$ code,  of the form $(\uv,\uv)$ where $\uv$ belongs to 
  $U$ and of the form $(\mathbf{0},\vv)$ where $\vv$ belongs to $V$.
  These numbers are given for even $w$ in $\{0,\dots,n\}$   by
  \begin{eqnarray*}
    a_{(U,U+V)}(w) & =& \frac{\binom{n/2}{w/2}}{2^{n/2-k_U}} + \frac{\binom{n/2}{w}}{2^{n/2-k_V}} + \frac{1}{2^{n-k_U-k_V}} \left( \binom{n}{w}-\binom{n/2}{w}-\binom{n/2}{w/2}\right)
  \end{eqnarray*}
  \begin{displaymath}
    a_{(U,U)}(w)  = \frac{\binom{n/2}{w/2}}{2^{n/2-k_U}} \quad ; \quad 
    a_{(0,V)}(w) =  \frac{\binom{n/2}{w}}{2^{n/2-k_V}}
  \end{displaymath}
  and for odd $w$ in $\{0,\dots,n\}$   by
  \begin{eqnarray*}
    a_{(U,U+V)}(w) & =&   \frac{\binom{n/2}{w}}{2^{n/2-k_V}} + \frac{1}{2^{n-k_U-k_V}} \left( \binom{n}{w}-\binom{n/2}{w}\right)
  \end{eqnarray*}
  \begin{displaymath}
    a_{(U,U)}(w)  =   0 \quad ; \quad 
    a_{(0,V)}(w)  = \frac{\binom{n/2}{w}}{2^{n/2 - k_{V}}}
  \end{displaymath}
  On the other hand, when we choose a code of length $n$ with a random parity-check matrix of size $(n-k_U-k_V)\times n$ 
  chosen uniformly at random, then the expected number $a(w)$ of codewords of weight $w>0$ is given by
  \begin{displaymath}
    a(w) = \frac{\binom{n}{w}}{2^{n-k_U-k_V}}.
  \end{displaymath}
\end{restatable}

\begin{remark}
  When the $(U,U+V)$ code is chosen in this way, its dimension is $k_U+k_V$ with probability $1-O\left(\max(2^{k_U-n/2},2^{k_V-n/2})\right)$. This also holds for the random codes of length $n$.
\end{remark}
We have plotted in Figure \ref{fig:density} the normalized logarithm of the density of codewords of the form $(\uv,\uv)$ and $(\mathbf{0},\vv)$ of relative 
{\em even} weight $x \eqdef \frac{w}{n}$ against $x$ in the case $U$ is of rate $\frac{k_U}{n/2}=0.6$ and
$V$ is of rate $\frac{k_V}{n/2}=0.4$. These two relative densities are defined respectively by
\begin{displaymath}
  \alpha_{(U,U)}(w/n) = \frac{\log_2(a_{(U,U)}(w)/a(w))}{n} \quad ; \quad 
  \alpha_{(0,V)}(w/n) =  \frac{\log_2(a_{(0,V)}(w)/a(w))}{n}
\end{displaymath}
We see that for a relative weight $w/n$ below approximately $0.18$ almost all the codewords are of the form $(\mathbf{0},\vv)$ in this case.

\begin{figure}
  \centering
  \includegraphics[scale = 0.2,height=6cm]{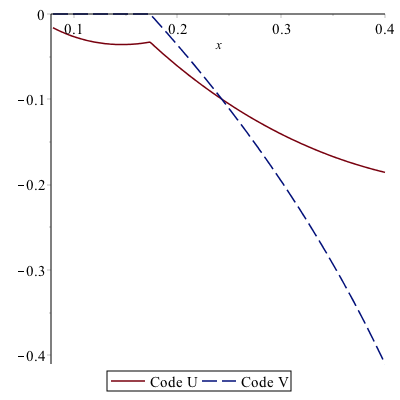}
  \caption{$\alpha_{(U,U)}(w/n)$ and $\alpha_{(0,V)}(w/n)$ against $x \eqdef \frac{w}{n}$.\label{fig:density}}
\end{figure}

Since the weight distribution is invariant by permuting the positions, this slight difference also survives
in the permuted version of $(U,U+V)$. These considerations lead to the
best attack we have found for recovering the structure of a permuted $(U,U+V)$ code.
It consists in applying known algorithms aiming at recovering low weight codewords in a linear code.
We run such an algorithm until getting  at some point either a permuted $(\uv,\uv)$ codeword where $\uv$ is in $U$ or
a permuted $(\mathbf{0},\vv)$ codeword where $\vv$ belongs to $V$.  
The rationale behind this algorithm is that the
density of codewords of the form $(\uv,\uv)$ or $(\mathbf{0},\vv)$ is bigger when the weight of the codeword gets smaller.

Once we have such a codeword we can bootstrap from there
very similarly to what has been done in \cite[Subs. 4.4]{OT11}.
Note that this attack is actually very close in spirit to the attack that was devised on the KKS signature scheme \cite{OT11}.
In essence, the attack against the KKS scheme really amounts to recover the support of the $V$ code. 
The difference with the KKS scheme is that the support of $V$ is much bigger in our case. As explained in the conclusion of \cite{OT11} the attack against the KKS scheme has in essence
an exponential complexity. This exponent becomes really prohibitive in our case when the parameters of $U$ and $V$
are chosen appropriately as we will now explain.

\subsection{Recovering the $V$ code up to a permutation}
\label{ss:V}

The aforementioned attack recovers $V$ up to some permutation of the positions. In a first step it recovers a basis of 
\begin{displaymath}
  V' \eqdef (0,V)\Pm = \{(\mathbf{0},\vv)\Pm: \vv \in V\}.
\end{displaymath}
Once this is achieved, the support $\Sp(V')$ of $V'$ can be obtained. Recall that this is the set of positions for which there exists at least one codeword
of $V'$ that is non-zero in this position. This allows to recover the code $V$ up to some permutation. 
The basic algorithm for recovering the support of $V'$ and a basis of $V'$ is given in Algorithm 
\ref{algo:ComputeV}.

\begin{algorithm}[htbp]
  \textbf{Parameters: }  (i) $\ell$ : small integer ($\ell \leqslant 40$),\\
  (ii) $p$ : very small integer (typically $1 \leqslant p
  \leqslant 10$).\\
  {\bf Input:} (i) $\Cpub$ the public code used for verifying signatures.\\
  (ii)       $N$ a certain number of iterations\\
  {\bf Output:} an independent set of elements in $V'$
  \begin{algorithmic}[1]
    \Function{ComputeV}{$\Cpub$,$N$}
    \For{$i=1,\dots,N$}
    \State $B \leftarrow \emptyset$
    \State Choose a set $I\subset \{1,\dots,n\}$  of size $n-k-\ell$ uniformly at random
    \State $\Lc \leftarrow$ \Call{Codewords}{$\punc_{I}(\Cpub),p$} \label{l:codewords}
    \ForAll{$\xv \in \Lc$}
    \State $\xv \leftarrow$ \Call{Complete}{$\xv,I,\Cpub$}
    \If{\Call{CheckV}{$\xv$}}
    \State add $\xv$ to $B$ if $\xv \notin <B>$ 
    \EndIf
    \EndFor
    \EndFor
    \State \Return $B$
    \EndFunction
  \end{algorithmic}
  \caption{\textsf{ComputeV}: algorithm that computes a set of independent elements in $V'$.} \label{algo:ComputeV}
\end{algorithm}
It uses other auxiliary functions
\begin{itemize}
\item \textsf{Codewords}$(\punc_{I}(\Cpub),p)$ which computes all (or a big fraction of) codewords of weight $p$ of the punctured public code 
  $\punc_{I}(\Cpub)$. All modern \cite{D91,FS09,MMT11,BJMM12,MO15} algorithms for decoding linear codes perform such 
  a task in their inner loop.
\item \textsf{Complete}$(\xv,I,\Cpub)$ which computes the codeword $\cv$ in $\Cpub$ such that its restriction outside $I$ is equal to $\xv$.
\item \textsf{CheckV}$(\xv)$ which checks whether $\xv$ belongs to $V'$. 
\end{itemize}

\subsubsection{Choosing $N$ appropriately.}
Let us first analyze how we have to choose $N$ such that
\textsf{ComputeV} returns $\Omega(1)$ elements. This is essentially
the analysis which can be found in \cite[Subsec 5.2]{OT11}.  This
analysis leads to
\begin{proposition}
  The probability $\Psucc$ that one iteration of the for loop (Instruction 2) in Algorithm
  \ref{algo:ComputeV} adds elements to the list $B$ is lower-bounded by 
  \begin{equation}
    \Psucc \geq \sum_{w=0}^{n/2} \frac{\binom{n/2}{w} \binom{n/2}{n-k-\ell-w}}{\binom{n}{n-k-\ell)}} f\left(  \binom{n/2-w}{p}2^{k_V+w-n/2} \right)
  \end{equation}
  where $f$ is the function 
  defined by 
  $f(x) \eqdef \max \left(x(1-x/2),1-\frac{1}{x} \right)$.
  Algorithm \ref{algo:ComputeV} returns a non zero list with probability $\Omega(1)$ when $N$ is chosen as 
  $N = \Omega\left( \frac{1}{\Psucc}\right)$.
\end{proposition}

\begin{proof}
  It will be helpful to recall \cite[Lemma 3]{OT11}

  \begin{lemma}
    \label{lem:lower_bound}
    Choose a random code $\Cr$ of length $n$ from a parity-check matrix of size $r \times n$ chosen uniformly at random
    in $\mathbb{F}_{2}^{r \times n}$.
    Let $X$ be some subset of $\mathbb{F}_{2}^n$ of size $m$.
    We have
    \begin{displaymath}
      \prob(X \cap \Cr \neq \emptyset) \geq f\left( \frac{m}{2^{r}}\right).
    \end{displaymath}
  \end{lemma}
  To lower-bound the 
  probability $\Psucc$ that an iteration is successful, we
  bring in the following random variables 
  \begin{displaymath}
    I' \eqdef I \cap \supp(I'') ~~~\mbox{and}~~~ W \eqdef \left| I'\right|
  \end{displaymath}
  where $I''$ is the set of positions that are of the images of the permutation $\Pm$ of the $n/2$ last positions. \textsf{ComputeV} outputs at least one element of $V'$ if
  there is an element of weight $p$ in $\punc_{I'}(V')$.
  Therefore the probability of success $\Psucc$ is given by
  \begin{equation}
    \label{eq:Psucc}
    \Psucc = \sum_{w=0}^{n/2}  \prob(W=w) 
    \prob\left(\exists \xv \in V' : |\xv_{\bar{I'}}|=p \mbox{ }| \mbox{ } W=w  \right)
  \end{equation}
  where $\bar{I'} \eqdef \supp(V') \setminus I'$.
  On the other hand, by using Lemma \ref{lem:lower_bound} with the
  set 
  \begin{displaymath}
    X \eqdef \left\{\xv=(x_j)_{j \in \supp(V')}: |\xv_{\bar{I'}}|=p  \right\}
  \end{displaymath}
  which is of size $\binom{n/2-w}{p} 2^{w}$, we obtain
  \begin{equation}
    \label{eq:p1}
    \prob\left(\exists \xv \in V' : |\xv_{\bar{I'}}|=p |W=w  \right)
    \geq f(x).
  \end{equation}
  with 
  \begin{displaymath}
    x \eqdef \frac{\binom{n/2-w}{p} 2^{w}}{2^{n/2-k_V}}= \binom{n/2-w}{p}2^{k_V+w-n/2}
  \end{displaymath}
  The first quantity is clearly equal to
  \begin{equation}
    \label{eq:pw1w2}
    \prob(W=w) = \frac{\binom{n/2}{w} \binom{n/2}{n-k-\ell-w}}
    {\binom{n}{n-k-\ell}}.
  \end{equation}
  Plugging in the expressions obtained in \eqref{eq:p1} and \eqref{eq:pw1w2} in \eqref{eq:Psucc} we have an explicit 
  expression of a lower bound on $\Psucc$
  \begin{equation}
    \Psucc \geq \sum_{w=0}^{n/2} \frac{\binom{n/2}{w} \binom{n/2}{n-k-\ell-w}}{\binom{n}{n-k-\ell)}} f\left(  \binom{n/2-w}{p}2^{k_V+w-n/2} \right)
  \end{equation}
  The claim on the 
  number $N$ of iterations follows directly from this. $\qed$ 
\end{proof}

\subsubsection{Complexity of recovering a permuted version of $V$.}
The complexity of a call to \textsf{ComputeV} can be estimated as follows. The complexity of computing the list of codewords of weight $p$ in a 
code of length $k+\ell$ and dimension $k$ is equal to $C_1(p,\ell)$ (this quantity is introduced in \S\ref{sec:attsourcedistortion}). It depends on the particular algorithm used here \cite{D91,FS09,MMT11,BJMM12,MO15}. This is  the complexity of the call \textsf{Codewords}$(\punc_{I}(\Cpub),p)$ in Step 
\ref{l:codewords} in Algorithm \ref{algo:ComputeV}. The complexity of  \textsf{ComputeV} and hence the complexity of recovering a permuted version of 
$V$ is clearly lower bounded by
$\Omega\left( \frac{C_1(p,\ell)}{\Psucc} \right)$. 
It turns out that the whole complexity of recovering 
a permuted version of $V$ is actually of this order, namely $ \Theta\left( \frac{C_1(p,\ell)}{\Psucc} \right)$. This can be done by a combination of two techniques
\begin{itemize}
\item Once a non-zero element of $V'$ has been identified, it is much easier to find other ones. This uses one of the tricks for breaking the KKS scheme
  (see \cite[Subs. 4.4]{OT11}). The point is the following: if we start again the procedure \textsf{ComputeV}, but this time by choosing a set $I$
  on which we puncture the code which contains the support of the codeword that we already found, then the number $N$ of iterations that we have to perform until finding a new element is negligible
  when compared to the original value of $N$. 
\item The call to \textsf{CheckV} can be implemented in such a way that the additional complexity coming from all the calls to this function is of the same order as the $N$ calls 
  to \textsf{Codewords}. The strategy to adopt depends on the values of the dimensions $k$ and $k_V$. In certain cases, it is easy to detect such codewords since they have 
  a typical weight that is significantly smaller than the other codewords. In more complicated cases, we might have to combine a technique checking first the weight of $\xv$, if it is
  above some prescribed threshold, we decide that it is not in $V'$, if it  is below the threshold, we decide that it is a suspicious candidate and use then the previous trick.
  We namely check whether  
  the support of the codeword $\xv$ can be used to find other suspicious candidates much more quickly than performing $N$ calls to \textsf{CheckV}.
\end{itemize}
To keep the length of this paper within some reasonable limit we avoid here giving the analysis of those steps and we will just use
the aforementioned lower bound on the complexity of recovering a permuted version of $V$.  

\subsection{Recovering the $U$ code up to permutation}	

We consider here the permuted code
\begin{displaymath}
  U' \eqdef (U,U)\Pm = \{(\uv,\uv)\Pm: \uv \in U\}.
\end{displaymath}
The attack in this case consists in recovering a basis of $U'$. Once this is done, it is easy to recover the $U$ code up to permutation by matching the pairs of coordinates which are equal in $U'$. The algorithm for recovering
$U'$ is the same as the algorithm for recovering $V'$.  We call the
associated function \textsf{ComputeU} though since they differ in the
choice for $N$. The analysis is slightly different indeed.

\subsubsection{Choosing $N$ appropriately.} As in the previous subsection let us analyze  how we have to choose $N$ in order that \textsf{ComputeU} returns 
$\Omega(1)$ elements of $U'$. 
We have in this case the following result.

\begin{proposition}
  The probability $\Psucc$ that one iteration of the for loop (Instruction 2) in \textsf{ComputeU} 
  adds elements to the list $B$ is lower-bounded by 
  \begin{equation}
    \Psucc \geq \sum_{w=0}^{n/2}  \frac{  \binom{n/2}{w} \binom{n/2-w}{k+\ell-2w}2^{k+\ell-2w}}
    {\binom{n}{k+\ell}} \max_{i=0}^{\lfloor p/2 \rfloor} f\left(\frac{\binom{k+\ell-2w}{p-2i} \binom{w}{i} }{2^{\max(0,k+\ell-w-k_U)}}\right)
  \end{equation}
  where $f$ is the function 
  defined by 
  $f(x) \eqdef \max \left(x(1-x/2),1-\frac{1}{x} \right)$.
  \textsf{ComputeU}  returns a non zero list with probability $\Omega(1)$ when $N$ is chosen as 
  $N = \Omega\left( \frac{1}{\Psucc}\right)$.
\end{proposition}

\begin{proof}
  Here the crucial notion is the concept of {\em matched positions}. We say that two positions $i$ and $j$ are matched if and only if
  $c_i=c_j$ for every $\cv \in U'$. There are clearly $n/2$ pairs of matched positions.
  $W$ will now be defined by the number of matched pairs that are included in $\{1,\dots,n\} \setminus I$.  
  We compute the probability of success as before by conditioning on the values taken by $W$:
  \begin{equation}
    \label{eq:Psucc2}
    \Psucc = \sum_{w=0}^{n/2}  \prob(W=w) 
    \prob\left(\left.\exists \xv \in U' : |\xv_{\bar{I}}|=p \;\right|W=w  \right)
  \end{equation}
  where $\bar{I} \eqdef \{1,\dots,n\} \setminus I$.    
  Notice that we can partition $\bar{I}$ as $\bar{I}= J_1 \cup J_2$
  where $J_2$ consists in the union of the matched pairs in $\bar{I}$. Note that $|J_2|=2w$.
  We may further partition $J_2$ as $J_2 = J_{21} \cup J_{22}$ where the elements of a matched pair are divided into the two sets.
  In other words, neither $J_{21}$ nor $J_{22}$ contains a matched pair. We are going to consider the codes
  \begin{eqnarray*}
    U" &\eqdef &\punc_{I}(U')\\
    U'''& \eqdef &\punc_{I \cup J_{22}}(U')
  \end{eqnarray*}
  The last code is of length $k+\ell-w$. The point of defining the first code is that 
  \begin{displaymath}
    \prob\left(\exists \xv \in U' : |\xv_{\bar{I}}|=p \mbox{ }| \mbox{ } W=w  \right)
  \end{displaymath} 
  is equal to the probability that $U"$ contains a codeword of weight $p$.
  The problem is that we can not apply Lemma \ref{lem:lower_bound} to it due to the matched positions it contains. 
  This is precisely the point of defining $U'''$. In this case,  we can consider that it is a random code whose  parity-check matrix is chosen uniformly at random among 
  the set of matrices of size $\max(0,k+\ell-w-k_U) \times (k+\ell-w)$. We can therefore apply Lemma \ref{lem:lower_bound} to it.
  We have to be careful about the words of weight $p$ in $U"$ though, since they do not have the same probability of occurring in $U"$ 
  due to the possible presence of matched pairs in the support. This is why we introduce for $i$ in $\{0,\dots,\lfloor p/2 \rfloor\}$ the sets 
  $X_i$ defined as follows
  \begin{displaymath}
    X_i \eqdef \{\xv=(x_i)_{i \in \bar{I} \setminus J_{22}} \mathbb{F}_{2}^{k+\ell-w}: |\xv_{J_1}|=p-2i,\mbox{ } |\xv_{J_{21}}|=i\}
  \end{displaymath}
  A codeword of weight $p$ in $U"$ corresponds to some word in one of the $X_i$'s by puncturing it in $J_{22}$. We obviously have the lower bound
  \begin{equation} 
    \prob\left\{\exists \xv \in U' : |\xv_{\bar{I}}|=p \mbox{ } | \mbox{ } W=w  \right\} \geq \mathop{\max}\limits_{i=0}^{\lfloor p/2 \rfloor} \left\{ \prob(X_i \cap U''' \neq \emptyset) \right\}
  \end{equation}
  By using Lemma \ref{lem:lower_bound} we have
  \begin{equation}
    \prob(X_i \cap U''' \neq \emptyset) \geq f\left(\frac{\binom{k+\ell-2w}{p-2i} \binom{w}{i} }{2^{\max(0,k+\ell-w-k_U)}}\right).
  \end{equation}
  On the other hand, we may notice that 
  $\prob(W=w)=\prob(w_2(\ev)=w)$ when $\ev$ is drawn uniformly at random among the binary words of weight $k+\ell$ and length $n$.
  By using Proposition \ref{prop:unifDistrib} we deduce 
  \begin{displaymath}
    \prob(W=w) = \frac{  \binom{n/2}{w} \binom{n/2-w}{k+\ell-2w}2^{k+\ell-2w}}
    {\binom{n}{k+\ell}}.
  \end{displaymath}
  These considerations lead to the following lower bound on $\Psucc$
  \begin{equation}
    \Psucc \geq \sum_{w=0}^{n/2}  \frac{  \binom{n/2}{w} \binom{n/2-w}{k+\ell-2w}2^{k+\ell-2w}}
    {\binom{n}{k+\ell}} \max_{i=0}^{\lfloor p/2 \rfloor} f\left(\frac{\binom{k+\ell-2w}{p-2i} \binom{w}{i} }{2^{\max(0,k+\ell-w-k_U)}}\right)
  \end{equation}
  $\qed$
\end{proof}

\subsubsection{Complexity of recovering a permuted version of $U$.} As for recovering the permuted $V$ code, the complexity for recovering the permuted $U$ is of order 
$\Omega\left( \frac{C_1(p,\ell)}{\Psucc} \right)$.

\subsection{Distinguishing a $(U,U+V)$ code}

It is not clear in the first case that from the single knowledge of $V'$ and a permuted version of $V$ we are able to find a permutation of the positions
which gives to the whole code the structure of a $(U,U+V)$-code. However in both cases as single successful call to 
\textsf{ComputeV} (resp. \textsf{ComputeU})  is really distinguishing the code from a random code
of the same length and dimension. In other words, we have a distinguishing attack whose complexity is given by 
$\min(O(C_U),O(C_V))$ where 
\begin{eqnarray*}
  C_U & \eqdef & \frac{C_1(p,\ell)}{\sum_{w=0}^{n/2}  \frac{  \binom{n/2}{w} \binom{n/2-w}{k+\ell-2w}2^{k+\ell-2w}}
                 {\binom{n}{k+\ell}} \max_{i=0}^{\lfloor p/2 \rfloor} f\left(\frac{\binom{k+\ell-2w}{p-2i} \binom{w}{i} }{2^{\max(0,k+\ell-w-k_U)}}\right)}\\
  C_V & \eqdef & \frac{C_1(p,\ell)}{\sum_{w=0}^{n/2} \frac{\binom{n/2}{w} \binom{n/2}{n-k-\ell-w}}{\binom{n}{n-k-\ell}} f \left(  \binom{n/2-w}{p}2^{k_V+w-n/2} \right)}
\end{eqnarray*}
and $f(x) \eqdef \max \left( x(1-x/2),1 - 1/x\right)$. As for the
decoders of \S\ref{sec:attsourcedistortion} the above numbers are
minimized (independently) over $p$ and $\ell$.

We end this section by remarking that the dual of a code $(U,U+V)$ is $(U^{\bot}+V^{\bot},V^{\bot})$ thus we have the same attack with the dual. With $k/n = 0.5$, these two attacks have the same complexity as $C_{U} = C_{V^{\bot}}$ and $C_{V} = C_{U^{\bot}}$.

 \section{Parameter Selection}
\label{sec:parameters}

In the light of the security proof in \S\ref{sec:securityProof} and
the rejection sampling method in \S\ref{sec:statDist}, we need to derive
parameters which lead to negligible success for the two following
problems:
\begin{enumerate}
\item Solve a syndrome decoding problem with $\qhash$ instances (DOOM)
  for parameters $n,k,w$.

\item Distinguish public matrices of the code family $(U,U+V)$ from
  random matrices of same size.
\end{enumerate}

In the security proof we required a salt size
$\lambda_{0}=\lambda+2\log_{2}(\qsig)$ where $\qsig$ is the number
of signature queries allowed to the adversary. Since
$\qsig\le 2^\lambda$ ($\lambda$ the security parameter) we choose a
conservative $\lambda_0=3\lambda$. We gave in
\S\ref{sec:attsourcedistortion} and \S\ref{sec:keyAttack}
state-of-the-art algorithms for the two problems mentioned above. This
served as a basis for the parameters proposed in
Table~\ref{tab:parameters}. For the key security, the estimates $C_U$
and $C_V$ are derived from the formulas at the end of
\S\ref{sec:keyAttack}. In those formulas the $C_1(p,\ell)$ term
derives from Dumer's algorithm. Using more involved techniques
\cite{MMT11,BJMM12,MO15} will reduce the key security but will leave
it above the security claims. For the message security ($\log_2 \WF$),
it is based on the DOOM variant of Dumer's algorithm, which is the
current state-of-the-art. Algorithmic improvements, like adapting DOOM
to BJMM, may lower the message security and require an adjustment of
the sizes.

\begin{table}[htb]
  \centering
  \begin{tabular}{|c||c|c|c|}
    \hline
    $\lambda$ (security) & 80 & 128 & 256 \\ 
    \hline
    $n$ & $4800$ & $7700$ & $15400$ \\
    $k=k_U+k_V$ & $2400$ & $3850$ & $7700$ \\
    $k_{V}$ & $916$ & $1470$ &$2940$ \\
    $w$ & $916$ & $1470$ & $2940$ \\
    \hline
    Signature length (bits) & $4940$ & $8084$ & $16168$
    \\
    Public key size (MBytes) & $0.720$ & $1.853$ & $7.411$ \\
    Secret key size (MBytes) & $0.347$ & $0.887$ & $3.525$ \\
    \hline
    $\log_2(\qhash\sqrt{\varepsilon})$ (Prop. \ref{prop:statDist} \S\ref{subsec:cbProb}) &  $-208$ & $-334$ & $-668$ \\ 
    \hline
    $\log_2 C_V$ (\S\ref{sec:keyAttack}) & 171 & 275 & 550 \\
    $\log_2 C_U$ (\S\ref{sec:keyAttack}) & 250 & 401 & 803 \\
    $\log_2 \WF$ (\S\ref{sec:attsourcedistortion}) & 80 & 128 & 256 \\
    \hline
  \end{tabular}
\vspace{0.3cm}
  \caption{Proposed Parameters for the $(U, U+V)$ Signature Scheme}
  \label{tab:parameters} \vspace{-1cm}
\end{table}

\subsubsection{Key Sizes.} The public key is a parity-check matrix, given in
systematic form $\Hpub=(\Imat\mid \Rm)$ which requires $k(n-k)$
bits. The secret key consists of a non-singular matrix $\Sm$, a secret
parity-check matrix $\Hsec$ and a permutation matrix $\Pm$ where
\begin{displaymath}
\Hpub = \Sm \Hsec \Pm 
\mbox{ with } \Hsec = \begin{pmatrix}
\Hm_{U} & \mathbf{0} \\
\Hm_{V} & \Hm_{V}
\end{pmatrix},
\end{displaymath}
and $\Hm_U$ and $\Hm_V$ are parity-check matrices of $U$ and
$V$. Here, we not not need to include $\Sm$ in the secret key, it is
used in the signature to compute $\Sm^{-1}\sv^T$ which can be derived
from $\Hsec$ and $\Pm$ when the public key is systematic. We have
\begin{displaymath}
\Sm^{-1}\sv^T = \Sm^{-1}(\Imat\mid\Rm)\left(\begin{array}{c}\sv^T\\\mathbf{0}\end{array}\right) =
\Sm^{-1}\Hpub\left(\begin{array}{c}\sv^T\\\mathbf{0}\end{array}\right) = \Hsec \Pm \left(\begin{array}{c}\sv^T\\\mathbf{0}\end{array}\right).
\end{displaymath}
The secret thus consists of $\Pm$, that is stored with $\log (n!)\leq n\ln(n)$ bits, and
$\Hm_U,\Hm_V$ in systematic form, that is stored with $k_U(n/2-k_U)+k_V(n/2-k_V)$ bits.
\subsubsection{Implementation.}
In Table~\ref{tab:parameters} the ratio $w/n$ is chosen close to
$0.191$ to minimize the rejection probability (see \S\ref{sec:statDist}). For the
three security levels we need to perform on average 27, 37, or 75
Gaussian eliminations to produce a signature. Most of those Gaussian
elimination are performed on parity-check matrices of shortened
codes. Finally, let us mention that the signature length ($n+3\lambda$
in the table) can be reduced (by about 30\%) by choosing a compact
representation of the sparse error vector.

 \section{Concluding remarks}
\label{sec:conclusion}

We have presented the first code-based signature scheme whose security parameter scales polynomially in key size. 
By code-based scheme, we mean here the restricted case of the Hamming metric for 
expressing the decoding problem. This setting presents the advantage that we are in the case where the decoding problem
has been thoroughly studied for many decades and where it can be considered that the complexity of the best known 
attacks has not dramatically changed since the early sixties.

\subsubsection{Comparison with TESLA-2.}
Contrarily to the overwhelming majority of lattice-based or code-based
proposals during the last decade, we avoided here the use of structured codes based on a ring structure (such
as quasi-cyclic codes). The entropy of our public key is a constant fraction of the entropy of a random matrix
of the same size. Note that the entropy loss in our case is much lower than the one observed for the aforementioned structured cases.
This is the first code-based signature scheme where such a low entropy loss
in the public key has been obtained. With a parameter selection matching tightly with the 
security reduction, the public key size stays reasonable:
less than 2 megabytes (MB) for 128 bits of classical security and less than 6 MB with the QROM reduction of \cite{CD17} for 128 bits of quantum 
security. This is strongly related to the tightness of our security reduction.
There are no other (Hamming metric) code-based to compare with our scheme.
The closest scheme we can compare with is TESLA-2 
\cite{ABBDEGKP17} which is an unstructured lattice based scheme that has a quantum security reduction in the QROM too.
The public key sizes are much bigger in their case: almost 22 MB for the same level of quantum security or
more than 11 MB for the same level of classical security.

\subsubsection{Problem P2: the $\UV$-code distinguishing problem.}
The second problem on which our security relies is indeed very clean and simple.
It consists in deciding whether there exists for a given linear code, a permutation of its positions that makes 
it a $\UV$ code.
This simplicity makes the problem really appealing in a cryptographic context. It departs from the actual trend in 
code-based or lattice-based cryptography which relies solely on the difficulty of decoding or
deciding whether or not a code/lattice contains codewords/lattice elements of 
low weight/norm.
Despite its simplicity, there are strong reasons such as its NP-completeness (see Theorem \ref{theo:UVNP}) to believe in the hardness of this problem.
Even weak versions of this problem are NP-complete. For instance, even the restricted problem of 
deciding whether there 
exists a permutation of the second half of the code positions which makes the code 
to be a $\UV$-code is NP-complete (see Theorem \ref{th:NPcomplete}).
The fact that $\UV$ codes seem to depart from random codes solely by their lowest/highest
weight codewords and the fact that the best algorithms to date for solving this problem use 
low weight codeword finding algorithms might indicate that there might be a tight connection of this problem 
to the problem of finding low weight codewords in a code. This is a tantalizing connection which is worth investigating. 

\section*{Acknowledgments}

We would thank Andr\'e Chailloux for a careful reading of a preliminary version of this paper \cite{DST17}, for finding an error in the security proof
and for suggesting Proposition \ref{prop:statDist} to fix this error.

\newcommand{\etalchar}[1]{$^{#1}$}

\newpage
\appendix
\section{Proofs for \S\ref{sec:securityProof}}
\subsection{List Emulation}
In the security proof, we need to build lists of indices (salts) in
$\F_2^{\lambda_0}$. Those lists have size $\qsig$, the maximum
number of signature queries allowed to the adversary, a number which
is possibly very large. For each message $\mv $ which is either hashed
or signed in the game we need to be able to
\begin{itemize}
\item create a list $\listM$ of $\qsig$ random elements of
  $\F_2^{\lambda_0}$, when calling the constructor {\tt new list()};
\item pick an element in $\listM$, using the method
  $\listM.\mathtt{next}()$, this element can be picked only once;
\item decide whether or not a given salt $\rv$ is in $\listM$, when calling  $\listM.\mathtt{contains}(\rv)$.
\end{itemize}

The straightforward manner to achieve this is to draw $\qsig$ random
numbers when the list is constructed, this has to be done once for
each different message $\mv $ used in the game. This may result in a
quadratic cost $\qhash\qsig$ just to build the lists. Once the
lists are constructed, and assuming they are stored in a proper data
structure (a heap for instance) picking an element or testing
membership has a cost at most $O(\log \qsig)$, that is at most
linear in the security parameter $\lambda$.
\begin{figure}[h!]
\centering
  \begin{tabular}{|l|l|}
    \hline
    \underline{class list} & \underline{method list.contains$(\rv)$} \\
    \quad elt, index & \quad return $\rv\in\{\mathtt{elt}[i],1\le i\le \qsig\}$ \\
    \quad list$()$ & \\\cline{2-2}
    \qquad $\mathtt{index}\gets0$ & \underline{method list.next$()$}  \\
    \qquad for $i=1,\ldots{},\qsig$ & \quad $\mathtt{index}\gets\mathtt{index}+1$ \\
    \qquad\quad $\mathtt{elt}[i]\gets\mathtt{randint}(2^{\lambda_0})$ & \quad return $\mathtt{elt[index]}$ \\
    \hline
  \end{tabular}
\caption{Standard implementation of the list operations.\label{fig:standard}}
\end{figure}

Note that in our game we condition on the event that {\em all elements of $\listM$ are different}.
This implies that now  $\listM$ is obtained by choosing among the  subsets of size $\qsig$ of $\F_2^{\lambda_0}$ uniformly at random.
We wish to emulate the list operations and never construct them explicitly such that the probabilistic model for 
$\listM\mathtt{.next()}$ and $\listM\mathtt{.contains}(\rv)$ stays the same as above (but again conditioned on the event that  
all elements of $\listM$ are different).
For this purpose, we want to ensure that at any time we call either $\listM\mathtt{.contains}(\rv)$ or $\listM\mathtt{.next()}$ we have 
\begin{eqnarray}
\label{eq:list.contains}
\prob(\listM.\mathtt{contains}(\rv)=\mathtt{true})  &= & \prob(\rv \in \listM | \Qc)\\
\prob(\rv = \listM.\mathtt{next()}) & = & p(\rv|\Qc) \label{eq:list.next}
\end{eqnarray}
for every $\rv \in \F_2^{\lambda_0}$. Here $\Qc$ represents the queries to $\rv$ made so far and whether or not these $\rv$'s belong to
$\listM$. Queries to $\rv$ can be made through two different calls. The first one is a call of the form {\tt Sign}$(\mv)$ when it 
chooses $\rv$ during the random assignment $\rv \Unif \{ 0,1 \}^{\lambda_{0}}$. This results in a call to {\tt Hash}$(\mv,\rv)$ which queries itself
whether $\rv$ belongs to $\listM$ or not through the call $\listM\mathtt{.contains}(\rv)$. 
 The answer is necessarily positive in this case. 
The second way to query $\rv$ is by calling 
{\tt Hash}$(\mv,\rv)$ directly. In this case, both answers {\tt true} and
{\tt false} are possible.
$p(\rv|\Qc)$ represents the probability distribution of $\listM\mathtt{.next()}$ that we have in the above implementation of the list operations 
given the previous queries $\Qc$.

A convenient way to represent $\Qc$ is  through three lists $\Lcts$, $\Lcth$ and $\Lcf$.
$\Lcts$ is the list of $\rv$'s that have been queried through a call {\tt Sign}$(\mv)$.
They belong necessarily to $\listM$. $\Lcth$ is the set of $\rv$'s that have not been queried so far through a call to {\tt Sign}$(\mv)$ but have been queried through a direct call {\tt Hash}$(\mv,\rv)$ and for which $\listM\mathtt{.contains}(\rv)$ returned {\tt true}.
$\Lcf$ is the list of $\rv$'s that have been queried by a call of the form {\tt Hash}$(\mv,\rv)$ and $\listM\mathtt{.contains}(\rv)$ returned 
{\tt false}. 
 
We clearly have
\begin{eqnarray}
\prob(\rv \in \listM | \Qc) &= & 0 \;\;\text{if $\rv \in \Lcf$} \label{eq:prrvinLm1}\\
\prob(\rv \in \listM | \Qc) &= & 1 \;\;\text{if $\rv \in \Lcts \cup \Lcth$} \label{eq:prrvinLm2}\\
\prob(\rv \in \listM | \Qc) &= & \frac{\qsig - |\Lcth| - |\Lcts|}{2^{\lambda_0}-|\Lcth| - |\Lcts|- |\Lcf|} \label{eq:prrvinLm3} \;\;\text{else.}
\end{eqnarray}
To compute the probability distribution $p(\rv|\Qc)$ it is helpful to notice that
\begin{equation}
\label{eq:nextisinLcth}
\prob(\listM\mathtt{.next()}\text{ outputs an element of $\Lcth$ }) = \frac{|\Lcth|}{\qsig - |\Lcts|}.
\end{equation}
This can be used to derive $p(\rv|\Qc)$ as follows
\begin{eqnarray}
p(\rv  | \Qc) &= & 0 \;\;\text{if $\rv \in \Lcf \cup \Lcts$} \label{eq:prq1}\\
p(\rv|\Qc) &= & \frac{1}{\qsig - \Lcts} \;\;\text{if $\rv \in \Lcth$} \label{eq:prq2} \\
p(\rv|\Qc) &= & \frac{\qsig - |\Lcts|-|\Lcth| }{(\qsig - \Lcts) (2^{\lambda_0}-|\Lcth| - |\Lcts|- |\Lcf|)} \;\;\text{else.}\label{eq:prq3}
\end{eqnarray}

\eqref{eq:prq1} is obvious. \eqref{eq:prq2} follows from that all elements of $\Lcth$ have the same probability to be chosen 
as return value for $\listM\mathtt{.next()}$ and \eqref{eq:nextisinLcth}. \eqref{eq:prq3} follows by a similar reasoning by arguing
(i) that all the elements of $\F_2^{\lambda_0} \setminus \left( \Lcts \cup \Lcth \cup \Lcf \right)$ have the same probability to be chosen as
return value for  $\listM\mathtt{.next()}$, (ii) the probability that $\listM\mathtt{.next()}$ outputs an element of $\F_2^{\lambda_0} \setminus \left( \Lcts \cup \Lcth \cup \Lcf \right)$ is the probability that it does not output an element of $\Lcth$ which is 
$1-\frac{|\Lcth|}{\qsig - |\Lcts|} = \frac{\qsig - |\Lcts| - |\Lcth|}{\qsig - |\Lcts|}$.

Figure \ref{fig:implementation} explains how we perform the emulation of the 
list operations so that they perform similarly to genuine list operations as specified above. 
The idea is to create and to operate explicitly on the lists $\Lcts$, $\Lcth$ and $\Lcf$ described earlier.
We have chosen there
\begin{displaymath}
  \beta = \frac{\qsig - |\Lcth| - |\Lcts|}{2^{\lambda_0}-|\Lcth| - |\Lcts|- |\Lcf|}
  \mbox{ and } \gamma = \frac{\wt{\Lcth}}{\qsig-\wt{\Lcts}}.
\end{displaymath}
we also assume that when we call {\tt randomPop()} on a list it outputs an element of the list uniformly at random
and removes this element from it.
The method {\tt
  push} adds an element in a list. The procedure $\mathtt{rand}()$
picks a real number between 0 and 1  uniformly at random. 

\begin{figure}
\centering
  \begin{tabular}{|l|l|l|}
    \hline
    \underline{class list} & \underline{method list.contains$(\rv)$} & \underline{method list.next$()$} \\
    \quad $\Lcth$, $\Lcf$, $\Lcts$ & \quad if $\rv\not\in\Lcth \cup\Lcf \cup\Lcts $ & \quad if $\mathtt{rand}()\le\gamma$ \\
    \quad list$()$ & \qquad if rand$()\le\beta$ & \qquad $\rv\gets\Lcth.\mathtt{randomPop}()$ \\
    \qquad $\Lcth \gets\emptyset$ & \quad\qquad $\Lcth.\mathtt{push}(\rv)$ & \quad else \\
    \qquad $\Lcf \gets\emptyset$ & \qquad else & \qquad $\rv\Unif\F_2^{\lambda_0}\setminus(\Lcth \cup\Lcts  \cup\Lcf)$ \\
    \qquad $\Lcts \gets\emptyset$ & \quad\qquad $\Lcf.\mathtt{push}(\rv)$ & \quad $\Lcts.\mathtt{push}(\rv)$ \\
    & \quad return $\rv\in\Lcth \cup\Lcts$ & \quad return $\rv$ \\
    \hline
  \end{tabular}
\caption{Emulation of the list operations. 
\label{fig:implementation}}
\end{figure}

The correctness of this emulation follows directly from the calculations given above.
For instance the correctness of the call $\listM.\mathtt{next()}$ follows from the fact that with probability 
$\frac{\wt{\Lcth}}{\qsig-|\Lcts|}=\gamma$ it outputs an element of $\Lcth$ chosen uniformly at random (see \eqref{eq:nextisinLcth}).
In such a case the corresponding element has to be moved from $\Lcth$ to $\Lcts$ (since it has been queried now through a call 
to {\tt Sign}$(\mv)$). The correctness of $\listM.\mathtt{contains}(\rv)$ is a direct consequence of the formulas for 
$\prob(\rv \in \listM | \Qc)$ given in \eqref{eq:prrvinLm1}, \eqref{eq:prrvinLm2} and \eqref{eq:prrvinLm3}.
All {\tt push}, {\tt pop}, membership testing above can be implemented
in time proportional to $\lambda_0$.

\subsection{Proof of Lemma \ref{lemm:fEvent}}
\label{lemm:failEvent}

	The goal of this subsection is to estimate the probability of a collision in a signature query for a message $\mv$ when we allow at most $\qsig$ queries (the event $F$ in the security proof) and to deduce Lemma \ref{lemm:fEvent} of \S\ref{sec:securityProof3}. We recall that in $\mathcal{S}_{\textup{code}}$ for each signature query, we  pick $\rv$ uniformly at random in $\{0,1\}^{\lambda_{0}}$. Then the probability we are looking for is bounded by the probability to pick the same $\rv$ at least twice after $\qsig$ draws. The following lemma will be useful.

	\begin{lemma} The probability to have at least one collision
          after drawing uniformly and independently $t$ elements in a
          set of size $n$ is upper bounded by ${t^{2}}/{n}$ for
          sufficiently large $n$ and $t^2< n$.
	\end{lemma}

	\begin{proof} The probability of no collisions after drawing
          independently $t$ elements among $n$ is:
          \[ p_{n,t} \eqdef \prod_{i=0}^{t-1}\left(1-\frac{i}{n}\right)
          \ge 1 - \sum_{i=0}^{t-1}\frac{i}{n} = 1 - \frac{t(t-1)}{2n} \]
          from which we easily get $1-p_{n,t}\le t^2/n$, concluding the proof.
	\end{proof}

	In our case, the probability of the event $F$ is bounded by
        the previous probability for $t = \qsig$ and $n =
        2^{\lambda_{0}}$, so, with $\lambda_0=\lambda+2\log_2\qsig$, we can conclude that
	\begin{displaymath} 
	\mathbb{P}\left( F \right) \leq \frac{\qsig^{2}}{2^{\lambda_{0}}}= \frac{1}{2^{\lambda_{0} - 2 \log_{2}(\qsig)}}= \frac{1}{2^{\lambda}}
	\end{displaymath} 
	which concludes the proof of Lemma \ref{lemm:fEvent}.
	
\subsection{Proof of Proposition \ref{prop:statDist} and Lemma \ref{lem:distribi}}
\label{lem:distrib}

Our goal in this subsection is to prove Lemma \ref{lem:distribi} of \S\ref{sec:securityProof3} and to achieve this we will first prove Proposition \ref{prop:statDist} of \S\ref{subsec:cbProb} which asserts that syndromes by $\Hpub$ of errors of weight $w$ are indistinguishable
from random elements in $\F_2^{n-k}$:
\propoDist*

Proposition \ref{prop:statDist} is based on two lemmas. The first one is a general lemma given in \S\ref{sec:securityProof3}.
\lemleftoverHash*

\begin{proof}
Let $q_{h,f}$ be the probability distribution of the discrete random variable
$(h_0,h_0(e))$ where $h_0$ is drawn uniformly at random in $\Hc$ and $e$ drawn uniformly at random in $E$
(i.e. $q_{h,f} = \prob_{h_0,e}(h_0=h,h_0(e)=f)$).
By definition of the statistical distance we have
\begin{eqnarray}
 \esp_h\left\{\rho(\Dc(h),\Uc)\right\} &= & \sum_{h \in \Hc} \frac{1}{|\Hc|} \rho(\Dc(h),\Uc) \nonumber \\
&= & \sum_{h \in \Hc} \frac{1}{2|\Hc|} \sum_{f \in F} \left| \prob_e(h(e)=f) - \frac{1}{|F|} \right| \nonumber \\
& = & \frac{1}{2} \sum_{(h,f) \in \Hc \times F}   \left| \prob_{h_0,e}(h_0=h,h_0(e)=f) - \frac{1}{|\Hc| \cdot |F|} \right| \nonumber \\
& = & \frac{1}{2} \sum_{(h,f) \in \Hc \times F}   \left| q_{h,f} - \frac{1}{|\Hc| \cdot |F|} \right|.
\end{eqnarray}
Using the Cauchy-Schwarz inequality, we obtain 
\begin{equation} \label{eq:cauchy-schwarz}
\sum_{(h,f) \in \Hc \times F}   \left| q_{h,f} - \frac{1}{|\Hc| \cdot |F|} \right| \leq  \sqrt{\sum_{(h,f) \in \Hc \times F}
\left( q_{h,f} - \frac{1}{|\Hc| \cdot |F|}\right)^2} \cdot \sqrt{|\Hc| \cdot |F|}.
\end{equation}
Let us observe now that
	\begin{align}
	\sum_{(h,f) \in \Hc \times F}
\left( q_{h,f} - \frac{1}{|\Hc| \cdot |F|}\right)^2 & =  \sum_{h,f}
 \left( q_{h,f}^2 - 2 \frac{q_{h,f}}{|\Hc| \cdot |F|}+ \frac{1}{|\Hc|^2 \cdot |F|^2}\right) \nonumber \\
& = \sum_{h,f} q_{h,f}^2 -2 \frac {\sum_{h,f} q_{h,f}}{|\Hc| \cdot |F|} + \frac{1}{|\Hc| \cdot |F|} \nonumber \\
& = \sum_{h,f} q_{h,f}^2  - \frac{1}{|\Hc| \cdot |F|}  \label{eq:q}.
	\end{align}
Consider for $i \in \{0,1\}$ independent random variables $h_i$ and $e_i$ 
that are  drawn uniformly at random in $\Hc$ and $E$ respectively.
We continue this computation by noticing now that
\begin{align}
\sum_{h,f} q_{h,f}^2 & =  \sum_{h,f} \prob_{h_0,e_0} (h_0 = h,h_0(e_0)=f) \prob_{h_1,e_1} (h_1 = h,h_1(e_1)=f)   \nonumber \\
& = \prob_{h_0,h_1,e_0,e_1} \left(h_0=h_1,h_0(e_0)=h_1(e_1)\right) \nonumber  \\
& = \frac{\prob_{h_0,e_0,e_1}\left(h_0(e_0)=h_0(e_1)\right) }{|\Hc|}\nonumber \\
& = \frac{1+\varepsilon }{|\Hc| \cdot |F|}.\label{eq:collision}
\end{align}
By substituting for $\sum_{h,f} q_{h,f}^2$ the expression obtained in \eqref{eq:collision} into
\eqref{eq:q} and then back into \eqref{eq:cauchy-schwarz} we finally obtain
$$
\sum_{(h,f) \in \Hc \times F}   \left| q_{h,f} - \frac{1}{|\Hc| \cdot |F|} \right|\leq  \sqrt{\frac{1+\varepsilon}{|\Hc| \cdot |F|}
- \frac{1}{|\Hc| \cdot |F|}}\sqrt{|\Hc| \cdot |F|} = \sqrt{\frac{\varepsilon}{|\Hc| \cdot |F|}}\sqrt{|\Hc| \cdot |F|} = \sqrt{\varepsilon}.
$$
This finishes the proof of our lemma.	$\qed$
\end{proof}
In order to use this lemma to bound the statistical distance we are interested in, we perform now the following computation
\begin{lemma}\label{lem:syndromeDistribution}
Assume that $\xv$ and $\yv$ are random vectors of $S_w$ that are drawn uniformly at random in this set. We have
	$$
	\mathbb{P}_{\Hpub,\xv,\yv}
\left( \Hpub \xv^{T} = \Hpub \yv^{T} \right)  \leq \frac{1}{2^{n-k}} (1 + \varepsilon) \mbox{ with } \varepsilon \mbox{ given in Proposition \ref{prop:statDist}.}  $$
\end{lemma}

\begin{proof}
Recall that $\Hpub$ is obtained as
	\[ \Hm_{\text{pub}} = \Sm \Hsec \Pm   \;\;\;\text{  with  } \;\;\;
	\Hsec \eqdef 
	\begin{pmatrix}
	\Hm_{U} & \mathbf{0} \\
	\Hm_{V} & \Hm_{V}
	\end{pmatrix}
 \]
where $\Hm_{U}$ has been chosen uniformly at random in $\F_2^{(n/2 - k_{U}) \times n/2}$, $\Hm_V$ has been chosen uniformly in 
$\F_{2}^{(n/2- k_{V}) \times n/2}$, $\Sm$ has been chosen uniformly at random among the invertible matrices in $\F_{2}^{(n-k)\times (n-k)}$
and $\Pm$ among the $n \times n$ permutation matrices.
As $\Sm$ is non-singular and $\Pm$ is a permutation, the probability of  the event $\Hpub \xv^{T} = \Hpub\yv^{T}$ is the same as the probability of the event
	\begin{align*}
	\begin{pmatrix}
	\Hm_{U} & \mathbf{0} \\
	\Hm_{V} & \Hm_{V}
	\end{pmatrix}  \xv^{T} =  
	\begin{pmatrix}
	\Hm_{U} & \mathbf{0} \\
	\Hm_{V} & \Hm_{V}
	\end{pmatrix} \yv^{T}. 
	\end{align*} 
	Let $\xv$ be a vector of $\mathbb{F}_{2}^{n}$, we will denote in the following by $\xv_{1}$ (resp. $\xv_{2}$) the vector formed by its first (resp. last) $n/2$ coordinates. 
In other words, the probability we are looking for is

	\begin{align*}
	\mathbb{P}&_{\Hm_U,\Hm_V,\xv,\yv} \big(  \Hm_{U}  (\xv_{1} + \yv_{1})^{T} = \mathbf{0} \wedge \Hm_{V} (\xv_{1} + \xv_{2} + \yv_{1} +  \yv_{2})^{T} = \mathbf{0} \big).
\end{align*}

To compute this probability we use Lemma \ref{lem:inrc} which says that:
	\begin{equation} 
	\label{eq:prob} 
	\mathbb{P}_{\Hm }\left( \Hm \ev^{T} = \mathbf{0}\right) = \frac{1}{2^{n-k}} \mbox{ if } \ev \neq 0 \mbox{ and } 1 \mbox{ otherwise}
	\end{equation} 
when $\Hm$ is chosen uniformly at random in $\F_2^{(n-k)\times n}$. This lemma motivates to distinguish between four disjoint events	
	\newline

	{\bf Event 1:}
	$$
	\mathcal{E}_{1} \eqdef \{  \xv_{1} + \yv_{1} = \mathbf{0} \wedge \xv_{1} + \xv_{2} + \yv_{1} + \yv_{2} \neq \mathbf{0}  \} 
	$$

	{\bf Event 2:}
	$$
	\mathcal{E}_{2} \eqdef \{  \xv_{1} + \yv_{1} \neq \mathbf{0} \wedge \xv_{1} + \xv_{2} + \yv_{1} + \yv_{2} = \mathbf{0}  \}  
	$$

	{\bf Event 3:}
	$$
	\mathcal{E}_{3} \eqdef \{  \xv_{1} + \yv_{1} \neq \mathbf{0} \wedge \xv_{1} + \xv_{2} + \yv_{1} + \yv_{2} \neq \mathbf{0}  \}  
	$$

	{\bf Event 4:} 
	$$
	\mathcal{E}_{4} \eqdef \{  \xv_{1} + \yv_{1} = \mathbf{0} \wedge \xv_{1} + \xv_{2} + \yv_{1} + \yv_{2} = \mathbf{0}  \} 
	$$

	Under these events we get thanks to \eqref{eq:prob}:
	\begin{align}
		&\mathbb{P}_{\Hsec,\xv,\yv} \left( \Hsec \xv^{T} = \Hsec \yv^{T} \right) \nonumber\\
	&= \sum_{i=1}^{4} \mathbb{P}_{\Hsec} \left(  \Hsec \xv^{T} = \Hsec \yv^{T} | \mathcal{E}_{i} \right) \mathbb{P}_{\xv,\yv} \left( \mathcal{E}_{i} \right) \nonumber\\ 
	&= \frac{\mathbb{P}_{\xv,\yv} \left( \mathcal{E}_{1} \right)}{2^{n/2 - k_{V}}}  + \frac{\mathbb{P}_{\xv,\yv } \left( \mathcal{E}_{2} \right)}{2^{n/2 - k_{U}}}  + \frac{\mathbb{P}_{\xv,\yv} \left( \mathcal{E}_{3} \right)}{2^{n - k}}  + \mathbb{P}_{\xv,\yv} \left( \mathcal{E}_{4} \right) \nonumber \\
	&= \frac{1}{2^{n-k}} \left(  \frac{\mathbb{P} \left( \mathcal{E}_{1} \right)}{2^{n/2 - k_{V} - n + k}}  + \frac{\mathbb{P} \left( \mathcal{E}_{2} \right)}{2^{n/2 - k_{U}-n+k}}  + \mathbb{P} \left( \mathcal{E}_{3} \right)  + 2^{n-k} \mathbb{P} \left( \mathcal{E}_{4} \right) \right)  \nonumber \\ 
	&\leq \frac{1}{2^{n-k}} \left( 1 +  2^{n/2-k_{U}}\mathbb{P} \left( \mathcal{E}_{1} \right) + 2^{n/2-k_{V}} \mathbb{P} \left( \mathcal{E}_{2} \right) + 2^{n-k} \mathbb{P}(\mathcal{E}_{4})  \right), \label{eq:sum4terms}
	\end{align}	
where we used for the last inequality the trivial upper-bound $	\mathbb{P} \left( \mathcal{E}_{3} \right) \leq 1$.
Let us now upper-bound (or compute) the probabilities of the events $\Ec_1$, $\Ec_2$ and $\Ec_4$.
For $\Ec_4$ we clearly have
	$$
	\mathbb{P}_{\xv,\yv } \left( \mathcal{E}_{4} \right) = \prob (\xv = \yv) = \frac{1}{\binom{n}{w}}.
	$$	
For $\Ec_1$ we derive the following upper-bound	
	\begin{align} 
	\mathbb{P}_{\xv,\yv} \left( \mathcal{E}_{1} \right) &\leq \mathbb{P}\left( \xv_{1} =\yv_{1} \right) \nonumber \\ 
	&= \sum_{w_{1} = 0}^{w} \frac{\binom{n/2}{w_{1}} \binom{n/2}{w-w_{1}}^{2} }{\binom{n}{w}^{2}} \nonumber \\
	&\leq \sum_{w_{1} = 0}^{w} \frac{\binom{n/2}{w_{1}} \binom{n/2}{w-w_{1}} }{\binom{n}{w}^{2}} \binom{n/2}{w/2} \label{eq:upper_bound_binomial}\\
	& = \frac{\binom{n/2}{w/2}}{\binom{n}{w}}                                  \label{eq:simplification} 
	\end{align} 
where \eqref{eq:upper_bound_binomial} follows from $ \binom{n/2}{w-w_{1}}^{2} \leq \binom{n/2}{w-w_{1}} \binom{n/2}{w/2}$
for all $w_1$ in $\{0,\dots,w\}$ and \eqref{eq:simplification} from \\  $\sum_{w_{1} = 0}^{w} \binom{n/2}{w_{1}} \binom{n/2}{w-w_{1}} = \binom{n}{w}$.
To upper-bound $\mathbb{P}\left( \mathcal{E}_{2} \right)$, let us first derive the distribution of 
$\xv_1+\xv_2$. 
We first observe that 
\begin{align}
\prob(\xv_1 + \xv_2 =\ev) & = \prob\Big( \xv_1 + \xv_2=\ev \Big| \;|\xv_1+\xv_2|=w_e \Big) \prob(|\xv_1+\xv_2|=w_e)\nonumber\\
&= \frac{1}{\binom{n/2}{w_e}}  2^{w_e} \frac{\binom{n/2}{(w-w_e)/2} \binom{n/2-(w-w_e)/2}{w_e}}{\binom{n}{w}} \;\text{(by Prop.  \ref{prop:unifDistrib})}
\label{eq:E2}
\end{align}
if $w_e \equiv w \pmod{2}$, where $w_e$ is the Hamming weight of $\ev$. 
If $w_e$ does not have the same parity as $w$, then this probability is equal to $0$.
From this we deduce that
	\begin{align*}
\mathbb{P}_{\xv,\yv}\left( \mathcal{E}_{2} \right)	& \leq   \mathbb{P} \left( \xv_{1} + \xv_{2} = \yv_{1} + \yv_{2} \right) \\
	&= \sum_{\substack { j \in \{0,\dots,w\} \\ j \equiv w \pmod{2} }} \sum_{\ev \in \mathbb{F}_{2}^{n/2} : |\ev| = j} \mathbb{P}_{\xv} \left( \xv_{1} + \xv_{2} = \ev \right)^{2} \\
& = \sum_{\substack { j \in \{0,\dots,w\} \\ j \equiv w \pmod{2} }} \frac{1}{\binom{n/2}{j}} 2^{2j} \frac{\binom{n/2}{(w-j)/2}^2 \binom{n/2-(w-j)/2}{j}^2}{\binom{n}{w}^2} \;\text{ (by Eq. \eqref{eq:E2})}
	\end{align*}	
		
	By plugging these upper-bounds in \eqref{eq:sum4terms}, we finally obtain:
\begin{align*}
& \mathbb{P}_{\Hpub,\xv,\yv} \left( \Hpub \xv^{T} = \Hpub \yv^{T} \right)\\
& \leq \frac{1}{2^{n-k}}\left( 1+ \frac{2^{n-k}}{\binom{n}{w}} + \frac{2^{n/2-k_U} \binom{n/2}{w/2}}{\binom{n}{w}}
+ \sum_{\substack { j \in \{0,\dots,w\} \\ j \equiv w \pmod{2} }}^w \frac{ 2^{2j+n/2-k_V}\binom{n/2}{(w-j)/2}^2 \binom{n/2-(w-j)/2}{j}}{ \binom{n/2}{j} \binom{n}{w}^2}^2 \right)
\end{align*}
	which concludes the proof. $\qed$

\end{proof}

These two lemmas imply directly  Proposition \ref{prop:statDist}.

\begin{proof}[Proposition \ref{prop:statDist}]
Indeed we let in Lemma \ref{lem:leftoverHash}, $E \eqdef \F_2^n$, $F \eqdef \F_2^{n-k}$ and $\Hc$ be the set of functions associated to 
the $4$-tuples $(\Hm_U,\Hm_V,\Sm,\Pm)$ used to generate a public parity-check matrix $\Hpub$ through \eqref{eq:def_Hpub}. These functions
$h$ are given by $h(\ev) = \Hpub \ev^T$.
Lemma \ref{lem:syndromeDistribution} gives an upper-bound for the $\varepsilon$ term in Lemma \ref{lem:leftoverHash} and this finishes the
proof of Proposition \ref{prop:statDist}. $\qed$	
	\end{proof}

We are now able to prove Lemma \ref{lem:distribi} (we use here notations of the security proof in \S\ref{sec:securityProof3}).
\lemdistribi*

\begin{proof}[Lemma \ref{lem:distribi}]
To simplify notation we let $q \eqdef \qhash$.
Then we notice that 
\begin{equation}
\label{eq:first}
\prob(S_1) \leq \prob(S_2) + \rho(\Dpubwq,\Dpub\otimes \Uc^{\otimes q}),
\end{equation}
where 
\begin{itemize}
\item $\Uc$ is the uniform distribution over $\F_2^{n-k}$;
\item $\Dpubwq$ is the distribution of the $(q+1)$-tuples
$(\Hpub,\Hpub \ev_1^T,\cdots,\Hpub \ev_q^T)$ where the $\ev_i$'s are independent and uniformly distributed in $S_w$;
\item
$\Dpub\otimes \Uc^{\otimes q}$  is the distribution of the $(q+1)$-tuples
$(\Hpub,\sv_1^T,\cdots, \sv_q^T)$ where the $\sv_i$'s are independent and uniformly distributed in $\F_2^{n-k}$.
\end{itemize}
We now observe that 
\begin{eqnarray*}
\rho(\Dpubwq,\Dpub\otimes \Uc^{\otimes q}) &= & \sum_{\Hm \in \F_2^{(n-k) \times n}} \prob(\Hpub=\Hm) \rho((\Dc_w^\Hm)^{\otimes q},\Uc^{\otimes q}) \\
& \leq & q \sum_{\Hm \in \F_2^{(n-k) \times n}} \prob(\Hpub=\Hm) \rho(\Dc_w^{\Hm},\Uc)\;\;\text{(by Prop. \ref{prop:product})} \\
& = & q \esp_{\Hpub} \left\{\rho(\Dpubw,\Uc)\right\} \\
& \leq & q \frac{\sqrt{\varepsilon}}{2} \;\;\text{(by Prop. \ref{prop:statDist})}.
\end{eqnarray*}
$\qed$
\end{proof}

\section{Proofs for  \S\ref{sec:statDist}}

\subsection{Proof of Proposition \ref{prop:sddDistrib} and Theorem \ref{th:statDec}}\label{sec:prop:uniform}
\label{subsec:distribu}

First of all it is
straightforward to check that the distributions $p_i^u$ are given by
\begin{proposition}[Distribution of $w_{1}$ and $w_{2}$]$ $
	\label{prop:unifDistrib}
	For all $i$ in $\{0,\dots,w\}$ such that $w \equiv i \pmod{2}$ 
	\begin{displaymath}
	p_{2}^{u}\left(\frac{w-i}{2}\right) = p_{1}^{u}(i) = 2^{i} \frac{\binom{n/2}{(w-i)/2} \binom{n/2-(w-i)/2}{i}}{\binom{n}{w}}
	\end{displaymath}
	and for other choices of $i$, $p_{1}(i)$ and $p_{2}(i)$ are equal to $0$.
\end{proposition} 
On the other hand the distributions $p_i^{sdd}$ of the source distortion decoder are given by
\begin{restatable}{proposition}{propsddDistrib}
	\label{prop:sddDistrib}
	Let $\theta$ denote the internal coin used in the probabilistic
	algorithm $D$ and 
	$$p(i) \eqdef \prob_{\sv,\theta}(|D(\Hm,\sv)|=i)$$ 
	If two
	executions of $D$ are independent, then for all $i$ in
	$\{0,\dots,w\}$ such that $w-i \equiv 0 \pmod{2}$ we have
	\begin{equation}
	\label{eq:p1sddi}
	p_{2}^{sdd}\left( \frac{w-i}{2} \right) = p_{1}^{sdd}(i) = \frac{x_{i} \;  p(i)}{p_{w}^{1}}
	\end{equation}
	where 
	\begin{displaymath}
	p_{w}^{1}  \eqdef \sum_{\substack{0 \leq j \leq w \\ j \equiv w
			\pmod {2} }} x_{j} \; p(j)
	\end{displaymath}
	and $p_{1}^{sdd}(i) = 0$ for other choices of $i$.
\end{restatable}

\begin{proof}
	Let $\ev$ be the output of Algorithm \ref{alg:2algsgn}. Recall that 
	$$
	p_{1}^{sdd}(j) \eqdef \mathbb{P}_{\ev} \left( w_{1}(\ev) = j \right) = \mathbb{P}_{\sv,\theta}(|D(\Hm_{V},\sv)| = j).
	$$
	As two executions of $D$ are independent, by a disjunction of independent events the probability to get an error $\ev$ such that $w_{1}(\ev) = i$ is given by:
	\begin{equation} \sum_{l=0}^{+\infty} \alpha^{l} \beta_{i}= \frac{\beta_i}{1-\alpha} \label{eq:geometric}
	\end{equation}
	where $\alpha$ denotes the probability that the output of $D$ at Instruction \ref{ins:evv} of Algorithm \ref{alg:2algsgn}
	is rejected and $\beta_{i}$ the probability to have an error of weight $i$ which is accepted. 
	These probabilities are readily seen to be equal to:	
	\[ \beta_{i} = p(i)  x_{i} \quad ; \quad \alpha = 1 - \sum_{\substack{0 \leq j \leq w \\ j \equiv w \pmod {2} }} x_{j}  p(j). \]
	Plugging this expression in \eqref{eq:geometric} finishes the proof.
	$\qed$

\end{proof}

Let us recall that $\cD_w$ is the distribution $\{ D_{\Hsec,w}(\sv) : \sv \Unif \mathbb{F}_{2}^{n-k} \}$ where $D_{\Hsec,w}$ is Algorithm \ref{alg:2algsgn}. Recall now Theorem \ref{th:statDec}
\propostatDec*

\begin{proof}
	Let us first introduce some notation. Let $\ev$ be a random variable whose distribution is $\cU_w$, \textit{i.e.} the uniform distribution over 
	$S_w$, and let $\tilde{\ev}$ be a random variable whose distribution is $\cD_w$. The last random variable can be viewed in a 
	natural way as the output of Algorithm \ref{alg:2algsgn} and is of the form $\tilde{\ev} = (\ev_U,\ev_U+\ev_V)$. We view $\ev_U$ and $\ev_V$ as random variables. We have
	\begin{equation}\label{eq:rhoDwUw}
	\rho\left( \mathcal{D}_{w}, \mathcal{U}_{w} \right) = \sum_{\ev_1,\ev_2 \in \F_2^{n/2}:|(\ev_1,\ev_2)|=w} 
	\left| \prob(\tilde{\ev}=(\ev_1,\ev_2)) - \prob(\ev=(\ev_1,\ev_2)) \right|.
	\end{equation}
	We notice now that
	\begin{eqnarray}
	\prob(\tilde{\ev}=(\ev_1,\ev_2)) & = & \prob(\ev_U=\ev_1|\ev_V=\ev_1+\ev_2)\prob(\ev_V=\ev_1+\ev_2) \nonumber\\
	& = & \prob(\ev_U=\ev_1|\ev_V=\ev_1+\ev_2)\prob_{\sv_2,\theta}(D(\Hm_{V},\sv_2)=\ev_1+\ev_2).\label{eq:evU}
	\end{eqnarray}
	From the assumption on the uniform behavior of $D$ we deduce that
	$\prob_{\sv_2,\theta}(D(\Hm_{V},\sv_{2})=\ev_1+\ev_2)$ only depends on the Hamming weight $|\ev_1+\ev_2|$ of $\ev_1+\ev_2$.
	We recall now that in Algorithm \ref{alg:2algsgn} we have
	$$ 
		\ev_{U} = D(\Hm_{U},\sv_{1},\ev_{V}) 
	$$

	Let 
	$$
		n ' \eqdef  n/2 - |\ev_{V}| \quad ; \quad
		w' \eqdef \frac{w - |\ev_{V}|}{2}
	$$
	and $\Hm_{U}'',\sv_{1}''$ are elements given by $(\Hm_{U},\sv_{1},\ev_{V})$ in Proposition \ref{prop:syndPuncAlt} in \S\ref{subsec:2.2}.
	It will now be convenient to split $\ev_U$ and $\ev_1$  into two parts: 
	the first one, denoted respectively by $\ev'_U$, and $\ev'_1$ is the restriction of these vectors
	to the complement of the support of $\ev_V$, whereas the second one, denoted respectively by 
	$\ev^{\prime\prime}_U$ and $\ev^{\prime\prime}_1$  is the restriction of these vectors to the support of $\ev_V$.
	With this notation, we now notice that
	\begin{eqnarray}
	\prob(\ev_U=\ev_1|\ev_V=\ev_1+\ev_2) & =  &
	\prob_{(\sv_1,\sv_2),\theta}(\ev'_U=\ev'_1,\ev^{\prime\prime}_U=\ev^{\prime\prime}_1|\ev_V=\ev_1+\ev_2) \nonumber\\
	& = & \prob_{\sv_1,\theta}(\ev'_U=\ev'_1)\prob_{\sv_1,\theta}(\ev^{\prime\prime}_U=\ev^{\prime\prime}_1) \nonumber\\
	& = & \prob_{\sv_1,\theta}(D_{w'}(\Hm_{U}'',\sv_{1}'')=\ev'_1)\prob_{\sv_1,\theta}(\ev^{\prime\prime}_U=\ev^{\prime\prime}_1) \nonumber\\
	& = & \frac{1}{\binom{n'}{w'}} \frac{1}{2^{n/2-n'}}.\label{eq:evUcond}
	\end{eqnarray}
	The last equality follows from the fact that 
	$D$ behaves uniformly on $\Hm_{U}''$ for all patterns $\ev_{V}$ and therefore the output of $D_{w'}(\Hm_{U}'',\sv_{1}'')$ is the uniform distribution over the
	set of words of weight $w'$ in $\F_2^{n'}$. Equality \eqref{eq:evUcond} implies that 
	$\prob(\ev_U=\ev_1|\ev_V=\ev_1+\ev_2)$ only depends on the weight of $w'$ which itself only depends on the weight 
	of $\ev_1+\ev_2$. Since $\prob_{\sv_2,\theta}(D(\Hm_{V},\sv_2)=\ev_1+\ev_2)$ has the same property, we deduce 
	from \eqref{eq:evU}, that $\prob(\tilde{\ev}=(\ev_1,\ev_2))$ only depends on the weight of $\ev_1+\ev_2$.
	Obviously $\prob(\ev=(\ev_1,\ev_2))$ also has this property. We may therefore write
	\begin{eqnarray*}
		\prob(\tilde{\ev}=(\ev_1,\ev_2)) & = &f(|\ev_1+\ev_2)|)\\
		\prob(\ev=(\ev_1,\ev_2)) & = &g(|\ev_1+\ev_2)|)
	\end{eqnarray*}
	for some functions $f$ and $g$. Plugging these expressions in \eqref{eq:rhoDwUw} yields
	by bringing in the quantity 
	$m_i$ which is the number of $\ev$ in $S_w$ such that $w_1(\ev)=i$:
	\begin{eqnarray*}
		\rho\left( \mathcal{D}_{w}, \mathcal{U}_{w} \right) &=& 
		\sum_{\substack{0 \leq i \leq w \\ i \equiv w \pmod {2} }} \sum_{\mv \in S_{w}|w_{1}(\mv)=i} 
		\left| \prob(\tilde{\ev}=\mv) - \prob(\ev=\mv) \right|\\
		&= & \sum_{\substack{0 \leq i \leq w \\ i \equiv w \pmod {2} }}
		m_i \left|f(i) - g(i) \right|\\
		&= & \sum_{\substack{0 \leq i \leq w \\ i \equiv w \pmod {2} }}
		\left|m_i(f(i) - g(i)) \right|\\
		& = & \sum_{\substack{0 \leq i \leq w \\ i \equiv w \pmod {2} }}
		\left|\prob(w_1(\tilde{\ev})=i) - \prob(w_1(\ev)=i)  \right|\\
		& = & \rho(p_1^{sdd},p_1^u).
	\end{eqnarray*}
	The last part of the proposition follows from the fact that the $p_{1}^{sdd}(i)$'s are functions of the non-rejection probability vector $\xv = (x_{i})$.
	Thanks to what we just proved, we can compute the $x_{i}$'s to have $\rho(p_1^{sdd},p_1^u)=0$.
	This will imply that the output of Algorithm \ref{alg:2algsgn} is the uniform distribution. Indeed,
	we first notice that for all $i$:
	\[0 \leq  x_{i} = \frac{1}{M_{rs}} \;\; \frac{p_{1}^{u}(i)}{p(i)} = \left( \mathop{\inf}\limits_{\substack{0 \leq j \leq w \\ w \equiv j \pmod {2} }} \frac{p(j)}{p_{1}^{u}(j)} \right) \frac{p_{1}^{u}(i)}{p(i)} \leq   \frac{p(i)}{p^{1}_{u}(i)} \;\; \frac{p_{1}^{u}(i)}{p(i)} = 1   \]
	which allows to assert that $\xv$ is a probability vector. We now use the following equations for all $i$:
	\begin{align*}
	p_{1}^{sdd}(i) &=  \frac{x_{i} \;  p(i)}{p_{w}^{1}} \\
	&= \frac{p_{1}^{u}(i)}{M_{rs} \mathop{\sum}\limits_{\substack{0 \leq j \leq w \\ w \equiv j \pmod {2} }}  \frac{1}{M_{rs}} p_{1}^{u}(j)} \\
	&= p_{1}^{u}(i) \\
	\end{align*}
	where the last line relies on the equality $\mathop{\sum}\limits_{\substack{0 \leq j \leq w \\ w \equiv j \pmod {2} }} p_{1}^{u}(j) = 1$.
	$\qed$
\end{proof}

\subsection{Proof of Proposition \ref{prop:weightDistribPrange} and discussion related to it}
\label{sec:proofPrange}

Here the internal coins are over the choices of the $n-k$ positions $I$ (columns of the parity-check matrix $\Hm$) 
we choose to invert in the Prange algorithm.
We have here
\begin{eqnarray*}
	p(w) & = & \sum_{\ev: |\ev|=w} \mathbb{P}_{\sv,I} \left( \ev = \Psd(\Hm,\sv) \right)\\
	& = & \sum_{I \subset \{1,\dots,n\}:|I|=n-k} \prob(I)\sum_{\ev: |\ev|=w}  \prob_{\sv}(\ev = \Psd(\Hm,\sv)|I)\\
	& = & \sum_{I \subset \{1,\dots,n\}:|I|=n-k} \prob(I)\sum_{\ev: |\ev|=w, \supp(\ev) \subset I} \frac{1}{2^{n-k}}\\
	& = & \frac{\binom{n-k}{w}}{2^{n-k}}.
\end{eqnarray*}
This ends the proof of Proposition \ref{prop:weightDistribPrange}. $\qed$

\section{Proofs of results of \S\ref{sec:keyAttack}}
\subsection{Proof of Proposition \ref{prop:density} in \S\ref{sec:keyAttack}}
\label{sec:proof_prop:density}
 Let us recall Proposition \ref{prop:density}
\propdensity*

We will need the following lemma.

\begin{lemma}
	\label{lem:inrc}
 Let $\yv$ be a non-zero vector of $\mathbb{F}_{2}^{n}$ and $\sv$ an arbitrary element in $\mathbb{F}_{2}^{r}$. We choose a matrix $\Hm$ of size $r \times n$ uniformly at random among the set of
 $r \times n$ binary matrices. In this case
 \[ \prob \left( \Hm \yv^{T} = \sv^{T} \right) = \frac{1}{2^{r}} \]
\end{lemma}

\begin{proof}
The coefficient of $\Hm$ at row $i$ and column $j$ is denoted by $h_{ij}$, whereas the coefficients of $\yv$ and $\sv$ are denoted 
by $y_i$ and $s_i$ respectively.
The probability we are looking for is the probability to have 
\begin{equation}
\label{eq:probabilite}
\sum_{j} h_{ij} y_j = s_i
\end{equation} for all $i$ in $\{1,\dots, r\}$. 
Since $\yv$ is non zero, it has at least one non-zero coordinate. Without loss of generality, we may assume that $y_1=1$. We may rewrite 
\eqref{eq:probabilite} as 
$h_{i1} = \sum_{j>1} h_{ij} y_j$. This event happens with probability $\frac{1}{2}$ for a given $i$ and with probability $ \frac{1}{2^{r}}$
on all $r$ events simultaneously due to the independence of the $h_{ij}$'s. 
\end{proof}

The last part of Proposition \ref{prop:density} is a direct application of this lemma.
We namely have

\begin{proposition}$ $
	\label{prop:wDistribRCode}
	Let $a(w)$ be the expected number of codewords of weight $w$ in a binary linear code $\cC$ of length $n$ whose parity-check matrix is chosen 
	$\Hm$ uniformly at random among all binary matrices of size $r \times n$. We have
	$$a(w) = \frac{\binom{n}{w}}{2^{r}}.$$
\end{proposition}

	\begin{proof} Let $Z \eqdef \sum_{ \xv \in \mathbb{F}_{2}^{n}: |\xv|=w} Z_{\xv}$ where $Z_{\xv}$ is the indicator function of the event ``$\xv$ is in $\cC$''. 
	We have
	\begin{eqnarray*}
	a(w) & = & \esp(Z) \\
	& = & \sum_{ \xv \in \mathbb{F}_{2}^n: |\xv|=w} \esp(Z_{\xv}) \\
	& = & \sum_{ \xv \in \mathbb{F}_{2}^n: |\xv|=w} \prob(\xv \in \cC)\\
	& = &  \sum_{ \xv \in \mathbb{F}_{2}^n: |\xv|=w} \prob(\Hm \xv^T = 0)\\
	& = &  \sum_{ \xv \in \mathbb{F}_{2}^n: |\xv|=w} \frac{1}{2^{r}} \\
	& = & \frac{\binom{n}{w}}{2^{r}}.
	\end{eqnarray*}
 	\end{proof}
	 This proves the part of Proposition \ref{prop:density} dealing with the expected weight distribution of a random 
	 linear code. 
	We are ready now to prove Proposition \ref{prop:density} concerning the expected weight distribution of a random 
	$(U,U+V)$ code.
	
\noindent
{\bf Weight distributions of} $(U,U) \eqdef \{(\uv,\uv): \uv \in U\}$ {\bf and} $(0,V) \eqdef \{(\mathbf{0},\vv): \vv \in V\}$.	This follows directly from Proposition \ref{prop:wDistribRCode} 
since $a_{(U,U)}(w)=0$ for odd and $a_{(U,U)}(w)$ is equal to the expected number of codewords of weight $w/2$ in 
a random linear code of length $n/2$ with a parity-check matrix of size $(n/2-k_U) \times n/2$ when $w$ is even. On the other hand $a_{(0,V)}$ is equal to the expected number of weight $w$ in a random linear code of length $n/2$ and with a 
parity-check matrix of size $(n/2-k_V)\times n/2$. In other words
\begin{eqnarray*}
a_{(U,U)}(w) & = & 0 \;\;\text{if $w$ is odd}\\
a_{(U,U)}(w) & = & \frac{\binom{n/2}{w/2}}{2^{n/2-k_U}}  \;\;\text{if $w$ is even}\\
a_{(0,V)}(w) & = & \frac{\binom{n/2}{w}}{2^{n/2-k_V}}
\end{eqnarray*}

\noindent
{\bf Weight distributions of} $(U,U+V)$.
The code $(U,U+V)$ is chosen randomly by picking up a parity-check matrix $\Hm_U$ of $U$ uniformly at random among the set of $(n/2-k_U)\times n/2$ binary 
matrices and a parity-check matrix $\Hm_V$ of $V$ uniformly at random among the set of $(n/2-k_V)\times n/2$ binary 
matrices.  Let $Z \eqdef \sum_{ \xv \in \mathbb{F}_{2}^{n}: |\xv|=w} Z_{\xv}$ where $Z_{\xv}$ is the indicator function of the event ``$\xv$ is in $(U,U+V)$''.

We have
\begin{eqnarray}
a_{(U,U+V)}(w)& =& \esp(Z) \nonumber \\
&= &\sum_{\xv \in \mathbb{F}_{2}^{n}:|\xv|=w} \esp(Z_{\xv}) \nonumber\\
&=& \sum_{\xv \in \mathbb{F}_{2}^{n}:|\xv|=w} \prob(Z_{\xv}=1) \nonumber \\
& = & \sum_{\xv \in \mathbb{F}_{2}^{n}:|\xv|=w} \prob(\xv \in (U,U+V)) \label{eq:auuvw}
\end{eqnarray}
By writing $\xv=(\xv_1,\xv_2)$ where $\xv_i$ is in $\mathbb{F}_{2}^{n/2}$ we know that $\xv$ is in $(U,U+V)$ if and only if
at the same time $\xv_1$ is in $U$ and $\xv_2 + \xv_1$ is in $V$, that is
$$\Hm_U \xv_1^T = 0,\;\;\Hm_V \xv_1^{T} = \Hm_V \xv_2^{T}.$$ 

There are three cases to consider

\noindent
{\bf Case 1:} $\xv_1=0$ and $\xv_2 \neq 0$. In this case
\begin{equation}
\label{eq:prob1}
 \prob(\xv \in (U,U+V))= \prob(\Hm_V \xv_2^{T}=\mathbf{0}) = \frac{1}{2^{n/2-k_V}}
\end{equation}

\noindent
{\bf Case 2:} $\xv_1=\xv_2$. 
In this case
\begin{equation}
\label{eq:prob2}
 \prob(\xv \in (U,U+V))= \prob(\Hm_U \xv_1^{T}=\mathbf{0}) = \frac{1}{2^{n/2-k_U}}
\end{equation}

\noindent
{\bf Case 3:} $\xv_1 \neq \xv_2$ and $\xv_1 \neq 0$.
In this case
\begin{equation}
\label{eq:prob3}
 \prob(\xv \in (U,U+V))= \prob(\Hm_U \xv_1^{T}=\mathbf{0} \wedge \Hm_V (\xv_1^{T}+\xv_2^{T})=\mathbf{0} ) = \frac{1}{2^{n/2-k_U}} \frac{1}{2^{n/2-k_V}}
\end{equation}
 Note that we used in each case Lemma \ref{lem:inrc}.
 
 By substituting $ \prob(\xv \in (U,U+V))$  in \eqref{eq:auuvw} we obtain for
 even $0< w \leq n $  
\begin{eqnarray*}
a_{(U,U+V)}(w) & =& \frac{\binom{n/2}{w/2}}{2^{n/2-k_U}} + \frac{\binom{n/2}{w}}{2^{n/2-k_V}} + \frac{1}{2^{n-k_U-k_V}} \left( \binom{n}{w}-\binom{n/2}{w}-\binom{n/2}{w/2}\right)
\end{eqnarray*}
and for odd $w \leq n$  
\begin{eqnarray*}
a(w) & =&   \frac{\binom{n/2}{w}}{2^{n/2-k_V}} + \frac{1}{2^{n-k_U-k_V}} \left( \binom{n}{w}-\binom{n/2}{w}\right)
\end{eqnarray*}
which concludes the proof. $\qed$

\subsection{Proof of Theorem \ref{prob:UVDist}}
	\label{subsec:proofUVDist}

	Recall first our problem.

	\probUVNP*

	\theoUVNP*

	\noindent We will use the generator matrix point of view to prove this theorem. Recall that a generator matrix of binary linear code of length $n$ is a matrix $\Gm \in \F_{2}^{k\times n}$ with $k \leq n$ such that:
	$$
	\cC = \{ \mv\Gm \mbox{ } : \mbox{ } \mv \in \F_{2}^{k} \}.
	$$
	In other words $\cC$ consists of all linear combinations of rows of $\Gm$ (they form a generator family of $\cC$).
	Moreover when $\Gm$ is of full rank the code $\cC$ has dimension $k$. The proof that Problem \ref{prob:UVDist} is NP-complete relies on the NP-completeness of the Three Dimensional Matching problem:
	
	\begin{problem}[\textup{3DM}]~\\
		\label{prob:3DM}
		\begin{tabular}{ll}
			Instance: & A matrix $\Gm_{\textup{3DM}} \in \F_{2}^{s\times 3t}$ where all its rows have a Hamming weight of $3$,\\
			Question: & Do there exist $t$ rows of $\Gm_{\textup{3DM}}$ which have pairwise disjoint supports? \\
		\end{tabular}
	\end{problem}

	\begin{remark} Without loss of generality we can always assume for this problem that $s \geq t + 1$ and there are no zero columns in $\Gm_{\textup{3DM}}$, otherwise verifying whether the problem has a solution is straightforward.
	\end{remark}
	Moreover, we are going to use this problem for our reduction by using some tricks taken from \cite{BGK17,W06} (by adding some identity matrices and using minimum distance arguments). 
	The following fact will be useful for our proof:

	\begin{fact} \label{lemm:spcode} Let $(\cv_{1},\cdots,\cv_{k})$ be a basis of a code $\cC$. We have
		$$
		|\Sp(\cC)| \leq \sum_{i=1}^{k}|\cv_{i}|.
		$$
		Furthermore, $|\Sp(\cC)| = \sum_{i=1}^{k} |\cv_{i}| \iff (\cv_{i})_{1\leq i \leq k} \mbox{\textup{ have pairwise disjoint supports}}$.
	\end{fact}

	\noindent In order to prove Theorem \ref{theo:UVNP} we introduce an ad-hoc problem, namely
	
	\begin{problem}[$(U,U+V)$-support distinguishing]~\\
		\label{prob:UVDistSup}
		\begin{tabular}{ll}
			Instance: & A generator matrix $\Gm \in \F_{2}^{k \times n}$, integers $k_{U}$ and $M$, \\
			Question: & Is there a permutation matrix
			$\Pm \in \F_{2}^{n\times n}$ such that $\Gm\Pm$ is a generator matrix of\\
			& a $(U,U+V)$-code where $\dim(U) = k_{U}$,  $|\Sp(U)| \geq M$ and $|\Sp(V)| = n/2$?
		\end{tabular}
	\end{problem}
	
	\noindent It is clearly an NP-problem, the following proposition gives its completeness.

	\begin{proposition} \label{prop:UVnp} Problem \ref{prob:UVDistSup} is \textup{NP}-complete.
	\end{proposition} 
	
	\noindent The following lemma will be useful for the proof. 

       \begin{restatable}{lemma}{lemmaUV} \label{lemm:distminUV} Let $U$ (\textup{resp.} $V$) be a code of minimum distance $d_{U}$ (\textup{resp.} $d_{V}$). The minimum distance $d$ of the $(U,U+V)$-code  is given by
       	$$
       	d = \min(2d_{U},d_{V}). 
         	$$
         	Moreover, codewords which achieve this minimum distance necessarily verify one of the following points: 
        	\begin{enumerate}

        		\item $(\uv,\uv)$ with $|\uv| = d_{U}$,

       		\item $(\mathbf{0},\vv)$ with $|\vv| = d_{V}$,

         		\item $(\uv,\mathbf{0})$ with $|\uv| = d_{V}$,

         		\item $(\uv,\uv+\vv)$ with $\uv \neq \mathbf{0}$, $\Sp(\uv) \subsetneq \Sp(\vv)$ and $|\vv| = d_{V}$.  
         	\end{enumerate}
         \end{restatable}

	\begin{proof} A $(U,U+V)$-code contains codewords of the form $(\uv,\uv)$ with $\uv \in U$ and $(\mathbf{0},\vv)$ with $\vv \in V$, therefore $d\leq \min(2d_{U},d_{V})$. Let $\uv \in U$ and $\vv \in V$ be such that  $(\uv,\uv+\vv)\neq\mathbf{0}$. Then if $\vv = \mathbf{0}$ we have $|(\uv,\uv+\vv)| = 2|\uv| \geq 2d_{U}$. Now if $\vv \neq \mathbf{0}$ we remark that:
		\begin{align*} 
		|(\uv,\uv+\vv)| &= |\uv| + |\uv + \vv| &\\
		&\geq |\uv| + |\vv| - |\uv| \quad (\mbox{triangle inequality})\\
		&=|\vv| \\
		&\geq d_{V}
		\end{align*} 
		then in both cases we have $d\geq \min(2d_{U},d_{V})$ which gives the first result about the minimum distance of a $(U,U+V)$-code.

		\noindent Let $(\uv,\uv+\vv)$ be a codeword of $(U,U+V)$ such that $\uv \neq \mathbf{0}$, $\vv \neq \mathbf{0}$, $\uv + \vv \neq \mathbf{0}$ and $\Sp(\uv) \not\subseteq \Sp(\vv)$. From $|\uv + \vv|=|\uv|-2|\Sp(\uv)\cap\Sp(\vv)|+|\vv|$ we deduce:
		\begin{equation}
		\label{eq:uv} 
		|(\uv,\uv+\vv)| = 2|\uv| - 2|\Sp(\uv)\cap\Sp(\vv)|  + |\vv|
		\end{equation}
		As $\Sp(\uv) \not\subseteq \Sp(\vv)$, we have $\Sp(\uv) \cap \Sp(\vv) \subsetneq \Sp(\uv)$. Therefore thanks to \eqref{eq:uv}: 
		$$
		|(\uv,\uv+\vv)| > |\vv| > 0
		$$
		which implies that $(\uv,\uv+\vv)$ cannot achieve the minimum distance as $(\mathbf{0},\vv) \in (U,U+V)$ which easily concludes the proof. 
		\qed  
	\end{proof} 
	\noindent We are now able to prove Proposition \ref{prop:UVnp}.

	\begin{proof}[Proposition \ref{prop:UVnp}] {\bf Polynomial time reduction from 3DM to the $(U,U+V)$-support distinguishing problem.} Let $\Gm_{\textup{3DM}} \in \F_{2}^{s \times 3t}$ be an instance of 3DM. Without loss of generality we can assume that it contains no zero column and that $s \geq t+1$. Let us now define for integers $p,u$:
		$$
		I_{p}(u) \eqdef \underbrace{\begin{pmatrix} I_{p} \cdots I_{p}\end{pmatrix} }_{\textrm{u times}} \in \F_{2}^{p\times up}
		$$
		and $\mathbf{0}_{p \times u}$ denotes the $\mathbf{0}$-matrix of size $p\times u$. We now build in polynomial time:
		$$
		\Gm = \left(\begin{array}{c|c}
		I_{t}(7) \ \ \mathbf{0}_{t \times 4(s-t)}  & \mathbf{0}_{t \times (4s + 3t)} \\
		\hline
		\mathbf{0}_{s\times(4s+3t)} & I_{s}(4) \ \ \Gm_{\textup{3DM}} \\  
		\end{array} \right) \in \F_{2}^{(s+t)\times 2(4s+3t)}
		$$
		and we consider the instance 
		$$
		(\Gm,t,7t)
		$$
		of the $(U,U+V)$-support distinguishing problem. 
		\newline

		\noindent{\bf YES-instance of 3DM $\implies$ YES-instance of $(U,U+V)$-support distinguishing.} Let us suppose that $\Gm_{\textup{3DM}}$ is a YES-instance of 3DM which means there exist $t$ rows which have pairwise disjoint supports. This gives the existence of a permutation $\Pm_{1}$ of size $(4s+3t)$ such that $t$ rows of $\begin{pmatrix}
		I_{s}(4) & \Gm_{\textup{3DM}} 
		\end{pmatrix}\Pm_{1}$ form the matrix $\begin{pmatrix}
		I_{t}(7) & \mathbf{0}_{t \times 4(s-t)} 
		\end{pmatrix}$. Then $\Gm\Pm$ where $\Pm$ is the permutation matrix $\Pm \eqdef \begin{pmatrix} I_{4s+3t} & \mathbf{0}_{4s+3t} \\ 
		\mathbf{0}_{4s+3t} & \Pm_{1} &  \end{pmatrix}$ which acts only on the last $4s+3t$ columns, generates a $(U,U+V)$-code where $U$ is generated by $\begin{pmatrix} I_{t}(7) & \mathbf{0}_{t \times 4(s-t)} \end{pmatrix}$ which has dimension $t$, support of size $7t$ and $V$ is generated by $\begin{pmatrix}
		I_{s}(4) & \Gm_{\textup{3DM}} 
		\end{pmatrix}\Pm_{1}$. As no column of $\Gm_{\textup{3DM}}$ is equal to $\mathbf{0}$ we have $|\Sp(V)| = 4s+3t$. 
		\newline

		\noindent{\bf YES-instance of $(U,U+V)$-support distinguishing $\implies$ YES-instance of 3DM.} Conversely, suppose there exists a permutation matrix $\Pm \in \F_{2}^{2(4s+3t)\times 2(4s+3t)}$ such that $\Gm\Pm$ generates a code $(U,U+V)$ where $\dim(U) = t$, $|\Sp(U)| \geq 7t$ and $|\Sp(V)| = 4s+3t$.

		\begin{lemma} \label{lemm:distmin} The matrix $\Gm$ generates a code of minimum distance $7$. Moreover, the codewords that achieve the minimum distance are the rows of $\Gm$. 
		\end{lemma} 
		\begin{proof} The sum of $r>1$ rows of $\begin{pmatrix}
			I_{s}(4) & \Gm_{\textup{3DM}} 
			\end{pmatrix}$ (\textit{resp.} $\begin{pmatrix}
			I_{t}(7) & \mathbf{0}_{t \times 4(s-t)}
			\end{pmatrix}$) gives a word of weight at least $4r>7$ (\textit{resp.} $7r > 7$). Moreover, all rows of $\Gm$ have weight $7$ which concludes the proof of this lemma.  \qed 
		\end{proof} 
		It directly follows that the minimum distance of the code $(U,U+V)$ is $7$ and therefore, from Lemma \ref{lemm:distminUV}, we remark that:
		$$
		7 = \min(2d_{U},d_{V}) \Rightarrow d_{V} = 7 \mbox{ and } d_{U} \geq 4
		$$
		where $d_{U}$ (\textit{resp.} $d_{V}$) is the minimum distance of $U$ (\textit{resp.} $V$). This crucial property leads to the following lemmas which summarizes the structure of the code $(U,U+V)$ that $\Gm\Pm$ generates.

		\begin{lemma}
			\label{lemm:urows}
			For all $\uv\in U$ we have:
			$$
			(\uv,\mathbf{0})\Pm^{-1} \mbox{ is a row of } \Gm \iff (\mathbf{0},\uv)\Pm^{-1} \mbox{ is a row of } \Gm 
			$$
		\end{lemma}

		\begin{proof} We know that for all $\uv \in U$, the codeword $(\uv,\uv) \in (U,U+V)$. Therefore it is clear that for all $\uv \in U$:
			$$
			(\uv,\mathbf{0}) \in (U,U+V) \iff (\mathbf{0},\uv) \in (U,U+V).
			$$
			Moreover, rows of $\Gm$ are the codewords of weight $7$ (cf Lemma \ref{lemm:distmin}). Then $(\uv,\mathbf{0})\Pm^{-1}$ is a row of $\Gm$ if and only if $(\mathbf{0},\uv)\Pm^{-1}$ is a row of $\Gm$ which concludes the proof of this lemma. \qed

		\end{proof}

		\begin{lemma}\label{lemm:intUV} There are in $\Gm$ exactly:
			\begin{itemize}
				\item $t$ rows of the form $(\uv,\mathbf{0})\Pm^{-1}$ where these $\uv$'s form a basis of the code $U$ and $|\uv|=7$,

				\item $t$ rows of the form $(\mathbf{0},\uv)\Pm^{-1}$ where $\uv \in U$ and $|\uv|=7$,

				\item $s-t$ rows of the form $(\mathbf{0},\vv)\Pm^{-1}$ where $\vv\in V$ and $|\vv| = 7$ but $\vv \notin U$. 
			\end{itemize}

		\end{lemma}

		\begin{proof} Lemmas \ref{lemm:distminUV} and \ref{lemm:distmin} imply that all rows of $\Gm$ are necessarily of the form with $\uv \in U$ and $\vv \in V$:
			\begin{enumerate}

				\item $(\uv,\uv)\Pm^{-1}$ with $2|\uv| = 7$,

				\item $(\mathbf{0},\vv)\Pm^{-1}$ with $|\vv| = 7$,

				\item $(\uv,\mathbf{0})\Pm^{-1}$ with $|\uv| = 7$,

				\item $(\uv,\uv+\vv)\Pm^{-1}$ with $\Sp(\uv) \subsetneq \Sp(\vv)$, $|\vv| = 7$ and $1 \leq |\uv| \leq 6$. 
			\end{enumerate}

			\noindent The first case $(\uv,\uv)$ is clearly impossible. We are going to show that Case 4 is impossible too. Let us denote by $\{ (\uv_{i},\uv_{i}+\vv_{i})\Pm^{-1}\}_{1\leq i \leq \alpha}$ (\textit{resp.} $\{ (\uv'_{i},\mathbf{0})\Pm^{-1} \}_{1\leq i \leq \beta}$) the rows which verify  Case 4 (\textit{resp.} 3) where $\alpha \in \llbracket 0,(s+t) \rrbracket$ (\textit{resp.} $\beta \in \llbracket 0,(s+t) \rrbracket$). We are now going to show that:
			\begin{equation} \label{eq:Ubas}
			\{ \uv_{1},\cdots,\uv_{\alpha},\uv'_{1},\cdots,\uv'_{\beta} \} \mbox{ \textit{is a basis of} } U.
			\end{equation}

			\noindent{\bf Generator Family.} As all codewords $(\uv,*)$ (an arbitrary word of the code $(U,U+V)$ for a fixed $\uv \in U$) can be generated, there is a generator family of vectors $(\uv,*)\Pm^{-1}$ in the generator matrix $\Gm$. In this way, as all rows of the form $(\uv,*)\Pm^{-1}$ with $\uv \neq \mathbf{0}$ in $\Gm$ have been considered we have the result.

			\noindent{\bf Free Family.} Let us denote by $L_{i}$ and $L'_{j}$ the rows of $\Gm$ which are defined as:
			$$
			\forall i \in \llbracket 1,\alpha \rrbracket, \mbox{ }L_{i} \eqdef (\uv_{i},\uv_{i}+\vv_{i})\Pm^{-1} \quad ; \quad \forall j \in \llbracket 1,\beta \rrbracket, \mbox{ }L'_{j} \eqdef (\uv'_{j},\mathbf{0})\Pm^{-1}.
			$$
			\noindent We remark now that $|\vv_{i}| = 7$, therefore by Lemma \ref{lemm:distmin}, codewords $(\mathbf{0},\vv_{i})\Pm^{-1}$ are rows of $\Gm$. Moreover, by Lemma \ref{lemm:urows}, codewords $(\mathbf{0},\uv'_{i})\Pm^{-1}$ are rows of $\Gm$ too. Then for each $i$ and $j$ it exists $k_{i} \neq i$ and $l_{j} \neq j$ such that
			$$
			\forall i \in \llbracket 1,\alpha \rrbracket, \mbox{ }L_{k_{i}} = (\mathbf{0},\vv_{i})\Pm^{-1} \neq L_{i} \quad ; \quad \forall j \in \llbracket 1,\beta \rrbracket, \mbox{ }L_{l_{j}} = (\mathbf{0},\uv'_{j})\Pm^{-1} \neq L_{j}
			$$
			are rows of $\Gm$. In this way let us denote by $\tilde{\Gm}$ the matrix for which we operate the following linear combination of rows of $\Gm$:
			$$
			\forall i \in \llbracket 1,\alpha \rrbracket, \mbox{ } L_{i} \leftarrow L_{i} + L_{k_{i}} \quad ; \quad \forall j \in \llbracket 1,\beta \rrbracket, \mbox{ }L_{j} \leftarrow L_{j} + L_{l_{j}}.
			$$
			In this way there are $\alpha+\beta$ rows in $\tilde{\Gm}$ of the form
			$$
			\{ (\uv_{1},\uv_{1})\Pm^{-1},\cdots,(\uv_{\alpha},\uv_{\alpha})\Pm^{-1}, (\uv'_{1},\uv'_{1})\Pm^{-1},\cdots,(\uv'_{\beta},\uv'_{\beta})\Pm^{-1}\}.
			$$
			Moreover, as $L_{k_{i}} \neq L_{i}$ and $L_{l_{j}} \neq L_{j}$ for all $i$ and $j$, the rank of $\Gm$ and $\tilde{\Gm}$ is the same.
			However the rank of $\Gm$ is $s+t>t$ and is given by its rows. It follows that the above family is free. Then codewords $\big((\uv_{i}),(\uv'_{j}\big))$ form a free family which leads to $\eqref{eq:Ubas}$ and in particular that $\alpha + \beta = \dim(U) = t$.

			\noindent Thanks to \eqref{eq:Ubas} we can apply Fact \ref{lemm:spcode} to the code $U$:
			$$
			|\Sp(U)| \leq \sum_{i=1}^{\alpha} |\uv_{i}| + \sum_{i=1}^{\beta} |\uv_{i}'| \leq \sum_{i=1}^{\alpha} 6 + \sum_{i=1}^{\beta} 7 = 6\alpha + 7\beta = 6\alpha + 7(t-\alpha) = 7t - \alpha.  
			$$
			As $|\Sp(U)| \geq 7t$ we have $\alpha = 0$ which implies that there do not exist exist rows which verify the fourth case in $\Gm$ and there are in $\Gm$ $t$ rows of the form $(\uv,\mathbf{0})\Pm^{-1}$ where codewords $\uv$ form a basis of $U$ and are of Hamming weight $7$. The $t$ rows $(\mathbf{0},\uv)\Pm^{-1}$ directly follow from Lemma \ref{lemm:urows}. All remaining rows are now of the form $(\mathbf{0},\vv)\Pm^{-1}$ with $|\vv|=7$. The case $\vv \in U$ for these words is impossible otherwise $(\vv,\mathbf{0})\Pm^{-1}$ (cf Lemma \ref{lemm:urows}) would be a row of $\Gm$ while we have considered all rows of this form and this concludes the proof. \qed
		\end{proof}

		\noindent From the above lemma, there are $t$ rows $\{ (\uv_{1},\mathbf{0})\Pm^{-1},\cdots,(\uv_{t},\mathbf{0})\Pm^{-1}\}$ in $\Gm$ where the $\uv_{i}$'s have a Hamming weight $7$ and form a basis of $U$. Therefore, 
		$$
		|\Sp(U)| \leq \sum_{i=1}^{t} |\uv_{i}| = \sum_{i=1}^{t} 7 = 7t
		$$
		and on the other hand we have $|\Sp(U)| \geq 7t$, which implies that the previous inequality is an equality. Then by the second assertion of Fact \ref{lemm:spcode}, codewords $\uv_{i}$'s have pairwise disjoint support and the $2t$ rows of $\Gm$ (cf Lemma \ref{lemm:urows}):
		$$\{(\uv_{i},\mathbf{0})\Pm^{-1},(\mathbf{0},\uv_{i})\Pm^{-1}\}_{1 \leq i \leq t}$$ 
		have pairwise disjoint support. Recall that the matrix $\Gm$ is defined as:
		$$
		\Gm = \left(\begin{array}{c|c}
		I_{t}(7) \ \ \mathbf{0}_{t \times 4(s-t)}  & \mathbf{0}_{t \times (4s + 3t)} \\
		\hline
		\mathbf{0}_{s\times(4s+3t)} & I_{s}(4) \ \ \Gm_{\textup{3DM}} \\  
		\end{array} \right) \in \F_{2}^{(s+t)\times 2(4s+3t)}
		$$
		We remark that the upper part of the matrix has $t$ rows, in this way we have at least $t$ rows of $\begin{pmatrix} \mathbf{0}_{s \times (4s + 3t)} | &I_{s}(4) & \Gm_{\textup{3DM}}\end{pmatrix}$ which have pairwise disjoint supports and this gives the existence of a matching for $\Gm_{\textup{3DM}}$ and concludes the proof of Proposition \ref{prop:UVnp}. 
		\qed

	\end{proof}

	We are now able to prove Theorem \ref{theo:UVNP} by using the NP-completeness of Problem \ref{prob:UVDist}. It firstly relies on the following lemma.

	\begin{lemma} \label{lemm:a} Let $\Gm \in \F_{2}^{k\times n}$ and integers $k_{U} \leq k$, $M \leq n/2$, we have:
		\begin{multline*} 
		\Gm \mbox{ generates a } (U,U+V) \mbox{ permuted code with } |\Sp(U)| \geq M \mbox{ and } |\Sp(V)| = n/2 \\ \iff \Gm \mbox{ generates a permuted } (U,U+V)\mbox{-code} \mbox{ with } |\Sp(V)| = n/2 \\ \mbox{ and the number of } \mathbf{0} \mbox{ columns in } \Gm \mbox{ is smaller than } n/2 -M.
		\end{multline*} 		
	\end{lemma}

	\begin{proof} Suppose that $\Gm\in\F_{2}^{k \times n}$ generates a permuted $(U,U+V)$-code. It follows that by the structure of $(U,U+V)$-codes there exists a non-singular matrix $\Sm \in \F_{2}^{k\times k}$ and a permutation matrix $\Pm \in \F_{2}^{n \times n}$ such that:
		$$
		\Sm\Gm\Pm = \begin{pmatrix}
		\Gm_{U} & \Gm_{U} \\
		\mathbf{0}_{(k-k_{U}) \times n/2} & \Gm_{V}
		\end{pmatrix} 
		$$
		where $\Gm_{U} \in \F_{2}^{k_{U} \times n/2}$ (\textit{resp.} $\Gm_{V} \in \F_{2}^{(k-k_{U})\times n/2}$) is a generator matrix of $U$ (\textit{resp.} $V$). Suppose now that $|\Sp(V)| = n/2$ which means there is no $\mathbf{0}$-column in $\Gm_{V}$.

		\noindent Let us now remark that if $|\Sp(U)| < M$, then there exists at least $n/2 - M$ columns which are equal to $\mathbf{0}$ in the matrix $\Sm\Gm\Pm$. Conversely, if there exist $n/2 - M$ columns which are equal to $\mathbf{0}$, as no column of $\Gm_{V}$ is equal to $\mathbf{0}$, we necessarily have $|\Sp(U)| < M$.

		\noindent Multiplication by $\Sm^{-1}$ and $\Pm^{-1}$ does not change the number of $\mathbf{0}$ columns and it easily follows that we have the same equivalence on $\Gm$ which concludes the proof. \qed 
	\end{proof}

	Theorem \ref{theo:UVNP} easily follows as we are going to show.

	\begin{proof} Let us consider an instance $(\Gm,k_{U},M)$ of Problem \ref{prob:UVDist}. The polynomial reduction into an instance of the $(U,U+V)$-distinguishing  is to check if it exists at  most $n/2-M$ columns of $\Gm$ which are equal to $\mathbf{0}$, to consider the code generated by $\Gm$ and the integer $k_{U}$. Then, Lemma \ref{lemm:a} is invoked to finish the proof. \qed  
	\end{proof}

\subsection{Proof of Theorem \ref{th:NPcomplete}}
	\label{subsec:proofUVDistbis} 

First it is clear that the weak $\UV$-distinguishing problem is in NP.  The proof that the problem is NP-complete 
relies on the hardness of the subcode equivalence problem \cite{BGK17}:
\begin{problem}[subcode  equivalence]~\\
	\label{prob:SE}
	\begin{tabular}{ll}
		Instance: & Two linear codes $\cC$ and $\cD$ of length $n$ \\
						Question: & Is there a permutation 
				$\sigma$ of the support such that 
				$\sigma(\cC) \subseteq \cD$
									\end{tabular}
\end{problem}
This problem was proved to be NP-complete in \cite{BGK17}.

We will show that any instance of the subcode-equivalence problem 
can be transformed into an instance of Problem \ref{prob:P1'} with the same answer.
Let us consider an instance $(\cC,\cD)$ of the subcode equivalence problem.
We will adopt the generator matrix point of view here which is more convenient for our purpose.
In other words, we have access to generator matrices $\Gm_{\cC}$ and $\Gm_{\cD}$ of codes $\cC$ and $\cD$. Let $\Gm$ be the following matrix
$$
\begin{pmatrix}
\Gm_{\cC} & \mathbf{0} \\
\mathbf{0} & \Gm_{\cD} 
\end{pmatrix}
$$

Suppose
that there
exists a permutation $\sigma$ such that $\sigma(\cC) \subseteq \cD \iff \cC \subseteq \sigma^{-1}(\cD) $. Then if we apply $\sigma^{-1}$ on the last $n/2$ columns of $\Gm$ we get:
$$
\begin{pmatrix}
\Gm_{\cC} & \mathbf{0} \\
\mathbf{0} & \sigma^{-1}(\Gm_{\cD}) 
\end{pmatrix}
$$
which generates the code $(\cC|\sigma^{-1}(\cD))$
which is equal to $(\cC|\cC + \sigma^{-1}(\cD))$. Then 
the code $(\cC,\cD)\eqdef \{(\cv,\dv):\cv \in \cC,\;\dv \in \cD\}$ is a YES instance of Problem \ref{prob:P1'}.

Conversely, suppose that $(\cC,\cD)$ 
is a YES instance of Problem \ref{prob:P1'}. 
This code has generator matrix $\Gm=\begin{pmatrix}
\Gm_{\cC} & \mathbf{0} \\
\mathbf{0} & \Gm_{\cD} 
\end{pmatrix}$ and this means that
$\Gm$ should generate
a permuted  $\UV$ code with a permutation which acts only on the second half of the code positions.
$\Gm_{\cC}$ is therefore necessarily a generator matrix of $U$.
$U$ and $\cC$ are 
therefore 
equal. 
It also follows that $\Gm_{\cD}$ (which generates $\cD$) has to generate a permutation of $U+V$.
Since $U$ is a subcode of $U+V$, 
it follows that $\cC$ is a subcode, up to a permutation, of $\cD$. \qed

 \end{document}